\documentclass[a4paper,12pt]{article}
\usepackage[utf8]{inputenc}
\usepackage{wrapfig}
\usepackage[english]{babel}

\usepackage{tikz}

\usepackage{epsfig}
\usepackage{amsthm}
\usepackage{amsmath, mathrsfs,psfrag}

\usepackage{amsbsy}

\usepackage{oldgerm}
\usepackage{color}
\usepackage{bbm}

\newcommand{\Cl}{\mathbbm{C}}
\newcommand{\Rl}{\mathbb{R}}
\newcommand{\Nl}{\mathbb{N}}

\definecolor{lightgray}{rgb}{0.8,0.8,0.8}

\newcommand{\Om}{\Omega}
\newcommand{\om}{\omega}

\newcommand{\te}{\theta}

\newcommand{\la}{\lambda}

\newcommand{\La}{\Lambda}
\newcommand{\eps}{\varepsilon}

\newcommand{\A}{\mathcal{A}}
\newcommand{\R}{\mathcal{R}}
\newcommand{\B}{\mathcal{B}}

\newcommand{\M}{\mathcal{M}}

\newcommand{\K}{\mathcal{K}}

\newcommand{\Hil}{\mathcal{H}}

\newcommand{\E}{\mathcal{E}}

\newcommand{\F}{\mathcal{F}}

\newcommand{\Ss}{\mathscr{S}}   

\newcommand{\Rti}{\tilde{R}}

\newcommand{\fti}{\tilde{f}}

\newcommand{\ghat}{\hat{g}}

\newcommand{\PG}{\mathcal{P}}

\newcommand{\frS}{\textfrak{S}}

\def\bp{\boldsymbol{p}}

\def\bx{\boldsymbol{x}}
\def\by{\boldsymbol{y}}

\def\bof{{\mbox{\boldmath{$f$}}}}

\def\bxi{{\mbox{\boldmath{$\xi$}}}}

\DeclareMathOperator{\supp}{supp}

\newcommand{\dom}{\mathrm{dom}\,}

\newcommand{\re}{\mathrm{Re}}

\DeclareMathOperator{\im}{Im}

\newcommand{\ot}{\otimes}

\newcommand{\tp}[1]{^{\otimes #1}}    

\newcommand{\zd}{z^{\dagger}}

\newcommand{\ad}{a^{\dagger}}

\usepackage{amsmath}
\usepackage{amsfonts}
\usepackage{mathrsfs}
\usepackage{graphicx}
\usepackage[T1]{fontenc}

\newtheorem{theorem}{Theorem}[section]
\newtheorem{proposition}[theorem]{Proposition}
\newtheorem{lemma}[theorem]{Lemma}

\newcounter{problem}

\theoremstyle{definition}
\newtheorem{definition}[theorem]{Definition}

\numberwithin{equation}{section}

\newlength{\dinwidth}
\newlength{\dinmargin}

\usepackage{mathrsfs}

\newcommand{\bg}{{\boldsymbol{g}}}

\newcommand{\bHil}{\boldsymbol{\mathcal H}}

\newcommand{\bpsi}{\boldsymbol{\psi}}
\newcommand{\bvarphi}{\boldsymbol{\varphi}}
\newcommand{\bPsi}{\boldsymbol{\Psi}}

\newcommand{\dS}{\text{dS}}

\def\ben{\begin{equation}}
\def\een{\end{equation}}
\def\bena{\begin{eqnarray}}
\def\eena{\end{eqnarray}}

\newcommand{\RR}{\mathbb{R}}
\newcommand{\HH}{\mathbb{H}}
\renewcommand{\SS}{\mathbb{S}}
\newcommand{\ZZ}{\mathbb{Z}}
\newcommand{\half}{\tfrac{1}{2}}

\newcommand{\dd}{{\rm d}}
\renewcommand{\H}{\mathcal{H}}

\newcommand{\bth}{{\underline{\theta}}}

\newcommand{\cO}{{\mathscr O}}

\date{May 30, 2017}

\begin{document}

\title{${SO(d,1)}$-invariant Yang-Baxter operators\newline and the dS/CFT correspondence}
\author{Stefan Hollands\thanks{Institut für Theoretische Physik, Universität Leipzig, Germany, stefan.hollands@uni-leipzig.de}\; and Gandalf Lechner\thanks{School of Mathematics, Cardiff University, UK, LechnerG@cardiff.ac.uk}}

\maketitle

\begin{abstract}
	We propose a model for the dS/CFT correspondence. The model is constructed in terms of a ``Yang-Baxter operator'' $R$ for unitary representations of the de Sitter group $SO(d,1)$. This $R$-operator is shown to satisfy the Yang-Baxter equation, unitarity, as well as certain analyticity relations, including in particular a crossing symmetry. With the aid of this operator we construct: a) A chiral (light-ray) conformal quantum field theory  whose {\em internal} degrees of freedom transform under the given unitary representation of $SO(d,1)$. By analogy with the $O(N)$ non-linear sigma model, this chiral CFT can be viewed as propagating in a de Sitter spacetime. b) A (non-unitary) Euclidean conformal quantum field theory on $\RR^{d-1}$, where $SO(d,1)$ now acts by conformal transformations in (Euclidean) {\em spacetime}. These two theories can be viewed as dual to each other if we interpret $\RR^{d-1}$ as conformal infinity of de Sitter spacetime. Our constructions use semi-local generator
  fields defined in terms 
of $R$ and abstract methods from operator algebras.
\end{abstract}

\section{Introduction}

Non-linear sigma models in 1+1 dimensions play an important role in several areas of theoretical and mathematical physics, see e.g. \cite{Ketov:2000} for a review. They become accessible to analytical methods in
particular when the target manifold is a coset manifold (maximally symmetric), such as $\mathbb{S}^{N-1} = O(N)/O(N-1)$. This model by construction has a manifest internal $O(N)$-symmetry as well as further hidden symmetries that make it integrable.
Its target space is the Riemannian manifold ${\mathbb S}^{N-1}$, and its internal symmetry group is compact. String theory is closely related to non-linear sigma models, where the target space comes into play as the spacetime in which the strings
propagate. Often, it is assumed to be of the form $\RR^{3+1} \times K_n$, where $K_n$ is a suitable compact Riemannian manifold representing the $n$ extra dimensions\footnote{Note that sigma-models are related to Nambu-Goto strings roughly speaking only after the 
reparameterization invariance of the world sheet has been taken into account in a suitable manner. 
In a quantum theory of strings, this forces $n$ to have certain well-known special values, depending whether or not one includes fermions. We will not in detail 
consider in this paper how this procedure would work in our case, so the correspondence to string theory is not a direct one.}. 

However, there is no a priori reason not to consider more general target spaces where not just the extra dimensions are curved. An example of this is the 
supersymmetric sigma-model in Anti-de Sitter spacetime arising from the quantization of the Green-Schwarz superstring in 
${\rm AdS}_5 \times \mathbb{S}^5$, which has been studied extensively in the literature, see~\cite{AF09} for a review. These studies are motivated by 
``AdS/CFT-correspondence''~\cite{Maldacena:1998,Witten:1998}, and therefore to a considerable extent aimed at the connection to gauge theories, 
see e.g.~\cite{Bei12,Bom16} for reviews.

It is also of obvious interest to consider the de Sitter spacetime ${\rm dS}_d = SO(d,1)/SO(d)$ as the target manifold, which, like the sphere or AdS space, 
is a coset manifold (maximally symmetric space). This type of non-linear sigma model would be expected to have an internal $SO(d,1)$-symmetry (which is not compact) together, perhaps, with further hidden symmetries that are in principle expected in any non-linear sigma model in a maximally symmetric space.

To exploit the hidden symmetries, say, in the $O(N)$-model, one may take advantage of the fact that its scattering matrix must be factorizing. In combination with the internal $O(N)$ symmetry, natural assumptions about the ``particle spectrum'' (basically the
representation of $O(N)$), hints from perturbation theory, and the highly constraining relations imposed by crossing symmetry, Yang-Baxter relation, analyticity, etc.,
one can often guess the form of the 2-body scattering matrix, which then consistently determines the $n$-body scattering matrix \cite{ZamolodchikovZamolodchikov:1978, AbdallaAbdallaRothe:2001}. In order to derive from such a scattering matrix quantities associated with the local operators of the theory, 
one can for example follow the bootstrap-form factor program~\cite{Smirnov:1992,BabujianFoersterKarowski:2006}. The aim of this program is to determine the matrix elements of local operators between in- and out-states (form factors), which are found using the scattering matrix and
various a priori assumptions about the form factors. The program
is largely successful but runs into technical difficulties when attempting to compute higher correlation functions in terms of the form factors, or, indeed, when trying to even show that the corresponding series converge.

\medskip

An alternative approach is to construct the operator algebras generated by the local quantum fields by abstract methods (see \cite{SchroerWiesbrock:2000-1,Lechner:2003}, and \cite{Lechner:AQFT-book:2015} for a review). The input is again the scattering matrix, but the
procedure is rather different. First, one constructs certain half-local generator fields. These ``left local fields''  $\phi(x)$ play an auxiliary role and are constructed
in such a way that with each of them, there is an associated ``right local field'' $\phi'(x')$ such that $[\phi(x), \phi'(x')]=0$ if $x$ and $x'$ are space like related points in 1+1 dimensional Minkowski space {\em and} $x'$ is to the right of $x$ in a relativistic sense. We will actually work with a similar construction for a ``chiral half'' of a massless theory on a lightray, where $[\phi(u), \phi'(u')] = 0$ if $u'> u$ with $u,u'\in\Rl$ lightray coordinates \cite{BostelmannLechnerMorsella:2011}.\footnote{There exists no meaningful scattering theory on a single lightray, and the underlying 2-body operator can here no longer be interpreted as a scattering operator. It rather serves as an algebraic datum (``R-matrix'', or ``Yang-Baxter operator'') which defines the theory.}

Both on two-dimensional Minkowski space and in the chiral lightray setting, the left and right local fields generate left and right local operator algebras, and suitable intersections of these algebras then contain the truly local fields \cite{BuchholzLechner:2004}. The latter are thereby characterized rather indirectly, and indeed, the local fields do not have a simple expression in terms of the auxiliary semi-local objects, but rather reproduce the full complexity of the form factor expansion \cite{BostelmannCadamuro:2012}. 

\medskip

In this article, we will consider such constructions for non-compact target spaces such as ${\rm dS}_d$. Since the internal symmetry group, $SO(d,1)$, is non-compact, its non-trivial unitary representations must necessarily be infinite-dimensional. This is an obvious major difference, say, to the $O(N)$-model, where the basic representation under which the single particle states transform is the fundamental representation, which is finite ($N$-) dimensional. Despite this difference, one
may proceed and ask whether the algebraic method can be generalized to non-compact groups such as $SO(d,1)$. For this, one first needs a 2-body scattering matrix (or rather, an ``$SO(d,1)$-invariant Yang-Baxter operator'', see Sect.~\ref{subsection:invariant-YB-ops}) satisfying suitable properties such as Yang-Baxter-relation, crossing symmetry, analyticity, unitarity, etc.
In turns out that the precise algebraic and analytic properties required to make the method work are related to each other in a rather intricate way, and one does not, a priori, see an obvious way to generate simple solutions to the requirements.
One result of our paper is to provide\footnote{
Some cases of our Yang-Baxter operator were previously derived in~\cite{DerkachovKorchemskyManashov:2001,DerkachovManashov:2006,DerkachovManashov:2008,ChicherinDerkachovIsaev:2013} using the powerful method of ``RLL-relations''. Apart from using a different formalism, the 
crucial difference to our work lies in the fact that we also investigate the analyticity properties, and in particular the crossing symmetry relation, which requires non-trivial adjustments. 
We also note that ~\cite{DerkachovKorchemskyManashov:2001,DerkachovManashov:2006,DerkachovManashov:2008,ChicherinDerkachovIsaev:2013} typically consider complex Lie-algebras rather than their real forms, so questions of unitarity -- resulting e.g. in the different types of series -- are not emphasized. Both crossing symmetry and unitarity play an essential role in our work.
} 
such a Yang-Baxter operator for the spin-$0$ principal, complementary and discrete series representations of $SO(d,1)$ in Sect.~\ref{section:YBOps-LorentzGroup}. This is facilitated by using an invariant geometrical description of the corresponding representations due to Bros, Epstein, 
and Moschella~\cite{BrosMoschella:1996,EpsteinMoschella:2014} (Sect.~\ref{section:UIR-LorentzGroup}). Our Yang-Baxter operator $R$ can hence be used to define left- and right half-local operator algebras for this model, 
as we show in a general framework in Sect.~\ref{section:crossing}, and more concretely in Sect.~\ref{subsection:chiral-models}. If it could be shown that suitable intersections of such algebras are sufficiently large, then this would indeed correspond to (a chiral half of) a local 1 + 1 dimensional field theory. 

Since the internal symmetry group of this model 
is by construction $SO(d,1)$, and since we believe that our $R$'s are, up to certain dressing factors, 
unique, it seems natural to expect that our theories have a relationship with quantized 
non-linear sigma models with de Sitter target space. To fix the dressing factors -- or more generally to establish a relationship with the bosonic string propagating in de Sitter -- one should complement our construction with an analysis along the lines of~\cite{AF09}\footnote{\cite{AF09} directly deals with the super-coset. 
In the de Sitter case, the corresponding super groups are given in table~\ref{tab:tableNahm} in the conclusion section.}. We leave this to a future investigation.

\medskip

The group $SO(d,1)$ is not just the isometry group of $d$-dimensional de Sitter spacetime ${\rm dS}_d$, but also the conformal isometry group of $(d-1)$-dimensional {\em Euclidean} flat space $\RR^{d-1}$. This dual role becomes geometrically
manifest if one attaches a pair of conformal boundaries ${\mathscr I}^\pm$ to ${\rm dS}_d$. Each of these boundaries is isometric to a round sphere ${\mathbb S}^{d-1}$, which in turn may be viewed as a 1-point conformal compactification of $\RR^{d-1}$ via the stereographic
projection. The induced action of $SO(d,1)$ by this chain of identifications provides the action of the conformal group on $\RR^{d-1}$.  
This well-known correspondence is at the core of the conjectured ``dS/CFT-correspondence/conjecture''~\cite{Strominger:2006,GuicaHartmannSongStrominger:2008}
which can be viewed as a cousin of the much better studied AdS/CFT-correspondence~\cite{Maldacena:1998,Witten:1998}. 
The idea behind this correspondence is that with the pair of infinities ${\mathscr I}^\pm$ there is associated a pair of Euclidean conformal field theories\footnote{
Such theories would not be expected to be reflection positive, i.e. have a unitary counterpart in Minkowski space $\RR^{d-2,1}$. One way to see this is that unitary representations of $SO(d,1)$ do not correspond to
unitary representations of $SO(d-1,2)$ via ``analytic continuation'' \cite{FrohlichOsterwalderSeiler:1983}.} acted upon by $SO(d,1)$. On the other hand, with the ``bulk'' ${\rm dS}_d$, there is associated a corresponding
``string-theory'' with internal symmetry group $SO(d,1)$. The action of the group essentially connects these two theories.

\medskip

Inspired by this circle of ideas, one might be tempted to ask  whether one can, in our setup, also naturally construct a Euclidean conformal field theory on $\RR^{d-1}$ associated with our sigma models with target space ${\rm dS}_d$.  In our approach, the core datum is a 2-body scattering matrix / Yang-Baxter operator $R$.
In Section~\ref{EuclQFT} we will outline an abstract procedure how to obtain a corresponding Euclidean conformal field theory from such an object. Thus, within our framework, there is a sense in which the essentially algebraic quantity $R$ can relate
a Euclidean conformal field theory in $(d-1)$ dimensions and a kind of ``string theory'' in $d$-dimensional de Sitter target space, 
and thereby gives a model for the dS/CFT correspondence. 

Our model of the dS/CFT correspondence may be described more concretely as follows. 
If $u=x_1+x_0$ is a lightray coordinate, the left-local chiral fields of the lightray CFT are given by 
\ben\label{field1}
\phi_R^{\rm Chir.}(u,X)
	=
	\int_{p, P} \left\{e^{iup}\,(X\cdot P)^{-\alpha-i\nu}\cdot\zd_R(\log p,P) 
	+ \dots
	\right\}
	\,,
\een
where $X$ is a point in de Sitter space and $(P \cdot X)^{-\alpha+i\nu}$ are ``de Sitter waves'' of ``momentum'' $P$ analogous to plane waves in Minkowski spacetime. 
The creation operators $z^\dagger_R(\theta, P)$ create a ``particle'' of lightray-rapidity $\theta=\log p$ and de Sitter ``momentum'' $P$ and obey a generalized Zamolodchikov-Faddeev algebra, $z_R(\theta, P) z_R(\theta', P') + R_{\theta-\theta'} z_R(\theta', P') z_R(\theta, P)=0$, 
where $R_\theta$ is our $R$-operator\footnote{See the main text for the full algebra and the concrete form of the $R$-operator.}.  
\begin{center}
  \begin{figure}[h]
		    \includegraphics[width=0.95\textwidth]{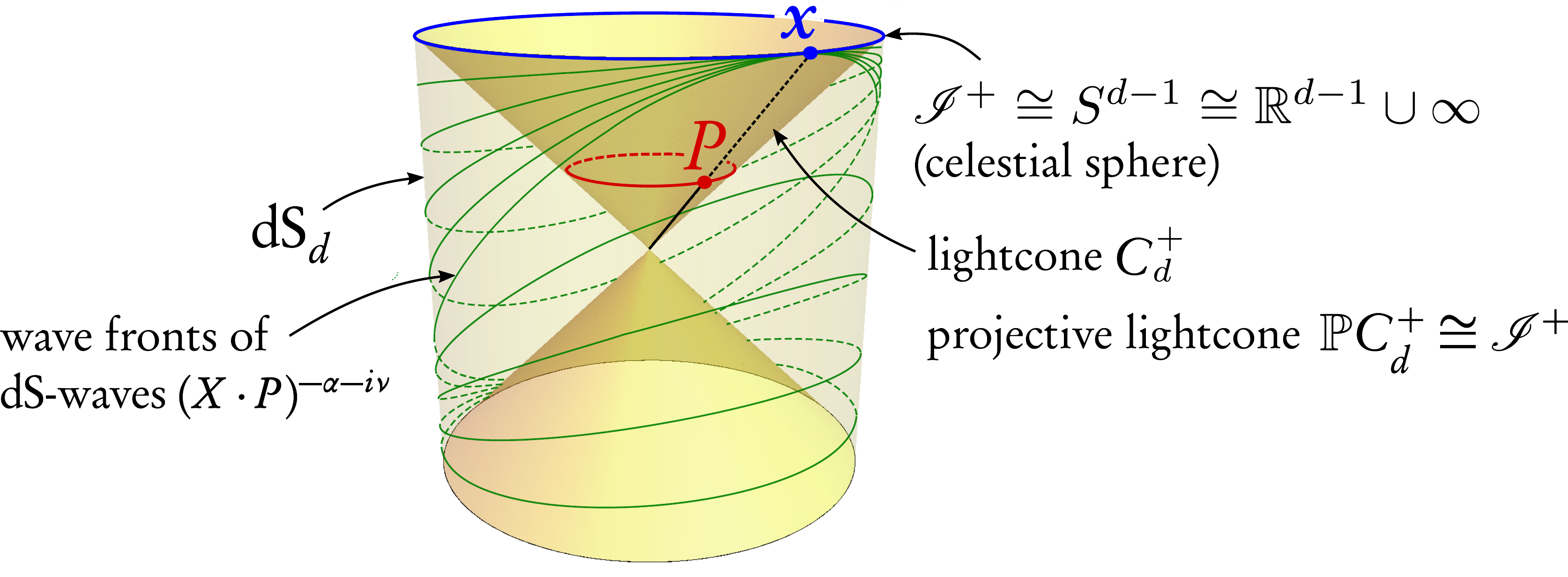}
		\caption{\em Correspondence $P\leftrightarrow\bx$}
  \end{figure}
  \end{center}
The de Sitter momentum $P$ is an element of the projective lightcone ${\mathbb P}C_d^+$. It can be identified with a point $\bx \leftrightarrow P$ in $(d-1)$-dimensional Euclidean space $\RR^{d-1} \cup \infty \cong {\mathbb S}^{d-1}$, as depicted in the figure above. This ``celestial sphere'' is identified with ${\mathscr I}^+$.
If we restrict $\theta$ to $N$ discrete values $\{\theta_1, \dots, \theta_N\}$ and set $z_{R,j}(\bx) = z_R(\theta_j, P)$ under the correspondence $\bx \leftrightarrow P$, 
we can also define a multiplet of $N$ {\em Euclidean} quantum fields
on $\RR^{d-1}$ as
\ben\label{field2}
\phi^{\rm Eucl.}_{R,j}(\bx) = z^\dagger_{R,j}(\bx) + z_{R,j}(\bx)\ , \quad j=1,\dots, N \ .
\een
The idea is that the duality maps the states created by the operators $\zd_{R}(\theta_j, P)$~\eqref{field1} to the states created by the operators $\zd_{R,j}(\bx)$~\eqref{field2} in the ``continuum limit'' $N \to \infty$. This map ``intertwines'' the unitary representation of the de~Sitter group $SO(d,1)$ on the respective Hilbert spaces. In particular, both sides of the dualtiy are based on the same irreducible unitary representations of the de Sitter group. The action of the de Sitter group is by construction geometrical on both sides, moving points $\bx$ of ${\mathscr I}^+$ in Euclidean CFT, and de~Sitter momenta $P$ in the sigma-model. Furthermore, fields at past null infinity~$\mathscr{I}^-$ are related by a TCP operator $\Theta$ which is introduced in the main text. Although our proposed correspondence is in principle mathematically precise in this sense, it remains to be seen how it is related to, say, a quantization of the  string in de~Sitter spacetime, say, along the lines of \cite{AF09}. It would 
also be interesting to see what the relation might be to other, rather different-looking proposals that are based on topological field theories, e.g. in \cite{H98}. We must leave this to a future investigation.

\section{Invariant Yang-Baxter operators and functions}\label{section:yb-ops}

The main input into all our constructions is a solution of the Yang-Baxter equation (YBE) with additional symmetries. In this section, we consider solutions to the YBE that are compatible with a representation  $V$ of a group $G$ and an associated conjugation $\Gamma$. In our subsequent construction of field-theoretic models, we will be interested in the case where $G=SO^\uparrow(d,1)$, $V$ is an irreducible (spin-0, principal, complementary or discrete series) representation of it, and $\Gamma$ the corresponding TCP operator. 

In order to compare with $O(N)$ sigma models and related constructions, we introduce in Section~\ref{subsection:invariant-YB-ops} Yang-Baxter operators and functions in the general context of a unitary representation of an arbitrary group $G$, and discuss some examples. The relevant aspects of the representation theory of $SO^\uparrow(d,1)$ are recalled in Section~\ref{section:UIR-LorentzGroup} in a manner suitable for our purposes, and the connection to the Klein-Gordon equation on de Sitter space is recalled in Section~\ref{subsection:Klein-Gordon}. These representations are then used in Section~\ref{section:YBOps-LorentzGroup}, where examples of invariant Yang-Baxter operators for the Lorentz group are presented.

\subsection{Definitions and examples}\label{subsection:invariant-YB-ops}

In the following, a {\em conjugation} on a Hilbert space means an antiunitary involution, and the letter $F$ is reserved for the flip $F:\K\ot\K\to\K\ot\K$, $F(k_1\ot k_2)=k_2\ot k_1$, on the tensor square of a Hilbert space $\K$. When it is necessary to emphasize the space, we will write $F_\K$ instead of $F$. For identities on various spaces, we write $1$ and only use more specific notation like $1_\K$ or $1_{\K\ot\K}$ where necessary.

\begin{definition}\label{definition:Invariant-YB-Operator}
     Let $G$ be a group, $V$ a unitary representation of $G$ on a Hilbert space $\K$, and $\Gamma$ a conjugation on $\K$. 
     \begin{enumerate}
     	\item An {\em invariant Yang-Baxter operator} (for $V,\Gamma$) is an operator $R\in\B(\K\ot\K)$ such that 
		\begin{enumerate}
			\item[(R1)] $R$ is unitary.
			\item[(R2)] $[R,\,V(g)\ot V(g)]=0$ for all $g\in G$.
			\item[(R3)] $[R,\,(\Gamma\ot\Gamma)F]=0$.
			\item[(R4)] $(R\ot1)(1\ot R)(R\ot1)=(1\ot R)(R\ot 1)(1\ot R)$ as an equation in $\B(\K\ot\K\ot\K)$ (i.e., with $1=1_\K$).
			\item[(R5)] $R^2=1_{\K\ot\K}$.
		\end{enumerate}
		The family of all invariant Yang-Baxter operators for a given representation $V$ and conjugation $\Gamma$ will be denoted $\R_{\rm op}(V,\Gamma)$.
		\item An {\em invariant Yang-Baxter function} (for $V$, $\Gamma$) is a function $R\in L^\infty(\Rl\to\B(\K\ot\K))$ such that for almost all $\te,\te'\in\Rl$,
	  \begin{itemize}
			\item[(R1')] $R(\te)$ is unitary.
			\item[(R2')] $[R(\te),\,V(g)\ot V(g)]=0$ for all $g\in G$.
			\item[(R3')] $(\Gamma\ot\Gamma) FR(-\te)F(\Gamma\ot\Gamma)
			=R(\te)$.
			\item[(R4')] $(R(\te)\ot1)(1\ot R(\te+\te'))(R(\te')\ot1)=(1\ot R(\te'))(R(\te+\te')\ot1)(1\ot R(\te))$ as an equation in $\B(\K\ot\K\ot\K)$ (i.e., with $1=1_\K$).
			\item[(R5')] $R(-\te)=R(\te)^{-1}$.
	  \end{itemize}
	The set of all invariant Yang-Baxter functions will be denoted $\R_{\rm fct}(V,\Gamma)$.
     \end{enumerate}
\end{definition}

Our main interest is in invariant Yang-Baxter {\em operators}, the Yang-Baxter functions serve as an auxiliary tool to construct them. Independent of $G,V,\Gamma,\K$, the four operators $\pm1, \pm F$ are always elements of $\R_{\rm op}(V,\Gamma)$; these are the trivial unitaries satisfying the constraints (R1)---(R5). The structure of $\R_{\rm op}(V,\Gamma)$ and $\R_{\rm fct}(V,\Gamma)$ depends heavily on the representation $V$ and group $G$, as we will see later in examples.

It is clear from the definition that any (say, continuous) function $R\in\R_{\rm fct}(V,\Gamma)$ defines an invariant Yang-Baxter operator $R(0)\in\R_{\rm op}(V,\Gamma)$, and any operator $R\in\R_{\rm op}(V,\Gamma)$ defines a (constant) Yang-Baxter function $R(\te):=R$. Furthermore, any invariant Yang-Baxter function defines an invariant Yang-Baxter operator on an enlarged space. This is spelled out in the following elementary construction, which we will use later on in the context of our QFT models.

\begin{lemma}\label{lemma:YBF->YBO}
	Let $R\in\R_{\rm fct}(V,\Gamma)$ (for some group $G$, on some Hilbert space $\K$), and consider the enlarged Hilbert space $\underline\K:=L^2(\Rl,d\te)\ot\K\cong L^2(\Rl\to\K,d\te)$, with $G$-representation $1\ot V$ and conjugation $((C\ot\Gamma)\psi)(\te):=\Gamma\psi(\te)$. On $\underline\K\ot\underline\K\cong L^2(\Rl^2\to\K\ot\K,d\te_1d\te_2)$, define the operator
	\begin{align}\label{eq:YBF->YBO}
		(\boldsymbol{R}\Psi)(\te_1,\te_2)
		:=
		R(\te_1-\te_2)\Psi(\te_2,\te_1)\,.
	\end{align}
	Then $\boldsymbol{R}\in\R_{\rm op}(1\ot V,C\ot\Gamma)$.
\end{lemma}

The proof of this lemma amounts to inserting the definitions and is therefore skipped. See \cite{LechnerSchutzenhofer:2013,BischoffTanimoto:2013} for similar results.

\medskip

For later use, we mention that we can also work with a different measure $d\nu(\te)$ than Lebesgue measure $d\te$ in this construction. For example, we can take a finite number $N$ of point measures, located at $\te_1,...,\te_N\in\Rl$. In that case, $L^2(\Rl^2,d\nu(\te_1)d\nu(\te_2))\cong\Cl^N\ot\Cl^N$, with orthonormal basis $\{e_{jl}\}_{j,l=1}^N$, invariant under the conjugation $C\ot C$. On vectors $\Psi_{jl}:=\Psi\ot e_{jl}\in \K\tp{2}\ot\Cl^N\ot\Cl^N$, our invariant Yang-Baxter operators then take the form 
\begin{align}\label{eq:R-amp}
	(R\Psi)_{jl}:=R(\te_j-\te_l)\Psi_{lj}\,.
\end{align}

This construction can also be seen as an example of the partial spectral disintegration formulas considered in \cite{BischoffTanimoto:2013}.

\bigskip 

Before presenting examples, we recall how an invariant Yang-Baxter operator $R$ gives rise to an $R$-symmetric Fock space, following \cite{LiguoriMintchev:1995,Lechner:2003,LechnerSchutzenhofer:2013}. In this Fock space construction, we consider an invariant Yang-Baxter operator $R$ and call its group representation, conjugation, and Hilbert space $V_1$, $\Gamma_1$, and $\Hil_1$, as these data enter on the one particle level.

As is well known, solutions of the Yang-Baxter equation (R4) induce representations of the braid group of $n$ strands on $\Hil_1\tp{n}$, by representing the elementary braid $\beta_k$, $k=1,...,n-1$, by ${\rm id}_{\Hil_1\tp{k-1}}\ot R\ot{\rm id}_{\Hil_1\tp{n-k-1}}$. This representation factors through the permutation group because of (R5), i.e. we have a representation $D_n^R$ of the symmetric group $\frS_n$ on $n$ letters on $\Hil_1\tp{n}$. Since $R$ is unitary (R1), so are the representations $D_n^R$. We denote by $\Hil_n^R\subset\Hil_1\tp{n}$ the subspace on which $D^R_n$ acts trivially, i.e.
\begin{align}\label{eq:def-R-symmetric-space}
     \Hil_n^R=P^R_n\Hil_1\tp{n}\,,\qquad
     P^R_n:=\frac{1}{n!}\sum_{\pi\in\frS_n}D_n^R(\pi)\,.
\end{align}
In view of (R2), the representation $g\mapsto V_1(g)\tp{n}$ of $G$ on $\Hil_1\tp{n}$ commutes with the projection $P^R_n$, and hence restricts to $\Hil^R_n$. We denote this restriction by $V_n:=V\tp{n}|_{\Hil^R_n}$.

Furthermore, we define a conjugation $\tilde\Gamma_n$ on $\Hil_1\tp{n}$ by
\begin{align}
     \tilde\Gamma_n:=\Gamma_1\tp{n}F_n\,,\qquad F_n(k_1\ot...\ot k_n):=k_n\ot...\ot k_1\,.
\end{align}
It is clear that $\tilde\Gamma_n$ is a conjugation on $\Hil_1\tp{n}$, and thanks to (R3), it commutes with $P^R_n$ and thus restricts to $\Hil_n^R$ \cite{LechnerSchutzenhofer:2013}. We call this restriction $\Gamma_n:=\tilde\Gamma_n|_{\Hil_n^R}$.

The $R$-symmetric Fock space over $\Hil_1$ is then defined as
\begin{align}
     \Hil^R:=\bigoplus_{n=0}^\infty\Hil^R_n\,,\quad \text{with }\;\Hil_0^R:=\Cl\,.
\end{align}
We denote\footnote{Note that despite our notation, also $V,\Gamma,\Om$ depend on $R$.} its Fock vacuum by $\Om:=1\oplus0\oplus0...$, the resulting ``$R$-second quantized'' representation of $G$ by $V:=\bigoplus_n V_n$, and the resulting conjugation by $\Gamma:=\bigoplus_n\Gamma_n$.

\bigskip

$R$-symmetric Fock spaces generalize the usual Bose/Fermi Fock spaces, which are given by the special cases $R=\pm F$. For our purposes, the $R$-symmetric spaces (for non-trivial $R$) will be convenient representation spaces for our models. We next give some examples of invariant Yang-Baxter operators and functions. 
 
\bigskip

\noindent{\bf Example 1: $\boldsymbol{O(N)}$.} We consider the group $G=O(N)$ in its defining representation $V$ on~$\Cl^N$, with complex conjugation in the standard basis as conjugation. This is a typical finite-dimensional example, which appears in particular in the context of the $O(N)$ sigma models. It is known from classical invariant theory that the $O(N)$-invariance constraint (R2) allows only three linearly independent solutions: The identity $1$ of $\Cl^N\ot\Cl^N$, the flip $F$, and a one-dimensional symmetric projection $Q$ \cite[Thm.~10.1.6]{GoodmanWallach:2009}. One can then check that (R1)---(R5) together only allow for trivial solutions, i.e. $\R_{\rm op}(V,C)=\{\pm1,\,\pm F\}$. 

However, non-trivial Yang-Baxter {\em functions} $R\in\R_{\rm fct}(V,C)$ do exist. A prominent example is
\begin{align*}
	R(\te)
	&=
	\sigma_1(\te)\cdot Q+\sigma_2(\te)\cdot1+\sigma_3(\te)\cdot F\,,\\
	\sigma_2(\theta)
	&:=
	g(\theta)g(i\pi-\theta),
	\quad \text{with} \quad
	g(\theta)
	:=
	\frac{\Gamma(\frac{1}{N-2}-i\frac{\theta}{2\pi})\Gamma(\frac{1}{2}-i\frac{\theta}{2\pi})}
	{\Gamma(\frac{1}{2}+\frac{1}{N-2}-i\frac{\theta}{2\pi})\Gamma(-i\frac{\theta}{2\pi})}
	  ,\\
	  \sigma_1(\theta)
	  &:=
	  -\frac{2\pi i}{(N-2)}\,\frac{\sigma_2(\theta)}{i\pi-\theta}
	  \,,\qquad
	  \sigma_3(\theta)
	  :=
	  -\frac{2\pi i}{(N-2)}\,\frac{\sigma_2(\theta)}{\theta}
	  \,,
\end{align*}
which satisfies not only (R1')---(R5'), but also the analytic properties (R6'), (R7') that will be introduced in Section~\ref{section:crossing}. This Yang-Baxter function describes the $O(N)$-invariant two-body S-matrix of the $O(N)$-sigma model \cite{ZamolodchikovZamolodchikov:1978}.

\bigskip

\noindent{\bf Example 2: The ``$\boldsymbol{ax+b}$'' group $\times$ inner symmetries.} As an example of a different nature, we consider ``$ax+b$'' group, i.e. the affine group $\PG_o$ generated by translations $u\mapsto u+x$ and dilations $u\mapsto e^{-\la} u$ on the real line $\Rl$. The physical interpretation is to view $\Rl$ as a lightray, which describes one chiral component of a massless field theory on two-dimensional Minkowski space\footnote{With minor modifications, this construction can also be carried out for a massive representation of the Poincar\'e group in two space-time dimensions.}. 

The group $\PG_o$ has a unique unitary irreducible representation $U_o$ in which the generator of the translations is positive. We may choose $L^2(\Rl,d\te)$ as our representation space, and then have, $(x,\la)\in \PG_o$,
\begin{align}\label{eq:Uo}
	(U_o(x,\la)\psi)(\te)
	=
	e^{ix\, \exp(\te)}\cdot\psi(\te-\la)
	\,,
\end{align}
where the variable $\te$ can be thought of as being related to the (positive) light like momentum $p$ by $p=e^\te$. The conjugation $(C\psi)(\te):=\overline{\psi(\te)}$ extends this representation to also include the reflection $x\mapsto-x$ on the lightray.

In this example, the invariant Yang-Baxter operators $R\in\R_{\rm op}(U_o,C)$ can all be computed, and in contrast to the $O(N)$ case, many such operators exist. The physically interesting ones are given by multiplication operators of the form
\begin{align}
	(R\Psi)(\te_1,\te_2)=\sigma(\te_1-\te_2)\cdot\Psi(\te_2,\te_1)
\end{align}
as in \eqref{eq:YBF->YBO}, where $\sigma\in L^\infty(\Rl\to\Cl)$ is a scalar function satisfying
\begin{align}
	\overline{\sigma(\te)}=\sigma(\te)^{-1}=\sigma(-\te)
	\,.
\end{align}
Such ``scattering functions'' include for example the two-body S-matrix of the Sinh-Gordon model, which is \cite{ArinshteinFateevZamolodchikov:1979}
\begin{align}\label{eq:scat-func}
	\sigma(\te)
	=
	\frac{\sinh\te-ib}{\sinh\te+ib}\,,
\end{align}
where $0<b<\pi$ is a function of the coupling constant.

\medskip

To generalize to a setting with inner symmetries, we can also, instead of $\PG_o$ alone, take the direct product $\PG_o\times G$ of $\PG_o$ with an arbitrary group $G$, which is thought of as the group of global gauge transformations. We then consider a unitary representation $V$ of $G$ on an additional Hilbert space $\K$, and form the direct product representation $U_o\otimes V$ on $\Hil_1=L^2(\Rl,d\te)\ot\K\cong L^2(\Rl\to\K,d\te)$, i.e.
\begin{align}\label{eq:UoV}
	((U_o(x,\la)\ot V(g))\psi)(\te)
	=
	e^{i\,x\,\exp(\te)}\cdot V(g)\psi(\te-\la)\,.
\end{align}
For later application, we stress that $\K$ can still be infinite-dimensional, as it is the case for the irreducible representations of $G=SO^\uparrow(d,1)$.

To also have a TCP operator in this extended setting, we assume that there exists a conjugation $\Gamma$ on $\K$ that commutes with $V$ (i.e., $V$ must be a self-conjugate representation), and then consider $C\ot\Gamma$ as TCP operator for $U_o\ot V$.

Under mild regularity assumptions, one can then show that essentially all invariant Yang-Baxter operators $\boldsymbol{R}\in\R_{\rm op}(U_o\ot V,C\ot\Gamma)$ are again of the form $(\boldsymbol{R}\Psi)(\te_1,\te_2)=R(\te_1-\te_2)\Psi(\te_2,\te_1)$ \eqref{eq:YBF->YBO}. That is, $\boldsymbol{R}$ acts by multiplying with an (operator-valued) function $R$, and this function $R$ has to exactly satisfy the requirements (R1')---(R5'). In particular, $R(\te)$ commutes with $V(g)\ot V(g)$ for all $\te\in\Rl$, $g\in G$.

This example is therefore quite different from the previous $O(N)$-example: Many invariant Yang-Baxter operators exist, and they are essentially all given by invariant Yang-Baxter functions via \eqref{eq:YBF->YBO}. Both examples can be combined by taking the inner symmetry group as $G=O(N)$, as one would do for describing the $O(N)$-models~\cite{LechnerSchutzenhofer:2013,Alazzawi:2014}. 

\medskip

The construction just outlined here can be used to describe the one-particle space of a (chiral component of) a massless sigma model with symmetry group $G$. To prepare our construction of such models for $G=SO^\uparrow(d,1)$, we review some representation theory of this group next.

\subsection{Unitary representations of $\boldsymbol{SO^\uparrow(d,1)}$}\label{section:UIR-LorentzGroup}

We now turn to the case of central interest for this article, the (proper, orthochronous) Lorentz group $G=SO^\uparrow(d,1)$, $d\geq2$. In later sections, this group will appear either as  the isometry group of $d$-dimensional de Sitter space $\dS_d$ or as the conformal group of $\Rl^{d-1}$. In this section, we first give a quick tour d'horizon of some of its representation theory. Readers familiar with this subject can skip to the next section.

Our exposition is in the spirit of~\cite{BrosMoschella:1996},~\cite{EpsteinMoschella:2014} (and references therein), and we will use the following notation: 
Capital letters $X,P$ etc. denote points in Minkowski space $\Rl^{d+1}$. The dot product of this Minkowski spacetime is defined with mostly minuses in this paper, 
\ben
X \cdot Y = X_0 Y_0 - X_1 Y_1 - \ldots - X_d Y_d \ . 
\een
Points in $\Rl^{d-1}$ are denoted by boldface letters, $\bx, \bp$, and their Euclidean norm is written as $|\bx|^2 = \sum_{i=1}^{d-1} x_i^2$. 

Unitary irreducible representations (UIRs) 
of $SO^\uparrow(d,1)$ are classified by a continuous or discrete parameter corresponding roughly to the ``mass'' in the Minkowski context, 
and a set of spins corresponding to the $\lceil \half d \rceil$ Casimirs of $ \frak{so}(d,1)$. In this paper we will only consider the case of zero spin\footnote{
For $d=2$, there is no spin, and our representations exhaust all possibilities, see~\cite{Lang:1975}.} and ``principal-'', 
``complementary-'' and ``discrete series'' representations. There are many unitarily equivalent models for these representations in the literature, see 
e.g.~\cite{Lang:1975,VilenkinKlimyk:1991}. 
The most useful description for our purposes is as follows. First define the future lightcone  
\ben
C_d^+ = \{ P \in \RR^{d+1} \mid P \cdot P = 0 \ , \ \ \  P_0 >0\}
\een
in $(d+1)$-dimensional  Minkowski space. We think of $C_d^+$ as a (redundant) version of momentum space in the Minkowski context. 
On $C_d^+$, consider smooth $\Cl$-valued ``wave functions'' $\psi$ which are homogeneous,
\ben\label{hom}
	\psi(\lambda P) = \lambda^{-\tfrac{d-1}{2} -i\nu}\cdot \psi(P) \ ,  \qquad \text{for all $\lambda>0$,}
\een
where at this stage, $\nu \in \Cl$ is arbitrary. As the fraction $\tfrac{d-1}{2}$ will appear frequently, we introduce the shorthand $\alpha:=\tfrac{d-1}{2}$.

The collection of these wave functions forms a complex vector space which we will call~$\K_\nu$. A linear algebraic representation 
of $\Lambda \in SO^\uparrow(d,1)$ is defined by pullback,
\ben
\label{Udef}
	V_\nu(\Lambda): \K_\nu \to \K_\nu \ , \qquad  [V_\nu (\Lambda) \psi](P) := \psi(\Lambda^{-1} P) \ . 
\een
In order for this to define a unitary representation, we must equip $\K_\nu$ with an invariant (under $V_\nu(\Lambda)$) positive definite inner product. It turns out that this is possible only for certain values of $\nu$. These are\footnote{In the case of 
the discrete series, the inner product is in fact only defined on an invariant subspace of $\K_\nu$, see below.}:
\begin{enumerate}
	\item[a)] (Principal series)  $\nu \in \RR$. 
	\item[b)] (Complementary series) $i\nu \in (0,\alpha)$.
	\item[c)] (Discrete series) $i\nu \in \alpha + {\mathbb N}_0$.
\end{enumerate}
Two complementary or discrete series representations $V_\nu$, $V_{\nu'}$ are inequivalent for $\nu\neq\nu'$, and two principal series representations $V_\nu$ and $V_{\nu'}$ are equivalent if and only if $\nu=\pm\nu'$.
We now explain what the inner products are in each case. As a preparation, consider first the $d$-form $\mu$ and vector field $\xi$ on $C_d^+$ defined by 
\ben
\mu = \frac{\dd P_1 \wedge \cdots \wedge \dd P_d}{P_0} \ , \quad \xi = P_0 \frac{\partial}{\partial P_0} + \dots + P_d \frac{\partial}{\partial P_d} \ . 
\een 
$\mu$ is the natural integration element on the future lightcone, and $\xi$ the generator of dilations. Both are invariant under any
$\Lambda \in SO^\uparrow(d,1)$, i.e.  $\Lambda^* \mu = \mu, \Lambda_* \xi = \xi$. We then form
\begin{align}\label{eq:omega}
	\omega = i_\xi \mu = \sum_{k=1}^d (-1)^{k+1} \frac{P_k}{P_0} \, \dd P_1 \wedge \dots \widehat{\dd P_k} \wedge \dots \dd P_d \ , 
\end{align}
where $i_\xi$ is Cartan's operator contracting the upper index of the vector $\xi$ into the first index of the $d$-form $\mu$. A key lemma which we use time and again
is the following \cite[Lemma~4.1]{BrosMoschella:1996}:

\begin{lemma}\label{lemma:closedform}
	Suppose $f$ is a homogeneous function on the future lightcone $C_d^+$ of degree $-(d-1)$. Then $f\omega$ is a closed $(d-1)$-form on $C_d^+$, $\dd(f \omega) = 0$.
\end{lemma}

Using this lemma, we can now describe the inner products. \\

\medskip
\noindent
{\bf a) Principal series:} Here the degree of homogeneity of the wave functions is $-\alpha- i\nu$ with $\nu$ {\rm real}. Consequently, the product $f = \overline \psi_1 \psi_2$ of two smooth 
wave functions $\psi_1,\psi_2\in\K_\nu$ is homogeneous of degree $-(d-1)$, so the lemma applies. We choose an ``orbital base'' \begin{wrapfigure}{r}{0.3\textwidth}
  \begin{center}
    \includegraphics[width=0.23\textwidth]{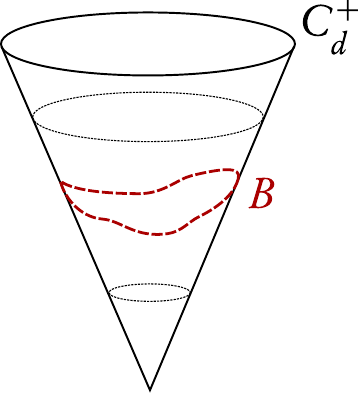}\\
  \end{center}
\end{wrapfigure}
$B \cong {\mathbb S}^{d-1}$ of $C_d^+$ (i.e. a closed manifold intersecting each generatrix of $C_d^+$ once, see figure on the right, and Appendix~A for explicit formulas), and define a positive definite inner product by
\ben\label{inn}
\big(\psi_1, \psi_2\big)_\nu :=  \int_B \omega(P) \, \overline{\psi_1(P)}\, \psi_2(P) \ . 
\een
By Lemma~\ref{lemma:closedform}, this definition is independent of the particular choice of $B$, in the sense that, if $B'$ is homologous to $B$, then the inner product defined with $B'$ instead of $B$ coincides with \eqref{inn}. This fact implies at once that 
the operators $V_\nu(\Lambda)$ are unitary with respect to this inner product. Since the type of argument is used time and again, we explain the details:
\begin{align}
	\label{unitary}
	\big(V_\nu(\Lambda) \psi_1, V_\nu(\Lambda) \psi_2\big)_\nu
	&=
	\int_B \omega(P) \, \overline{\psi_1(\Lambda^{-1} P)} \psi_2(\Lambda^{-1} P)
	\\
	&=\int_B (\Lambda^{-1})^* \omega( P) \, \overline{\psi_1(\Lambda^{-1} P)} \psi_2(\Lambda^{-1} P)
	\nonumber
	\\
	&=
	\int_{\Lambda^{-1} \cdot B} \omega(P) \, \overline{\psi_1(P)} \psi_2(P)
	\nonumber
	\\
	&=
	\int_{B} \omega(P) \, \overline{\psi_1(P)} \psi_2(P)
	\nonumber
	\\
	&=
	\big(\psi_1, \psi_2\big)_\nu
	\,.
	\nonumber
\end{align}

The first equality sign is the definition. In the second equality, it is used that $(\Lambda^{-1})^* \omega = \omega$, and in the third 
equality, a change of variables $P \to \Lambda^{-1} \cdot P$ was made. In the last step we used Stokes theorem, noting that the integrand is a closed form, and that $B$ and $\Lambda^{-1} \cdot B$ are homologous. 

The Hilbert space of the principal series representation is defined as the completion in the inner product~\eqref{inn} of the space $\K_\nu$, which we denote by the same symbol. 
\\\\
Pointwise complex conjugation does not leave $\K_\nu$ invariant (unless $\nu=0$) because the degree of homogeneity is complex, and conjugation changes $\nu$ to $-\nu$. Since $\om$ is real, this implies that $\psi\mapsto\overline{\psi}$ is an antiunitary map $\K_\nu\to\K_{-\nu}$, intertwining $V_\nu$ and $V_{-\nu}$.

To compare $V_\nu$ and $V_{-\nu}$, it is useful to introduce the integral operator
\ben\label{Idef1}
(I_\nu \psi)(P)
:=
(2\pi)^{-\alpha} \,
\frac{\Gamma(\alpha - i\nu)}{\Gamma(i\nu)} \int_B \omega(P') \, (P \cdot P')^{-\alpha + i\nu} \, \psi(P')\,,\qquad
\psi\in\K_\nu
\,.
\een
If $\alpha - i\nu,i\nu \notin -\Nl_0$, the poles of the Gamma function are avoided. The integrand has a singularity at $P=P'$ which is integrable if $\re(i\nu)>0$ (see Appendix~A for explicit formulas arising from particular choices of $B$). Thus $I_\nu$ is well-defined as it stands in particular for $i\nu\in(0,\alpha)$, corresponding to the case of a complementary series representation, to be discussed below. For the principal series representations, $\nu$ is real, and we may define \eqref{Idef1} by replacing $\nu$ with $\nu-i\eps$ and then taking the limit $\eps\searrow0$ (as a distributional boundary value, see for example \cite[Chap.~3]{VilenkinKlimyk:1991}). This adds a delta function term and yields a well-defined operator $I_\nu$ for $\nu\in\Rl\backslash\{0\}$. Finally, for $\nu=0$, one has to take into account the Gamma factors in \eqref{Idef1} when performing the limit $\nu\to0$. In this case, one obtains $I_0=1$ (this follows from the delta function relation \eqref{Delta} in Appendix~B).

After these remarks concerning the definition of $I_\nu$, note that for $\psi\in\K_\nu$, the value of the integral \eqref{Idef1} does not depend on our choice of orbital base $B$, because the integrand clearly has homogeneity $-(d-1)$ in $P'$, and is thus a closed form by Lemma~\ref{lemma:closedform}.

\begin{lemma}\label{lemma:TPC-principal}
	\begin{enumerate}
		\item In the principal series case ($\nu \in \RR$), $I_\nu: \K_\nu \to \K_{-\nu}$ is a unitary operator intertwining $V_\nu$ and $V_{-\nu}$, i.e.
		\ben\label{intw0}
			V_{-\nu}(\Lambda) I_\nu = I_\nu V_\nu(\Lambda)\,,\qquad\La\in SO^\uparrow(d,1)\,.
		\een
		Furthermore, it holds that $I_{-\nu}^{}=I_\nu^* = I_{\nu}^{-1}$ and $I_0 = 1$. 
		\item Each principal series representation $V_\nu$ is selfconjugate: 	$(\Gamma_\nu\psi)(P):=\overline{(I_\nu \psi)(P)}$ is a conjugation on $\K_\nu$ commuting with $V_\nu$.
	\end{enumerate}
\end{lemma}
\begin{proof}
	$a)$ By using the invariance $(P\cdot P')=(\La P\cdot \La P')$ and the same type of argument as given in eqs.~\eqref{unitary}, $I_\nu$ is seen to have the intertwining property \eqref{intw0}. The equality $I_\nu^*=I_{-\nu}$ follows by a routine calculation. To show $I_{-\nu}=I_{\nu}^{-1}$ (for $\nu\neq0$), it is useful to choose a convenient parameterization of $B$. Using the spherical parameterization (see Appendix A), the identity $I_{-\nu}=I_{\nu}^{-1}$ follows by application of the composition relation \eqref{eq:Composition-Relation}. The special case $I_0=1$ has been explained above already.
	
	$b)$ It is clear that $\Gamma_\nu$ is an antiunitary operator on $\K_\nu$. As complex conjugation $C$ satisfies $CI_\nu C=I_{-\nu}=I_\nu^{-1}$, one also sees that $\Gamma_\nu$ is an involution, i.e. $\Gamma_\nu^2=1$.
\end{proof}

\medskip
\noindent
{\bf b) Complementary series:} Here the homogeneity of the wave functions is $-\alpha - i\nu$, with $i \nu \in (0, \alpha)$, so $\nu$ is {\em imaginary}. In this case, we  cannot apply the same procedure 
as in the case of the principal series to form a scalar product, because the product $f = \overline{\psi_1}\psi_2$ of two smooth 
wave functions $\psi_1,\psi_2\in\K_\nu$ is not homogeneous of degree $-(d-1)$, and consequently, Lemma~\ref{lemma:closedform} does not apply. To get around this problem, we can use the operator $I_\nu:\K_\nu\to\K_{-\nu}$ \eqref{Idef1}, which is well-defined also for $i\nu\in(0,\alpha)$.

As in the case of the principal series, also here the integral operator $I_\nu$ \eqref{Idef1} has the intertwining property $V_{-\nu}(\Lambda) I_\nu = I_\nu V_\nu(\Lambda)$ for all $\Lambda \in SO^\uparrow(d,1)$. As $\nu$ is imaginary, the intertwining operator $I_\nu$ ensures that the function 
 $f = \overline{\psi_1}\, I_\nu \psi_2$ on $C^+_d$ formed 
from two wave functions $\psi_1, \psi_2 \in \K_\nu$ is homogeneous of degree $-(d-1)$, and Lemma~\ref{lemma:closedform}
shows that 
\ben\label{Idef}
	(\psi_1, \psi_2)_\nu
	:=
	\int_{B} \omega(P)\, \overline{\psi_1(P)}  (I_\nu \psi_1)(P)
\een
is again independent of the choice of orbital base $B$. The same argument as that given in eq.~\eqref{unitary} then 
also yields that the inner product just defined is invariant under the representation $V_\nu$. For the complementary series, we therefore take \eqref{Idef} as our inner product. We have:

\bigskip
\begin{lemma}
	For $i\nu \in (0, \alpha)$, the inner product \eqref{Idef} is positive definite.
\end{lemma}
\begin{proof}
	Since we are free to choose any orbital base in \eqref{Idef}, \eqref{Idef1}, we can make a convenient choice. If we choose the spherical model, $B \cong \SS^{d-1}$ (see Appendix~A for the different canonical models) then our lemma reduces to Lemma~5.5 of \cite{NeebOlafsson:2014-2}. The proof is however more transparent choosing the flat model, where $B \cong \RR^{d-1} \cup \infty$ is parametrized by $\Rl^{d-1}\ni\bx\mapsto P = (\half (|\bx|^2 + 1), \bx , \half(|\bx|^2 - 1))$. Using this parametrization, we find 
	\ben
	\begin{split}
	(\psi, \psi)_\nu
	&=c\,\frac{\Gamma(\alpha-i\nu)}{\Gamma(i\nu)}\, \int \dd^{d-1} \bx_1 \dd^{d-1} \bx_2 \ |\bx_1 - \bx_2|^{-(d-1)+2i\nu} \overline{\psi(\bx_1)} \psi(\bx_2) \\
	&= c'\,\frac{\Gamma(\alpha-i\nu)}{\Gamma(i\nu)} \int \dd^{d-1} \bp \ | \bp |^{-2i\nu} |\hat \psi(\bp)|^2 \ge 0 \ . 
	\end{split}
	\een
	Here $c,c'$ are positive numerical constants. In the second line we used the Plancherel theorem and a well-known formula for the Fourier transform of $|\bx|^s$ (see e.g. Ex. VII 7.13 of \cite{Schwartz:1966}). Taking into account standard properties of the Gamma function, we see that the prefactor of the integral is positive if $i\nu\in(0,\alpha)$.
\end{proof}

In the complementary series, the degree of homogeneity is real, and thus complex conjugation is a well-defined operation on $\K_\nu$. Moreover, complex conjugation commutes with~$I_\nu$ for imaginary $\nu$. Thus, if we define $(\Gamma_\nu\psi)(P):=\overline{\psi(P)}$, then
$\Gamma_\nu$ is an antiunitary involution on $\K_\nu$ in the case of the complementary series.

\begin{lemma}\label{lemma:TCP-complementary}
	Each complementary series representation is selfconjugate: $(\Gamma_\nu\psi)(P):=\overline{\psi(P)}$ is a conjugation commuting with $V_\nu$.{\hfill $\square$}
\end{lemma}

\medskip
\noindent
{\bf c) Discrete series:} Here the degree of homogeneity of the wave functions is $-\alpha- i\nu$ with $i\nu=\alpha+n, n\in\Nl_0$. For these values, the Gamma-factors in the definition of $I_\nu$ (see~\eqref{Idef1}), 
and hence also in the inner product of the complementary series~\eqref{Idef} become singular. Thus, one cannot, for this reason alone, define an inner product for the discrete series by analytic continuation of~\eqref{Idef}.
The way out is to pass from $\K_\nu$ to an $SO^\uparrow(d,1)$-invariant subspace of wave functions for which the scalar product~{\em can} be defined by analytic continuation. Since the kernel of $I_\nu$ is, up to 
divergent Gamma-factors, given by $(P \cdot P')^n$ for the discrete series, a natural choice for this subspace is the set of $\psi \in \K_\nu$ such that 
\ben\label{discr}
\psi(\lambda P) = \lambda^{-(d-1)-n} \psi(P) \ , \qquad \int_B \omega(P') (P \cdot P')^n \psi(P') = 0 \ , 
\een
where the first equality just repeats the homogeneity condition for the case $i\nu=\alpha+n$. The second condition is  independent of the choice of $B$, and hence indeed $SO^\uparrow(d,1)$-invariant. By abuse of 
notation, we denote the set of $\psi$ satisfying~\eqref{discr} again by $\K_\nu$. For such $\psi$, analytic continuation of~\eqref{Idef},\eqref{Idef1} to $i\nu = \alpha+n$ is now possible. Since the residue of $\Gamma$ at $-n$ is $(-1)^n/n!$, we find
\begin{align}\label{Kdef}
(\psi_1, \psi_2)_\nu
	&:=
	\int_{B} \omega(P)\, \overline{\psi_1(P)}  (I_\nu \psi_1)(P)
 \\
(I_\nu \psi)(P)
&:=
(2\pi)^{-\alpha} \,
\frac{(-1)^{n+1}}{n!\,\Gamma(\alpha+n)} \int_B \omega(P') \, (P \cdot P')^{n} \log (P \cdot P') \, \psi(P')\ . 
\end{align}
Using \eqref{discr}, one again verifies that the definition of $I_\nu$ remains independent of $B$, and therefore, by the same argument as already 
invoked several times, that the inner product~\eqref{Kdef} is invariant under $V_\nu(\Lambda)$. 
Since analytic continuation does not usually preserves positivity, it is non-trivial, however, that this inner product is actually positive definite~\cite{EpsteinMoschella:2014}. 

\begin{lemma}
	For $i\nu = \alpha + n, n=0,1,2, \dots$, the inner product \eqref{Kdef} is positive definite.
\end{lemma}
\begin{proof}
	Since we are free to choose any orbital base in \eqref{Kdef}, we can make a convenient choice. We choose the spherical model, 
	$B \cong \SS^{d-1}$ (see Appendix~A for the different canonical models), where $P \cdot P'=
	1-\hat p \cdot \hat p'$, with $\hat p, \hat p' \in \SS^{d-1}$. We have the 
	series, for $|x| < 1$
	\ben
	(-1)^{n+1} (1-x)^n \log (1-x) = \sum_{m>n} \frac{n!(m-n-1)!}{m!} x^m \ . 
	\een
	For $n>0$, this series is absolutely convergent, including the limit as $|x| \to 1$, and all its coefficients are evidently positive. We apply this 
	identity to $x = \hat p \cdot \hat p'$ in the inner product~\eqref{Kdef}. Exchanging integration and summation it follows that 
	\ben
	(\psi, \psi)_\nu =  \frac{(2\pi)^{-\alpha}}{\Gamma(\alpha+n)} \sum_{m>n} \frac{(m-n-1)!}{m!} \| a_m \|^2 \ge 0 \ , 
	\een
	where $a_m$ is the rank $m$ tensor on $\RR^d$ given by $a_m = \int_{\SS^{d-1}} \hat p^{\otimes m} \psi(\hat p) \ \dd^{d-1} \hat p$ and $\|a_m\|$ denotes the 
	norm of such a tensor inherited from the Euclidean metric on $\RR^d$. The integral form of the triangle inequality gives 
	$\|a_m\| \le  \int_{\SS^{d-1}} |\psi(\hat p)| \ \dd^{d-1} \hat p$, so the series is absolutely convergent, meaning that exchanging summation and integration was permissible. 
	The case $n=0$ can be treated e.g. using the flat model and applying a Fourier transform--we omit the details.
\end{proof}

As in the complementary series, the degree of homogeneity is real, and thus complex conjugation is a well-defined operation on $\K_\nu$. 
Moreover, complex conjugation commutes with $I_\nu$ for imaginary $\nu$. Thus, if we define $(\Gamma_\nu\psi)(P):=\overline{\psi(P)}$, then
$\Gamma_\nu$ is an antiunitary involution on $\K_\nu$ in the case of the discrete series.

\medskip
\noindent
This finishes our outline of the representations. In conclusion, we mention that the conjugation $\Gamma_\nu$,
\ben\label{Thdef}
	(\Gamma_\nu\psi)(P) = \begin{cases}
	\overline{(I_\nu \psi)(P)} & \nu\in\Rl \text{ (principal series)} \\
	\overline{\psi(P)} & i\nu\in(0,\alpha) \ {\rm or} \ i\nu \in \alpha + \mathbb{N}_0 \text{ (compl. or discrete series)}  
	\end{cases}\,,
\een
can be interpreted as a TCP operator, at least on the superficial level that it is an antiunitary involution and commutes with the representation, as the reflection $X\mapsto-X$ on $\dS_d$ commutes with $SO^\uparrow(d,1)$. These properties are clearly still satisfied if we multiply $\Gamma_\nu$ by a phase factor (as we shall do in later sections).
With the help of $\Gamma_\nu$, one can therefore extend $V_\nu$ to include the reflection $X\mapsto-X$.\footnote{We mention as an aside that by considering also partial reflections which only invert the sign of one of the components of
$X$, we could in fact extend $V_\nu$ to a (pseudo-unitary) representation of the full Lorentz group $O(d,1)$.}

\subsection{Representations of $\boldsymbol{SO^\uparrow(d,1)}$ and Klein-Gordon fields on $\boldsymbol{\dS_d}$}\label{subsection:Klein-Gordon}

In this short section we recall the relation between the principal and complementary series representations $V_\nu$ and classical and quantum Klein-Gordon fields on $d$-dimensional de Sitter spacetime\footnote{Such a relation 
exists also for the discrete series, but is more complicated, see~\cite{EpsteinMoschella:2014} for details.} ${\rm dS}_d$. For our purposes, this space is best defined as the hyperboloid 
\ben
	\dS_d = \{ X \in \RR^{d+1} \mid X \cdot X = -1 \}
\een
embedded in an ambient $(d+1)$-dimensional Minkowski spacetime $\RR^{d+1}$. The metric of $\dS_d$ is simply (minus) that induced from the ambient Minkowski space. It is manifest from this definition that the group of isometries of $\dS_d$ is $O(d,1)$, where group elements act by $X \mapsto \Lambda X$. 

On test functions $F\in C_0^\infty(\dS_d)$, orthochronous Lorentz transformations act according to $F\mapsto F_\La:=F\circ\La^{-1}$, and we choose an antilinear action of the full spacetime reflection $X\mapsto -X$ by $F\mapsto \overline{F_-}:X\mapsto\overline{F(-X)}$. The precise relationship between these transformations on $\dS_d$ and the ``momentum space'' representations $V_\nu$ from the previous section is, as in flat space, via a special choice of ``plane wave mode functions''. These mode functions are defined as follows. Let $P$ be any vector in $C_d^+$, choose any time-like vector such as $e=(1,0,0, \dots, 0)$ in the ambient $\RR^{d+1}$, and define 
\begin{align}\label{eq:dS-wave}
	u_P^\pm(X) 
	:=
	(X \cdot P)^{-\alpha -i\nu}_\pm = 
	\lim_{\epsilon \to 0+} [(X\pm i\epsilon e) \cdot P]^{-\alpha-i\nu} \ . 	
\end{align}
Adding a small imaginary part removes the phase ambiguity for $(X\cdot P)^{-\alpha-i\nu}$ when $X \cdot P$ becomes negative. The limit is understood in the sense of a (distributional) boundary value. The difference between the ``$+$'' (``positive frequency'') and ``$-$'' (``negative frequency'') mode arising when we cross to the other Poincar\'e patch is basically a phase.

The modes $(X \cdot P)^{-\alpha -i\nu}_\pm$ are (distributional) solutions to the Klein-Gordon equation in $X$ with mass $m^2=\alpha^2 + \nu^2$ on the {\em entire} de Sitter manifold, 
\ben
(\square + m^2) u^\pm_P = 0 \ . 
\een

Conversely, if $\psi \in \K_\nu$ is smooth, then the corresponding ``wave packet'' 
\ben
u_\psi^\pm(X) = \int_B \omega(P) \ \psi(P)(X \cdot P)^{-\alpha -i\nu}_\pm
\een
is a globally defined, smooth solution to the KG-equation.

To make contact with the representations $V_\nu$, we define for $F\in C_0^\infty(\dS_d)$
\begin{align}\label{eq:F-dS-pm}
	F^\pm_\nu(P)
	:=
\int_{\dS_d}d\mu(X)\,F(X)\,u_P^\pm(X)\,,
\end{align}
where $d\mu$ is the $O(d,1)$-invariant integration element on $\dS_d$ (Fourier-Helgason transformation). With these definitions, we have the following lemma.

\bigskip

\begin{lemma}\label{lemma:fourier-helgason}
Let $F\in C_0^\infty(\dS_d)$, and $\nu\in\Rl\cup-i(0,\alpha)$. Then 
	\begin{enumerate}
		\item $F^+_\nu\in\K_\nu$\,,
		\item For $\La\in O^\uparrow(d,1)$, it holds that $(F_\La)^+_\nu=V_\nu(\La)F^+_\nu$.
		\item $(\overline{F_-})^+_\nu=\gamma_\nu\cdot \Gamma_\nu F^+_\nu$, where $\gamma_\nu\in\Cl$ is the phase factor
		\begin{align}\label{eq:gamma-factor}
		\gamma_\nu
		&=
		\begin{cases}
			2^{i\nu}e^{-i\pi\alpha}\frac{\Gamma(\alpha-i\nu)}{\Gamma(\alpha+i\nu)} & \qquad\text{principal series}
			\\
			e^{-i\pi(\alpha+i\nu)} & \qquad\text{complementary series}
		\end{cases}
		\;.
		\end{align}
	\end{enumerate}
\end{lemma}
\begin{proof}
	$a)$ The de Sitter wave \eqref{eq:dS-wave} is homogeneous of degree $-\alpha-i\nu$ in $P$. $b)$ follows directly from the invariance of $\mu$. For $c)$, we first note $(-X\cdot P)^{-\alpha-i\nu}_-=e^{i\pi\alpha-\pi\nu}(X\cdot P)_+^{-\alpha-i\nu}$. For the complementary series, we have
	\begin{align*}
		(\Gamma_\nu(F_-)^+_\nu)(P)
		=
		\overline{(F_-)^+_\nu(P)}
		&=
		\int_{\dS_d}d\mu(X)\,\overline{F(-X)}\,(X\cdot P)_-^{-\alpha-i\nu}
		\\
		&=
		e^{i\pi\alpha-\pi\nu}
		\int_{\dS_d}d\mu(X)\,\overline{F(X)}\,(X\cdot P)_+^{-\alpha-i\nu}
		\\
		&=
		e^{i\pi\alpha-\pi\nu}\,(\overline{F})^+_\nu(P)\,,
	\end{align*}
	which implies the result in this case. For the principal series, we first recall \cite[Lemma~4.1]{EpsteinMoschella:2014}
	\begin{align}\label{eq:u-expansion}
		u_P^\pm(X)
		=
		\frac{e^{\pm\pi\nu}\,2^{i\nu}}{(2\pi)^\alpha}\,\frac{\Gamma(\alpha-i\nu)}{\Gamma(-i\nu)}\int_B\om(P')\,(P\cdot P')^{-\alpha-i\nu}\,(P'\cdot X)_\pm^{-\alpha+i\nu}\,.
	\end{align}
	Together with \eqref{eq:Composition-Relation}, this gives 
	\begin{align*}
		(\Gamma_\nu(F_-)^+_\nu)(P)
		&=
		2^{-i\nu}e^{i\pi\alpha}\frac{\Gamma(\alpha+i\nu)}{\Gamma(\alpha-i\nu)}\cdot(\overline{F})^+_\nu(P)\,,
	\end{align*}
	and the claimed result follows.
\end{proof}

We briefly indicate how these facts can be used to define a covariant Klein-Gordon quantum field on $\dS_d$: Denoting by $a,\ad$ the canonical CCR operators on the Fock space over $\K_\nu$
(this is the special case of the ``$R$-twisted'' Fock space of Section~\ref{subsection:invariant-YB-ops} given by taking $R$ as the tensor flip), we define
\begin{align}
	\varphi_\nu(F):=\ad(F^+_\nu)+a((\overline{F})^+_\nu)
	=
	\int_{\dS_d}d\mu(X)\,F(X)\varphi_\nu(X)\,.
\end{align}
or in the informal notation explained in more detail below in Sect.~\ref{section:qft}
\begin{align}\label{dSqft}
\varphi_\nu(X) = \int_B \omega(P)[ \ad(P) (X \cdot P)^{-\alpha -i\nu}_+ + {\rm h.c.} ]. 
\end{align}
It then follows immediately that the field $\varphi_\nu$ is a solution of the Klein-Gordon equation $(\Box_X+m^2)\varphi_\nu(X)=0$, which is real, $\varphi_\nu(X)^*=\varphi_\nu(X)$. It transforms covariantly under the second quantization of $V_\nu$, namely $V_\nu(\La)\varphi_\nu(X)V_\nu(\La)^{-1}=\varphi_\nu(\La X)$. Furthermore, taking as TCP operator $\Theta_\nu:=\gamma_\nu^{-1}\Gamma_\nu$, we have $\Theta_\nu F^+_\nu=(\overline{F_-})^+_\nu$, and hence also the TCP symmetry $\Theta_\nu\varphi_\nu(X)\Theta_\nu=\varphi_\nu(-X)$.

\subsection{$\boldsymbol{SO^\uparrow(d,1)}$-invariant Yang-Baxter operators}\label{section:YBOps-LorentzGroup}

We now take the Lorentz group $G=SO^\uparrow(d,1)$ in one of the representations $V_\nu$ from Section~\ref{section:UIR-LorentzGroup}, and associated conjugation $\Gamma_\nu$, and ask for the invariant Yang-Baxter operators and functions $\R_{\rm op}(V_\nu,\Gamma_\nu)$ and $\R_{\rm fct}(V_\nu,\Gamma_\nu)$, analogously to the examples for $G=O(N)$ and $G=\PG_o$ considered in Section~\ref{subsection:invariant-YB-ops}.

Our construction will be clearest in a slightly more general setting: Instead of a single representation, we consider two representations $V_{\nu_1}$, $V_{\nu_2}$ from either the principal, complementary or discrete 
series, i.e. $\nu_1,\nu_2\in\Rl\cup-i(0,\alpha) \cup \{-i\alpha-i{\mathbb N}_0\} $, and then construct an operator $R^{\nu_1\nu_2}:\K_{\nu_1}\ot\K_{\nu_2}\to\K_{\nu_2}\ot\K_{\nu_1}$ intertwining $V_{\nu_1}\ot V_{\nu_2}$ with $V_{\nu_2}\ot V_{\nu_1}$. 
This operator $R^{\nu_1\nu_2}$ will be an integral operator with distributional kernel $R^{\nu_1\nu_2}(P_1,P_2; P_1',P_2')$,
\ben\label{Rdef}
	(R^{\nu_1\nu_2}\Psi^{\nu_1\nu_2})(P_1, P_2)
	:=
	\int_{B \times B} \omega(P_1') \wedge \omega(P_2') \ R^{\nu_1\nu_2}(P_1, P_2; P_1', P_2')\, \Psi^{\nu_1\nu_2}(P_1',P_2') \ . 
\een
The degrees of homogeneity $d(P_k^{(\prime)})$ of the kernel in its four variables $P_1,P_1',P_2,P_2'$ will be
\begin{align*}
     d(P_1)&=-\alpha - i\nu_2\,,\\
     d(P_2)&=-\alpha - i\nu_1\,,\\
     d(P_1')&=-\alpha + i\nu_1\,,\\
     d(P_2')&=-\alpha + i\nu_2\,.
\end{align*}
For $\nu_1,\nu_2$ in the principal or complementary series, this implies immediately that the integrand of \eqref{Rdef} is homogeneous of degree $-(d-1)$ in both $P_1'$ and $P_2'$, so that \eqref{Rdef} does not depend on the choice of orbital base $B$ by Lemma~\ref{lemma:closedform}. Furthermore, it follows that $R^{\nu_1\nu_2}\Psi^{\nu_1\nu_2}$, $\Psi^{\nu_1\nu_2}\in\K_{\nu_1}\ot\K_{\nu_2}$, lies in $\K_{\nu_2}\ot\K_{\nu_1}$, i.e. $R^{\nu_1\nu_2}$ is a map $\K_{\nu_1}\ot\K_{\nu_2}\to\K_{\nu_2} \otimes \K_{\nu_1}$. The same conclusions also hold when one of the parameters $\nu_1,\nu_2$ (or both of them) belong to the discrete series. Here one has to check in addition that the constraint \eqref{discr} is preserved by the integral operator. This follows from the relation \eqref{Delta}.

\medskip

The integral kernel will be  taken of the form $(P_1' \cdot P_2)^{w_1} (P_2' \cdot P_1)^{w_2}(P_1' \cdot P_2')^{w_3} (P_1 \cdot P_2)^{w_4}$, where the exponents $w_i$ are complex numbers to be determined. This ansatz presumably does not really imply any serious loss of generality, because a general invariant kernel may be reduced to such expressions via a Mellin-transform along the lines of~\cite{Hollands:2013}. Since it only contains Lorentz invariant inner products, it follows immediately that the corresponding integral operator intertwines $V_{\nu_1}\ot V_{\nu_2}$ with $V_{\nu_2}\ot V_{\nu_1}$ whenever it is well-defined. Imposing the above degrees of homogeneity in $P_1,P_2,P_1',P_2'$ onto this kernel fixes the powers $w_1, \dots ,w_4$ up to one free parameter, which we call $i\te$. This then gives the integral kernels
\ben\label{Rdef-kernel}
\begin{split}
R_\te(P_1,P_2; P_1',P_2') =&\; c_{\nu_1, \nu_2}(\theta) \  (P_1 \cdot P_2)^{-i\theta - \half i\nu_1 - \half i\nu_2}  (P_1 \cdot P_1')^{-\alpha+i\theta + \half i\nu_1 - \half i\nu_2} \\
& \hspace{1.5cm}   (P_2 \cdot P_2')^{-\alpha + i\theta -\half i\nu_1 + \half i\nu_2}
(P_1' \cdot P_2')^{-i\theta + \half i\nu_1 + \half i\nu_2}
. 
\end{split}
\een

As in $I_\nu$ \eqref{Idef1}, there can be singularities whenever the two momenta in an inner product coincide, and the same regularization as discussed earlier is understood also here. 

The constant $c_{\nu_1, \nu_2} (\theta)$ is taken to be 
\begin{align}\label{eq:c}
	c_{\nu_1, \nu_2} (\theta) =  \frac{1}{(2\pi)^{d-1}} \frac{\Gamma(\alpha - i\theta - \half i\nu_1 + \half i\nu_2) \Gamma(\alpha - i\theta + \half i\nu_1 - \half i\nu_2) }{\Gamma(i\theta - \half i\nu_1 + \half i\nu_2)
\Gamma(i\theta + \half i\nu_1 - \half i\nu_2)}\,,
\end{align}

and $\te$ is taken to be real. With these definitions, we have

\begin{theorem}\label{theorem:R-dS}
	Let $\nu_1,\nu_2,\nu_3\in\Rl\cup-i(0,\alpha) \cup \{-i\alpha-i{\mathbb N}_0\}$, $\te\in\Rl$ be such that 
	the poles in \eqref{eq:c} are avoided, and $R_\te^{\nu_i\nu_j}:\K_{\nu_i}\ot\K_{\nu_j}\to\K_{\nu_j}\ot\K_{\nu_i}$ the integral operators defined above. Then, $\te,\te'\in\Rl$,
	\begin{enumerate}
		\item[(R1'')] $R_\te$ is unitary.
		\item[(R2'')] $(V_{\nu_2}(\La)\ot V_{\nu_1}(\La)) R^{\nu_1\nu_2}_\theta = R^{\nu_1\nu_2}_\theta(V_{\nu_1}(\La)\otimes V_{\nu_2}(\La))$ for all $\La\in SO^\uparrow(d,1)$. 
		\item[(R3'')] $\bullet$ $F^{\nu_2\nu_1}\,R_\theta^{\nu_1\nu_2}\,F^{\nu_2\nu_1}=R_\te^{\nu_2\nu_1}$ with $F^{\nu_i\nu_j}:\K_{\nu_i}\ot\K_{\nu_j}\to\K_{\nu_j}\ot\K_{\nu_i}$ the tensor flip.\\
		$\bullet$ $(\Gamma_{\nu_2}\ot\Gamma_{\nu_1})\,R_\theta^{\nu_1\nu_2}(\Gamma_{\nu_1}\ot\Gamma_{\nu_2}) =R_{-\theta}^{\nu_1\nu_2}$ 
		\item[(R4'')] On $\K_{\nu_1}\ot\K_{\nu_2}\ot\K_{\nu_3}$,
		\begin{align}\label{YBE-with-nus}
			(R_\te^{\nu_1\nu_2}\otimes 1)(1 \otimes R^{\nu_2\nu_3}_{\theta + \theta'})(R_{\theta'}^{\nu_1\nu_2} \otimes 1)
			= 
			(1 \otimes R_{\theta'}^{\nu_2\nu_3})(R_{\theta+\theta'}^{\nu_1\nu_2}\otimes 1)(1\otimes R_\te^{\nu_2\nu_3}).
		\end{align}
		\item[(R5'')] $(R^{\nu_1\nu_2}_\theta)^{-1}=R^{\nu_2\nu_1}_{-\theta}$.
	\end{enumerate}
 	In particular, whenever $\nu_1=\nu_2=:\nu$, the function $R^{\nu\nu}:\te\mapsto R_\te^{\nu\nu}$ is an invariant Yang-Baxter function, $R^{\nu\nu}\in\R_{\rm fct}(V_\nu, \Gamma_\nu)$. One has the normalization
	\begin{align}\label{eq:R-Normalization}
		R_0^{\nu\nu}=1_{\K_\nu\ot\K_\nu}\,.
	\end{align}
\end{theorem}

The proof of this theorem is given in Appendix~B.

\medskip

If we go to one of the canonical models for the orbital base $B$ (described in Appendix~A), we get concrete formulas for $\omega$ and the kernel $R_\theta$. 
In the case of the flat model, our expression for $R_\theta$ then coincides, up to a phase, with an expression derived previously~\cite{ChicherinDerkachovIsaev:2013} (see also~\cite{DerkachovKorchemskyManashov:2001}) for the 
case of the principal series representation. These authors also proved the Yang-Baxter equation (R4''), and a version of the idempotency relation (R5''). Their formalism for finding $R_\theta$ is based on a different model for the representations. 

\medskip

As explained before, the exponents $w_i$ are fixed by homogeneity requirements in the variables $P_1, P_2, Q_1, Q_2$ up to one remaining free parameter, which is $i\theta$. Setting this parameter to zero (as required for the Yang-Baxter equation (R4)) leads to a trivial solution \eqref{eq:R-Normalization}. We thus conjecture that $\R_{\rm op}(V_\nu,\Gamma_\nu)$ contains only trivial operators (as in the $O(N)$ case, Example~1).

\medskip

There do however exist many other invariant Yang-Baxter functions, because there are two operations we may carry out on the above integral operator without violating the properties (R1'')---(R5''): Scaling of $\te$ and multiplication by suitable $\te$-dependent scalar factors.

\begin{proposition}\label{proposition:tweak-R}
	Let $R^{\nu_1\nu_2}_\te$ be the integral operator defined by the kernel \eqref{Rdef-kernel}, $a\in\Rl$, and $\sigma_{\nu_1\nu_2}\in L^\infty(\Rl,\Cl)$ a function satisfying
	\begin{align}\label{eq:sigma-properties}
		\overline{\sigma_{\nu_1\nu_2}(\te)}
		=
		\sigma_{\nu_1\nu_2}(\te)^{-1}
		=
		\sigma_{\nu_1\nu_2}(-\te)
		=
		\sigma_{\nu_2\nu_1}(-\te)
		\,,\qquad 
		\te\in\Rl\,.
	\end{align}
	Then also $\sigma_{\nu_1\nu_2}(\te)\cdot R^{\nu_1\nu_2}_{a\cdot\te}$ satisfies (R1'')---(R5'').
\end{proposition}
The proof consists in a straightforward check of the conditions (R1'')---(R5''), and is therefore omitted. We conjecture that the operators $\sigma_{\nu_1\nu_2}(\te)\cdot R^{\nu_1\nu_2}_{a\cdot\te}$ form all solutions of the constraints (R1'')---(R5'').

The multipliers $\sigma_{\nu_1\nu_2}$ satisfy exactly the requirements on ``scalar'' Yang-Baxter functions \eqref{eq:scat-func}. In particular, there exist infinitely many functions satisfying the requirements \eqref{eq:sigma-properties}. The freedom of adjusting $R$ by rescaling the argument and multiplying with such a scalar function will be exploited in the next section.

\section{Crossing symmetry and localization}\label{subsection:crossing-YB-ops}

As explained in Section~\ref{subsection:invariant-YB-ops}, an invariant Yang-Baxter operator $R\in\R_{\rm op}(V_1,\Gamma_1)$ (Def.~\ref{definition:Invariant-YB-Operator}) gives rise to an $R$-symmetric Fock space on which twisted second quantized versions $V,\Gamma$ of the representation $V_1$ and the conjugation $\Gamma_1$ act. These ``covariance properties'' are one essential aspect of the Yang-Baxter operators in our setting. The other essential aspect are locality properties, which are linked to specific analyticity requirements on $R$. These analyticity properties, to be described below, have their origin in scattering theory, where they describe the relation between scattering of (charged) particles and their antiparticles \cite{Iagolnitzer:1993}.

\subsection{Crossing-symmetric $R$ and half-local quantum fields}\label{section:crossing}

For the following general discussion, we first consider the conformal sigma models with some arbitrary inner symmetry group $G$ (``Example 2'' of Section~\ref{subsection:invariant-YB-ops}), and later restrict to $G=SO^\uparrow(d,1)$. Our exposition is related to \cite{BostelmannLechnerMorsella:2011}, where a scalar version of such models was presented, and \cite{LechnerSchutzenhofer:2013}, where a massive version with finite-dimensional representation $V_1$ was analyzed. 

We consider on the one hand the representation $U_o$ \eqref{eq:Uo} of the translation-dilation group $\PG_o$ of the lightray on $L^2(\Rl,d\te)$, and the conjugation $(C\psi)(\te)=\overline{\psi(\te)}$ on that space. On the other hand, we consider an arbitrary group $G$, given in a unitary representation $V$ with commuting conjugation $\Gamma$ on a Hilbert space $\K$. Our one-particle space is then $\Hil_1=L^2(\Rl\to\K,d\te)$, with the representation \eqref{eq:UoV} and the conjugation 
\begin{align}
 \Theta:=C\ot\Gamma\,,\quad (\Theta\bxi)(\te)=\Gamma\bxi(\te)\,,\qquad \bxi\in L^2(\Rl\to\K,d\te)\,.
\end{align}

We pick an invariant Yang-Baxter function $R\in\R_{\rm fct}(V_1,\Gamma_1)$ as the essential input into the following construction of quantum fields. These fields will be operators on the $R$-symmetric Fock space $\Hil^R$ over $\Hil_1$ given by the invariant Yang-Baxter operator $\boldsymbol{R}$ defined by $R$ via \eqref{eq:YBF->YBO}.

\medskip

The $R$-symmetric Fock space carries natural creation/annihilation operators: With the help of the projections $P_n^R$, we define as in \cite{Lechner:2003}, $\bxi\in\Hil_1$,
\begin{align}\label{eq:defzdR}
	\zd_R(\bxi)\bPsi_n&:=\sqrt{n+1}\,P_{n+1}^R(\bxi\ot\bPsi_n)\,,\qquad \bPsi_n\in\Hil_R^n\,,\\
	z_R(\bxi)&:=\zd(\bxi)^*\,.
\end{align}
With these definitions, $z_R(\bxi)$ is an annihilation operator (in particular, $z_R(\bxi)\Om=0$), and $\zd_R(\bxi)$ is a creation operator (in particular $\zd_R(\bxi)\Om=\bxi$). It directly follows that, $\bxi\in\Hil_1$, $g\in G$\,,
\begin{align}\label{eq:zdR-Covariance}
	V(g)z_R^\#(\bxi)V(g)^{-1}&=z_R^\#(V_1(g)\bxi)\,,
\end{align}
where $z_R^\#$ denotes either $z_R$ or $\zd_R$.

To introduce a field operator on the lightray, we define\footnote{The prefactor $\pm ie^\te$ in \eqref{eq:fpm} is motivated by the fact that we want to study a chiral field, which has an infrared singularity at zero momentum because of the divergence in the measure $(p^2+m^2)^{-1/2}dp$ for $m=0$. This problem is most easily resolved by passing to the current of this field, which amounts to taking derivatives of test functions. As these derivatives result in the prefactor $\pm ie^\te$, subsequent formulas will be easier if we include this factor from the beginning.} for a test function $f\in\Ss(\Rl)$
\begin{align}\label{eq:fpm}
     f^\pm(\te)
     :=
     \pm ie^\te\int_\Rl du\,e^{\pm iue^\te}\,f(u)
     =
     \pm \sqrt{2\pi}\,ie^\te\cdot\fti(\pm e^\te)
     \,.
\end{align}
In analogy to the Fourier-Helgason transforms \eqref{eq:F-dS-pm}, these functions lie in the representation space $L^2(\Rl,d\te)$, and the definition is covariant under $\PG_o$ in the following sense: 
\begin{align}
     f_{(x,\la)}^\pm
     =
     U_o(\pm x,\la)f^\pm\,,\quad 
     f_{(x,\la)}(u):=f(e^\la(u-x))
     \,.
\end{align}
Analogously, the TCP transformed function $\overline{f_-}:u\mapsto\overline{f(-u)}$ yields $(\overline{f_-})^\pm=-\overline{f^\pm}$, where the minus sign is due to the fact that we consider the current.

Given any vector $k\in\K$, we define the field operators (cf.~ \cite{SchroerWiesbrock:2000-1,Lechner:2003})
\begin{align}\label{eq:def-phiR}
	\phi_{R,k}(f)
	:=
	\zd_R(f^+\ot k)+z_R(\overline{f^-}\ot\Gamma k)\,,\qquad f\in\Ss(\Rl)\,.
\end{align}
We may think of the vector $k$ as a label for the different ``components'' $\phi_{R,k}$ of the field $\phi_R$. Note, however, that $\K$ can be infinite-dimensional, so that $\phi_R$ can be a field with infinitely many independent components.

By proceeding to delta distributions $\delta_u$, sharply localized at a point $u$ on the lightray, we may also describe this field in terms of the distributions $\phi_{R,k}(u):=\phi_{R,k}(\delta_u)$. 

\medskip

In the present general setting, one can show that the field operators transform covariantly under $U_o$ and $V$, but not under the TCP operator $\Theta$. Furthermore, $\phi_{R,k}(f)^*=\phi_{R,\Gamma k}(\overline{f})$ on an appropriate domain. We do not repeat the calculations from \cite{LechnerSchutzenhofer:2013} here (see, however, Section~\ref{section:qft} for a concrete version of these properties), but rather focus on the locality aspects. 

To begin with, one realizes that $\phi_R$ is a non-local field unless $R=F$, in which case it satisfies canonical commutation relation and reduces to a free field. For general $R$, the locality properties of $\phi_R$ are best analyzed by introducing a second ``TCP conjugate'' field,
\begin{align}\label{eq:tcp-field-general}
	\phi_{R,k}'(u):=\Theta\phi_{R,\Gamma k}(-u)\Theta
	\,,\qquad u\in\Rl\,,\;k\in\K\,,
\end{align}
which in its smeared version is, $f\in\Ss(\Rl)$,
\begin{align}
     \phi_{R,k}'(f)
     :=
     \Theta\phi_{R,\Gamma k}(\overline{f_-})\Theta
     =
     -\Theta\zd_R(\Theta(f^+\ot k))\Theta-\Theta z_R(\Theta(\overline{f^-}\ot\Gamma k))\Theta
     \,.
\end{align}
This field shares many properties with $\phi_R$, and in particular, also transforms covariantly under $U_o$ and $V$.

We next want to analyze the commutation relations between $\phi_R$ and $\phi_R'$. To this end, one first computes that the ``TCP conjugate creation operator'' acts on $\bPsi_n\in\Hil_n^R$ according to
\begin{align}\label{eq:zd'}
	\Theta\zd_R(\bxi)\Theta\, \bPsi_n
	=
	\sqrt{n+1}\,P_{n+1}^R(\bPsi_n\ot \Theta\bxi)\,,\qquad\bxi\in\Hil_1\,,
\end{align}
i.e. in comparison to the left action \eqref{eq:defzdR} of $\zd_R(\bxi)$, this operator ``creates from the right''. In particular, $\bxi,\bxi'\in\Hil_1$,
\begin{align}\label{eq:zz'comm}
	[\zd_R(\bxi),\,\Theta\zd_R(\bxi')\Theta]&=0\,,\qquad
	[z_R(\bxi),\,\Theta z_R(\bxi')\Theta]=0\,.
\end{align}
To control the mixed commutators $[\zd_R,\Theta z_R \Theta]$, the following definition is essential.

\begin{definition}\label{definition:Crossing-Symmetric-YB-Operator}
     An invariant Yang-Baxter function $R\in\R_{\rm fct}(V,\Gamma)$ (for some group $G$, on some Hilbert space $\K$) is called {\em crossing-symmetric}\footnote{The term {\em crossing-symmetric} is usually reserved to mean just property (R7'). We have included here (R6') for a concise definition; for our work the combination of (R6') and (R7') is the relevant property. The analyticity condition (R6') can be understood as ruling out bound state poles.} if it satisfies the following two conditions.
	\begin{itemize}
		\item[(R6')] $\te\mapsto R(\te)$ extends to a bounded analytic function on the strip $0<\im\te<\pi$.
		\item[(R7')] Crossing symmetry holds: For $\te\in\Rl$,
		\begin{align}
			\langle\xi\ot\psi,\,R(i\pi-\te)\,(\varphi\ot\xi')\rangle_{\K\ot\K}
			=
			\langle\psi\ot\Gamma\xi',\,R(\te)\,(\Gamma\xi\ot\varphi)\rangle_{\K\ot\K}
			\,.
		\end{align}
	\end{itemize}
\end{definition}

In case that $\K\cong\Cl^N$ is finite-dimensional, these requirements coincide with the corresponding ones in \cite[Def.~2.1]{LechnerSchutzenhofer:2013}, where the conjugation was taken as $(\Gamma \zeta)_k:=\overline{\zeta_{\overline{k}}}$, $\zeta\in\Cl^N$, with the components $\zeta_k$ referring to a fixed basis and $k\mapsto\overline{k}$ an involutive permutation of $\{1,\dots,N\}$.

The significance of (R6')---(R7') is best explained in terms of Tomita-Takesaki modular theory (see, for example, \cite{KadisonRingrose:1986}), as we shall do now. Note, however, that in the explicit examples to be considered in Section~\ref{section:qft}, we will also give a purely field-theoretic formulation (Theorem~\ref{theorem:commutation-concrete}).

\bigskip

For the following argument, we adopt the concept of ``modular localization'' \cite{BrunettiGuidoLongo:2002}. The main idea is to anticipate a quantum field theory, defined in terms of a system of local von Neumann algebras \cite{Haag:1996}, and the connection between the modular data of the algebra corresponding to the half line $\Rl_+$ and the one-parameter group of dilations $u\mapsto e^{-2\pi\la}u$ which leaves $\Rl_+$ invariant, and the reflection $u\mapsto-u$, which flips $\Rl_+$ into $\Rl_-$ (``Bisognano-Wichmann property'', \cite{BisognanoWichmann:1976}). We define
\begin{align}\label{eq:Delta}
     \Delta^{it}:=U_o(0,-2\pi t)\,,
\end{align}
and view this as either an operator on $L^2(\Rl,d\te)$, or on $\Hil_1\cong L^2(\Rl,d\te)\ot\K$, where it just acts on the left factor. 

If $f\in\Ss(\Rl)$ is supported on the right, $\supp f\subset\Rl_+$, the function $f^+$ \eqref{eq:fpm} has an $L^2$-bounded analytic continuation to the strip $S_\pi:=\{\te\in\Cl\,:\,0<\im(\te)<\pi\}$, with $f^+(\te+i\pi)=f^-(\te)$, $\te\in\Rl$. For $g$ supported on the left instead, $\supp g\subset\Rl_-$, the same properties hold for $g^-$ \eqref{eq:fpm} \cite[Lemma~4.1]{BostelmannLechnerMorsella:2011}. In terms of the operator $\Delta$ \eqref{eq:Delta}, this means
\begin{align*}
	f^+\in\dom\Delta^{-1/2},\;\Delta^{-1/2}f^+=f^-\quad\text{for}\;\supp f\subset\Rl_+
	\,,\\
	g^+\in\dom\Delta^{1/2},\;\Delta^{1/2}g^+=g^-\quad\text{for}\;\supp g\subset\Rl_-\,.
\end{align*}
The conditions (R6')---(R7') imply that $R$ has matching analyticity properties. Namely, the matrix elements of $(1\ot\Delta^{it})\boldsymbol{R}(\Delta^{-it}\ot1)$ are,
$\bpsi,\bpsi',\bvarphi,\bvarphi'\in\Hil_1$,
	\begin{align*}
		\langle\bpsi'\ot\bpsi,(1\ot\Delta^{it})\boldsymbol{R}(\Delta^{-it}\ot1)\bvarphi\ot\bvarphi'\rangle
		=
		\int d^2\te\langle\bpsi'(\te_1)\ot\bpsi(\te_2),\,R_{\te_1-\te_2-2\pi t} \bvarphi(\te_2)\ot\bvarphi'(\te_1)\rangle_{\K\tp{2}}
		\,,
	\end{align*}
	and therefore analytically continue to $-\frac{1}{2}<\im(t)<0$ in view of (R6'), with the boundary value
	\begin{align}
		\langle\bpsi'\ot\bpsi,(1\ot\Delta^{1/2})\boldsymbol{R}(\Delta^{-1/2}\ot1)&\bvarphi\ot\bvarphi'\rangle
		=
		\int d^2\te\langle\bpsi'(\te_1)\ot\bpsi(\te_2),\,R_{\te_1-\te_2+i\pi}\bvarphi(\te_2)\ot\bvarphi'(\te_1)\rangle_{\K\tp{2}}
		\nonumber
		\\
		&=
		\int d^2\te\langle\bpsi(\te_2)\ot\Gamma\bvarphi'(\te_1),\,R_{\te_2-\te_1}\Gamma\bpsi'(\te_1)\ot\bvarphi(\te_2)\rangle_{\K\tp{2}}
		\nonumber
		\\
		&=
		\langle\bpsi\ot\Theta\bvarphi',\,\boldsymbol{R}(\Theta\bpsi'\ot\bvarphi)\rangle
		\,.
		\label{eq:full-crossing}
	\end{align}
	These relations imply the following theorem.

\begin{theorem}\label{theorem:commutation}
	Let $k,k'\in\K$. Then, in the sense of distributions,   
	\begin{align}
		[\phi_{R,k}(u),\,\phi'_{R,k'}(u')\,]=0\quad\text{for}\quad u<u'\,, 
	\end{align}
	on the space of vectors of finite Fock particle number.
\end{theorem}
\begin{proof}
	The proof follows the same strategy as \cite{Lechner:2003}, generalized in \cite{LechnerSchutzenhofer:2013, BischoffTanimoto:2013}. We introduce as a shorthand the vector-valued functions $\bof:=f\ot k$, $\bg:=g\ot k'$, with field operators $\phi_R(\bg)=\phi_{R,k}(g)$ and $\phi_R'(\bof)=\phi_{R,k'}'(f)$. Here $f,g\in\Ss(\Rl)$ are scalar test functions, and to prove the theorem, we have to demonstrate $[\phi_R(\bg),\phi_R'(\bof)]=0$ for $\supp f\subset\Rl_+$, $\supp g\subset\Rl_-$.
	
	We pick $\bpsi,\bvarphi\in\Hil_1$ and compute using \eqref{eq:zz'comm}, \eqref{eq:zd'}, and the definitions
	\begin{align}
		\langle\bpsi,\,[\phi_R(\bg),\,\phi'_R(\bof)\,]\,\bvarphi\rangle
		&=
		\langle\bpsi,\,\left([\zd_R(\bg^+),\Theta z_R(\bof^-)\Theta]+[z_R(\Theta\bg^-),\Theta\zd_R(\Theta\bof^+)\Theta]\right)\,\bvarphi\rangle
		\nonumber
		\\
		&=
		\langle\langle\bg^+,\bpsi\rangle\Om,\Theta\langle \bof^-,\Theta\bvarphi\rangle\Om\rangle
		-2\langle P_2^R(\bpsi\ot \Theta\bof^-),P_2^R(\bg^+\ot\bvarphi)\rangle
		\nonumber
		\\
		&\quad+
		2\langle P_2^R(\Theta\bg^-\ot\bpsi),P_2^R(\bvarphi\ot \bof^+)\rangle
		-\langle \Theta\langle \Theta\bof^+,\Theta\bpsi\rangle\Om,\langle \Theta\bg^-,\bvarphi\rangle\Om\rangle
		\nonumber
		\\
		&=
		\langle \Theta\bg^-\ot\bpsi,\,\boldsymbol{R}(\bvarphi\ot \bof^+)\rangle
		-
		\langle\bpsi\ot \Theta\bof^-,\,\boldsymbol{R}(\bg^+\ot\bvarphi)\rangle
		\,.
		\label{eq:mixed-comm}
	\end{align}
	To show that these terms cancel in case $R$ satisfies (R6') and (R7'), we consider the first term in \eqref{eq:mixed-comm}, and insert an identity $1=(\Delta^{it}\ot\Delta^{it})(\Delta^{-it}\ot\Delta^{-it})$, $t\in\Rl$, in front of $\boldsymbol{R}$. In view of the invariance (R2) of $\boldsymbol{R}$, we then see that this scalar product equals
	\begin{align}\label{eq:asas}
		\langle\Delta^{-it-1/2}\Theta\bg^+\ot\bpsi, \,(1\ot\Delta^{it})\boldsymbol{R}(\Delta^{-it}\ot1)\;(\bvarphi\ot\Delta^{-it}\bof^+)\rangle
		\,.
	\end{align}
	The vectors in the left and right hand entry of the scalar product are analytic in the strip $-\frac{1}{2}<\im(t)<0$ (taking into account the antilinearity of the left factor, and the fact that $\bof^+\in\dom\Delta^{-1/2}$ in the right factor). As explained above, the same analyticity holds for the operator-valued function $t\mapsto(1\ot\Delta^{it})\boldsymbol{R}(\Delta^{-it}\ot1)$. Taking into account the boundary values $\Delta^{-1/2}\bof^+=\bof^-$ and \eqref{eq:full-crossing}, we see that \eqref{eq:asas} coincides with
	\begin{align*}
		\langle\Theta\bg^+\ot\bpsi, \,(1\ot\Delta^{1/2})\boldsymbol{R}(\Delta^{-1/2}\ot1)\,(\bvarphi\ot \bof^-)\rangle
		=
		\langle\bpsi\ot\Theta\bof^-,\,\boldsymbol{R}(\bg^+\ot\bvarphi)\rangle\,,
	\end{align*}
	which is identical to the second term in \eqref{eq:mixed-comm}. Thus the matrix elements of $[\phi_R(\bof),\phi'_R(\bg)]$ between single particle states vanish. 
	
	Using the same arguments as in \cite{LechnerSchutzenhofer:2013, BischoffTanimoto:2013}, one shows analogously that matrix elements between arbitrary vectors of finite particle number vanish.	
\end{proof}

We may interpret this commutation theorem by regarding $\phi_{R,k}(u)$ as localized in the left halfline $(-\infty,u)$, and $\phi'_{R,k'}(u')$ as localized in the right halfline $(u',\infty)$. This interpretation is consistent with both, covariance and locality. The fields $\phi_R$, $\phi_R'$ should however not be regarded as the ``physical'' quantum fields of the model, but rather as auxiliary objects \cite{BorchersBuchholzSchroer:2001}. To proceed to the physical observables/fields, localized in finite intervals on the lightray, it is helpful to use an operator-algebraic setting.

\bigskip

By construction, the field operators satisfy (on a suitable domain) $\phi_{R,k}(u)^*=\phi_{R,k}(u)$ if $\Gamma k=k$. One can show by the same method as in \cite{LechnerSchutzenhofer:2013} that $\phi_{R,k}(f)$ is then essentially selfadjoint. In this case, we can form the unitaries $e^{i\phi_{R,k}(f)}$ by the functional calculus, and pass to the generated von Neumann algebra
\begin{align}\label{eq:M}
	\M_R
	:=
	\{e^{i\phi_{R,k}(f)}\,:\,\overline{f}\ot \Gamma k=f\ot k\,,\supp f\subset\Rl_-\}''\subset\B(\Hil^R)\,.
\end{align}

To conclude the present general section, we point out a few further properties of the algebra $\M_R$, including Haag duality with the algebra generated by the reflected field.

\begin{proposition}
	\begin{enumerate}
		\item The Fock vacuum $\Om$ is cyclic and separating for $\M_R$. 
		\item The modular conjugation of $(\M_R,\Om)$ is $J=\Theta$, and the modular operator is the previously defined $\Delta$.
		\item The commutant of $\M_R$ is
		\begin{align}
			\M_R'
			=
			\Theta\M_R \Theta
			=
			\{e^{i\phi'_{R,k}(f)}\,:\,\overline{f}\ot \Gamma k=f\ot k,\;\supp f\subset\Rl_+\}''
			\,.
		\end{align}
	\end{enumerate}
\end{proposition}
\begin{proof}
	$a)$ It follows by standard arguments that $\Om$ is cyclic for $\M_R$ and $\tilde\M_R:=\{e^{i\phi'_{R,k}(f)}\,:\,\overline{f}\ot \Gamma k=f\ot k,\;\supp f\subset\Rl_+\}''$ \cite{LechnerSchutzenhofer:2013}. As in \cite{Lechner:2003}, one shows by an analytic vector argument that the unitaries $\exp(i\phi_{R,k}(g))$ and $\exp(i\phi'_{R,k'}(f))$ commute for $\supp f\subset\Rl_+$, $\supp g\subset\Rl_-$. Hence $\tilde\M_R\subset\M_R'$, and $\Om$ is also separating for $\M_R$ (and $\tilde\M_R$).
	
	$c)$ is a straightforward consequence of $b)$ by Tomita-Takesaki theory and the definitions. The proof of $b)$ follows by the same line of arguments as in \cite{BuchholzLechner:2004, Alazzawi:2014}.
\end{proof}

\subsection{Crossing-symmetric Yang-Baxter functions for $\boldsymbol{SO^\uparrow(d,1)}$}

In this section we construct examples of crossing symmetric Yang-Baxter functions $R\in\R_{\rm fct}(V_\nu,\Gamma_\nu)$ for the $SO^\uparrow(d,1)$ representations and conjugations $V_\nu$, $\Gamma_\nu$ from Section~\ref{section:YBOps-LorentzGroup}. This will be done by exploiting the freedom to adjust our basic $SO^\uparrow(d,1)$-invariant Yang-Baxter function $R$ \eqref{Rdef} by scaling the parameter $\te$ and multiplying by a suitable function of $\te$ (Proposition~\ref{proposition:tweak-R}). We define
\begin{align}\label{eq:R-scaled}
	\Rti_\te^{\nu_1\nu_2}:=\sigma_{\nu_1\nu_2}(-\tfrac{\alpha}{\pi}\,\te)\cdot R^{\nu_1\nu_2}_{-\frac{\alpha}{\pi}\te}
	\,.
\end{align}
where $\alpha=\frac{d-1}{2}$ as before, and 
\begin{align}\label{eq:sigma-main}
	\sigma_{\nu_1\nu_2}(\te)
	&:=
	\frac{\Gamma(\alpha+\epsilon-i\te)^2}{\Gamma(\alpha+\epsilon+i\te)^2} \prod_{n=0}^\infty \bigg\{
	\frac{\Gamma(i\theta + \alpha(2n+1)-\epsilon)^2 \Gamma(-i\theta + \alpha(2n-1)-\epsilon)^2}{\Gamma(-i\theta+\alpha(2n+1)-\epsilon)^2 \Gamma(i\theta + \alpha(2n-1)-\epsilon)^2} \nonumber \\
	& \quad
	\cdot \prod_{p,q=\pm} \frac{\Gamma(i\theta+\alpha(2n+1) - p \nu_{12}^q) \Gamma(-i\theta+\alpha(2n+2) - p \nu_{12}^q)}{\Gamma(-i\theta+\alpha(2n+1) - p \nu_{12}^q) \Gamma(i\theta+\alpha(2n+2) - p \nu_{12}^q)}  \bigg\} ,
\end{align}
with $\nu_{12}^\pm:=\frac{i}{2}(\nu_1\pm\nu_2)$. $\epsilon$ is a real parameter chosen below depending on the representations. It is seen that the infinite 
product converges absolutely in each case considered. 
It is shown in Appendix~B that $\sigma_{\nu_1\nu_2}$ satisfies the requirements of Prop.~\ref{proposition:tweak-R}.

\begin{theorem}\label{theorem:crossing}
	Let $\nu_1,\nu_2\in\Rl$ label two principal series representations, $\te\in\Rl$, and $\Rti_\te^{\nu_1\nu_2}$ the integral operators defined above
	with $\epsilon=0$ in \eqref{eq:sigma-main}. Then 
	\begin{itemize}
	\item[(R6'')] $\te\mapsto\Rti_\te^{\nu_1\nu_2}$ extends to an analytic bounded function on the strip $0<\im(\te)<\pi$, and 
		\item[(R7'')] Crossing symmetry holds: If $\psi_1,\psi_1'\in\K_{\nu_1}$, $\xi_2,\xi_2'\in\K_{\nu_2}$, then 
		\begin{align}
			\langle\xi_2\ot\psi_1,\,\Rti^{\nu_1\nu_2}_{i\pi-\te}\,(\psi_1'\ot\xi'_2)\rangle
			=
			\langle\psi_1\ot\Gamma_{\nu_2}\xi_2',\,\Rti_{\te}^{\nu_2\nu_1}(\Gamma_{\nu_2}\xi_2\ot\psi_1')\rangle
			\,.
		\end{align}
	\end{itemize}
The same holds true if $\nu_1=\nu_2=\nu$ belong to a complementary or discrete series representation, and 
we set $\epsilon=i\nu$ in \eqref{eq:sigma-main}. Thus for all principal and complementary series representations, $\Rti^{\nu\nu}$ is crossing-symmetric in the sense of Def.~\ref{definition:Crossing-Symmetric-YB-Operator}.
\end{theorem}

In Appendix~B, we also provide an alternative integral representation of the 
factor $\sigma_{\nu_1, \nu_2}(\theta)$. In case we allow for arbitrary combinations of principal, complementary and discrete 
series representations $\nu_1,\nu_2$, the (rescaled) Gamma factors in \eqref{eq:c} can produce poles in the strip $0<\im(\te)<\pi$, 
which is the reason for our restriction to two principal series representations, or a single complementary resp. discrete series representation.

\bigskip

The function $\sigma_{\nu_1\nu_2}$ \eqref{eq:sigma-main} is precisely constructed in such a way that the crossing relation (R7'') holds\footnote{
This function may be viewed as a solution to a cocycle problem. A related, but not identical, problem of this general nature is treated in~\cite{Schmudgen:1995}.
}. There remains however a large freedom to modify $\sigma_{\nu_1\nu_2}$ without violating the conditions (R1'')--(R7''). In fact, we may multiply $\Rti^{\nu_1\nu_2}_\te$ by another function $\rho_{\nu_1\nu_2}(\te)=\rho_{\nu_2\nu_1}(\te)$, which is analytic and bounded on the strip $0<\im(\te)<\pi$, and has the symmetry properties
\begin{align}
	\rho_{\nu_1\nu_2}(i\pi+\te)
	=
	\rho_{\nu_1\nu_2}(-\te)
	=
	\overline{\rho_{\nu_1\nu_2}(\te)}
	=
	\rho_{\nu_1\nu_2}(\te)^{-1}\,,\qquad\te\in\Rl\,.
\end{align}
There exist infinitely many of such ``scattering functions'', they are given by all inner functions $\rho$ of the strip $0<\im(\te)<\pi$ with the symmetry properties $\rho(i\pi-\te)=\rho(\te)=\rho(-\te)^{-1}$. Using the canonical factorization of inner functions (see, for example, \cite{Duren:1970}), one can then write down explicit formulas for $\rho$. In particular, if $\rho$ contains no singular part, it has the form
\begin{align}
	\rho_{\nu_1\nu_2}(\te)
	=
	\pm e^{i\kappa\sinh\te}\prod_l\frac{\sinh\te-i\sinh b_l}{\sinh\te+i\sinh b_l}
	\,,
\end{align}
where $\kappa\geq0$, and the zeros $b_l$ have to satisfy $0<{\rm Im} b_k<\pi$ and certain symmetry and summability conditions to ensure $\rho(i\pi-\te)=\rho(\te)=\rho(-\te)^{-1}$ and convergence of the product \cite{Lechner:2006}.

\section{QFTs from $\boldsymbol{SO^\uparrow(d,1)}$-invariant Yang-Baxter functions}\label{section:qft}

We now use the crossing symmetric $SO^\uparrow(d,1)$-invariant Yang-Baxter function $\Rti$ \eqref{eq:R-scaled} to build concrete quantum field theoretic models. 
A first class of models will be constructed within the general setup of Section~\ref{section:crossing}, where now the group is taken as $G=SO^\uparrow(d,1)$, and we restrict ourselves to a principal or complementary series representation $V_\nu$.

In a second section, we show that by a variant of this construction, one also gets Euclidean conformal field theories in $(d-1)$ dimensions.

\subsection{CFTs with target de Sitter spacetime}\label{subsection:chiral-models}

Our first family of models is a concrete version of the abstract field operators in Section~\ref{section:crossing}. The symmetry group is here the direct product $\PG_o\times SO^\uparrow(d,1)$ of the translation-dilation group (acting on lightray coordinates), and the Lorentz group (acting on de Sitter coordinates). Group elements will be denoted as $(x,\la,\La)=(x,\la)\times\La$ in an obvious notation. 

This group is represented on $\Hil_1:=L^2(\Rl,d\te)\ot\K_\nu\cong L^2(\Rl\to\K_\nu,d\te)$ by the representation $V_1:=U_o\ot V_\nu$ as in Section~\ref{section:crossing}, where $V_\nu$ may belong to the principal or complementary series. That is, the single particle vectors are scalar functions $\psi:\Rl\times \dS_d\to\Cl$ depending on two ``momentum coordinates'', $\te\in \Rl$ and $P\in C_d^+$, and 
\begin{align}
	(V_1(x,\la,\La)\psi)(\te,P)=e^{ix\exp(\te)}\,\psi(\te-\la,\La^{-1}P)\,.
\end{align}
As TCP operator on this space, we take $\Theta_\nu:=\gamma_\nu\cdot C\ot\Gamma_\nu$, where $C$ denotes complex conjugation, $\Gamma_\nu$ is defined in \eqref{Thdef} and $\gamma_\nu$ is the phase factor \eqref{eq:gamma-factor}.

We then consider the $SO^\uparrow(d,1)$-invariant crossing-symmetric Yang-Baxter function $R=\Rti^{\nu\nu}\in\R_{\rm fct}(V_\nu,\Gamma_\nu)$ \eqref{eq:R-scaled}, which defines an $(\PG_0\times SO^\uparrow(d,1))$-invariant Yang Baxter operator $\boldsymbol{R}\in\R_{\rm op}(U_o\ot V_\nu,\Theta_\nu)$ as in Section~\ref{subsection:crossing-YB-ops}. In order not to overburden our notation, we have dropped the tilde from $R$, but still mean {\em the rescaled and multiplied version from} \eqref{eq:R-scaled}.

\bigskip

Following the general construction, we then obtain the $R$-symmetric Fock space $\Hil^R$. An $n$-particle vector $\Psi_n\in\Hil_n^R$ is here given by a function of $(2n)$ momentum space variables, namely $\Psi_n(\te_1,P_1,\ldots, \te_n,P_n)$, which is square integrable in each $\theta_j$ and homogeneous of 
degree $-\alpha - i\nu$ in each $P_j$ (as well as square integrable on any orbital base $B$ of the cone $C_d^+$). The $R$-symmetry is expressed by the equations
\begin{align}\label{eq:R-Symmetry-explicit}
	&\Psi_n(\te_1,P_1,\ldots,\te_n,P_n)
	\\
	&=
	\int\limits_{P_j', P_{j+1}'}
	R_{\te_{j}-\te_{j+1}}(P_j,P_{j+1};\,P'_{j+1},P_j')\,\Psi_n(\te_1,P_1,..,\te_{j+1},P_{j+1}',\te_j,P_{j}',..,\te_n,P_n)\,,
	\nonumber
\end{align}
to be satisfied for each $1\leq j\leq n-1$. Here, and in the rest of the section, we will usually use the shorthand
$\int_B \omega(P) = \int_P$ for an arbitrary orbital base $B$ to simplify the notation. 

\medskip

The representation $V$ of $\PG_0\times SO^\uparrow(d,1)$ takes explicitly the form, $\Psi_n\in\bHil_n^R$,
\begin{align*}
	(V(x,\la,\La)\Psi_n)(\te_1,P_1,...,\te_n,P_n)
	=
	\exp[ix(e^{\te_1}+...+e^{\te_n})] \cdot
	\Psi_n(\te_1-\la,\La^{-1} P_1,...,\te_n-\la,\La^{-1} P_n)
	\,.
\end{align*}
The explicit form of the TCP symmetry $\Theta$ is different for the two series: We have
\begin{align*}
	&(\Theta \Psi_n)(\te_1,P_1,...,\te_n,P_n)=
        \gamma_\nu^n\cdot\overline{\Psi_n(\te_n,P_n,...,\te_1,P_1)}
	\qquad \text{(complementary series)} \\
	&(\Theta\Psi_n)(\te_1,P_1,...,\te_n,P_n) = c_\nu^n \gamma_\nu^n \cdot
	\int\limits_{\{P_j'\}} 
	\prod_{j=1}^n (P_j \cdot P_j')^{-\alpha - i\nu}\, \overline{\Psi_n(\te_n,P_n',...,\te_1,P_1')} \\
	& \qquad\qquad\qquad\qquad\qquad\qquad\qquad \text{ (principal series) }
\end{align*}
with $c_\nu=(2\pi)^\alpha \Gamma(\alpha -i\nu)/\Gamma(i\nu)$.
For the creation/annihilation operators, we write informally, $\bxi\in\Hil_1$,
\ben
z^\#_R(\bxi)
=
\int_{\RR\times B}  \dd\te\wedge\omega(P)\, z^\#_R(\te,P) \,\bxi^\#(\te,P) \ , 
\een
where, informally, we take $z_R^\dagger (\te,\lambda\,P) = \lambda^{-\alpha + i\nu} z_R^\dagger (\te,P)$, so that the integral informally is independent of the choice of $B$, by the argument based on Lemma~\ref{lemma:closedform}.

The transformation law \eqref{eq:zdR-Covariance} then takes the form
\begin{align}\label{eq:Covariance-zR}
	V(x,\la,\La)z_R^\#(\te,P)V(x,\la,\La)^{-1}
	&=
	e^{\pm ix\,\exp(\te+\la)}\,z_R^\#(\te+\la,\La P)
	\,,
\end{align}
where the ``$+$'' sign is used for the creation operator $\zd_R$, and the ``$-$'' sign for the annihilation operator $z_R$.

The commutation relations are obtained in this informal but efficient notation as 
\bena\label{zccr}
	z_R(\theta_1, P_1) z_R(\theta_2, P_2) - R_{\theta_1-\theta_2} \ z_R(\theta_2, P_2) z_R(\theta_1, P_1) &=& 0\,,
	\\
	z_R^\dagger(\theta_1, P_1) z_R^\dagger (\theta_2, P_2) - R_{\theta_1-\theta_2} \ z_R^\dagger(\theta_2, P_2) z_R^\dagger(\theta_1, P_1) &=& 0\,,
\eena
with $R_\te$ acting as in \eqref{Rdef}. The form of the mixed commutation relation differs according to whether we 
are in the principal series case ($\nu \in \RR$), the complementary ($i\nu \in (0, \alpha)$) or discrete series case ($i\nu \in \alpha + \mathbb{N}_0$). In the first case we have
\ben
z_R(P_1, \theta_1) z_R^\dagger(P_2, \theta_2) - R_{\theta_2-\theta_1} \ z_R^\dagger(P_2, \theta_2) z_R(P_1, \theta_1) = \delta(\theta_1-\theta_2)\delta(P_1,P_2) \cdot 1
\een
where $\delta(P_1, P_2)$ is the Dirac delta function on $B$ (relative to the integration measure $\omega$) when we integrate this identity against 
smooth test functions on any orbital base $B$. In the case of the complementary series we have instead
\ben
z_R(P_1, \theta_1) z_R^\dagger(P_2, \theta_2) - R_{\theta_2-\theta_1} \ z_R^\dagger(P_2, \theta_2) z_R(P_1, \theta_1)
=  c_\nu \delta(\theta_1 - \theta_2) (P_1 \cdot P_2)^{-\alpha+i\nu} \cdot 1.
\een
In the case of the discrete series, we have
\ben
z_R(P_1, \theta_1) z_R^\dagger(P_2, \theta_2) - R_{\theta_2-\theta_1} \ z_R^\dagger(P_2, \theta_2) z_R(P_1, \theta_1)
=  c_n \delta(\theta_1 - \theta_2) (P_1 \cdot P_2)^n \log(P_1 \cdot P_2) \cdot 1
\een
where $c_n = (2\pi)^{-\alpha} (-1)^{n+1}/n!\Gamma(\alpha+n)$.
The differences arise from the 
differences in the definition of the scalar product in each case, see \eqref{inn},\eqref{Idef}, respectively \eqref{Kdef}. All these exchange relations are generalizations of the Zamolodchikov--Faddeev algebra \cite{ZamolodchikovZamolodchikov:1979,Faddeev:1984}.

\bigskip

We next describe a concrete version of the quantum field $\phi_R$ from Section~\ref{section:crossing}, replacing the ``components'' $\phi_{R,k}$ by an additional dependence on a de Sitter variable. This is done by replacing the vector $k$ by a Fourier-Helgason transform $F^+_\nu$ \eqref{eq:F-dS-pm} of some test function $F\in C_0^\infty(\dS_d)$. In view of Lemma~\ref{lemma:fourier-helgason}~$c)$, this amounts to the field operator, $f\in\Ss(\Rl)$, $F\in C_0^\infty(\dS_d)$,
\begin{align}\label{eq:phiR-smeared}
	\phi_R(f\ot F)
	=
	\zd_R(f^+\ot F^+_\nu)
	+
	z_R(\overline{f^-}\ot(\overline{F_-})^+_\nu)\,,
\end{align}
where $f^\pm$ are defined in \eqref{eq:fpm}, and $F^\pm_\nu$ in \eqref{eq:F-dS-pm}. 

In the following, it will be convenient to describe the field operator in terms of its distributional kernels, writing $\phi_R(f \ot F) = \int d u d\mu(X) \phi_R(u,X) f(u) F(X)$. Then we have, informally ($u \in \RR, X \in \dS_d$), 
\begin{align}\label{eq:phiR-concrete}
	\phi_R(u,X)
	&=
	\int_{\theta, P} \left\{ie^\te e^{iue^\te}\,(X\cdot P)_+^{-\alpha-i\nu}\cdot\zd_R(\te,P) 
	-ie^\te e^{-iue^\te}\,(-X\cdot P)_-^{-\alpha+i\overline{\nu}}\cdot z_R(\te,P)
	\right\}
	\,.
\end{align}
Here, the $\theta$-integral is over $\RR$, and $\int_P \equiv \int_B \omega(P)$ for an arbitrary orbital base $B$, as before. 
We collect a few properties of this field in the following proposition.

\begin{proposition}\label{proposition:phiR-props}
	The field $\phi_R$ \eqref{eq:phiR-concrete} is an operator-valued distribution on $\Ss(\Rl)\times C_0^\infty({\rm dS_d})$ with the following properties.
	\begin{enumerate}
		\item The field is neutral in the sense that $\phi_R(u,X)^*=\phi_R(u,-X)$ (on an appropriate domain).
		\item The field is $(\PG_o\times SO^\uparrow(d,1))$-covariant, i.e. 
		\begin{align}
			V(x,\la,\La) \phi_R(u,X) V(x,\la,\La)^{-1}
			=
			e^{-\la}\,\phi_R(e^{-\la}u+x,\La X)\,.
		\end{align}
		\item The field solves the Klein-Gordon equation of mass $m^2=\alpha^2+\nu^2$ on de Sitter space,
		\begin{align}
			(\square_X + m^2)\phi_R(u,X) = 0 \,. 
		\end{align}
	\end{enumerate}
\end{proposition}
\begin{proof}
	$a)$ is evident from \eqref{eq:phiR-smeared}, \eqref{eq:phiR-concrete}, and $b)$ is a consequence of \eqref{eq:Covariance-zR}. Note that the prefactor $e^{-\la}$ is due to the fact that we consider the current. $c)$ is satisfied because the de Sitter waves $u_P^\pm$ are solutions of the Klein-Gordon equation. 
\end{proof}

The field $\phi_R$ is seen to fail the usual condition of Einstein causality in both, its lightray and its de Sitter coordinate,  due to the presence of the $R$-factors in~\eqref{zccr}.

So in this sense $\phi_R$ does not, by itself, straightforwardly define a local quantum field neither on the lightray nor on de Sitter space. However, the field satisfies a kind of remnant of the locality condition in the lightray variable. As explained in the abstract setting in Section~\ref{section:crossing}, this is best understood in interplay with its TCP reflected partner field.

The definition \eqref{eq:tcp-field-general} translates here to
\begin{align}\label{eq:phi'R}
	\phi_R'(u,X)
	:=
	\Theta\phi_R(-u,-X)\Theta\,.
\end{align}
It is clear from this definition that also $\phi_R'$ has the properties $a)$---$c)$ of Proposition~\ref{proposition:phiR-props}. Explicitly, we have from \eqref{eq:phiR-concrete}
\begin{align}\label{eq:phi'R-concrete}
	\phi_R'(u,X)
	&=
	\int_{\theta, P} \left\{-ie^\te e^{iue^\te}\,(-X\cdot P)_-^{-\alpha+i\overline{\nu}}\cdot z_R^{\prime \dagger}(\te,P) 
	+ie^\te e^{-iue^\te}\,(X\cdot P)_+^{-\alpha-i\nu}\cdot z_R'(\te,P)
	\right\}
	\, .
\end{align}
This is different from $\phi_R(u,X)$ because $z_R' = \Theta z_R \Theta \neq z_R$, i.e. the Zamolodchikov operators do not transform covariantly under $\Theta$.

We have the following concrete version of Theorem~\ref{theorem:commutation}.

\begin{theorem}\label{theorem:commutation-concrete}
	Let $X,X'\in \dS_d$ be arbitrary, and $\nu$ corresponding to a principal or complementary series representation\footnote{For simplicity, we do not consider the discrete series here, although analogous results are expected to hold in that case as well.}. Then, in the sense of distributions,
	\begin{align}
		[\phi_R(u,X),\,\phi'_R(u',X')]=0\quad\text{for }\;u<u'\,,
	\end{align}
	on the space of vectors of finite Fock particle number.
\end{theorem}
\begin{proof}
	This theorem follows from Theorem~\ref{theorem:commutation} as a special case. However, we give an independent, explicit argument which illustrates how the properties of the integral operator $R$ enter. We focus on the principal series for definiteness. 
	
	We expand both $\phi_R,\phi'_R$ in terms of the Zamolodchikov-Faddeev creation/annihilation operators $z_R,z_R^\dagger$ respectively  their primed counterparts. The commutator $[\phi_R, \phi'_R]$ then gets contributions of the type $[z_R,z'_R], [z_R^\dagger, z_R^{\prime \dagger}]$, as well as $[z_R',z_R^\dagger], [z_R, z_R^{\prime \dagger}]$. It is relatively easy to see that the $[z_R,z_R']$ and $[z_R^\dagger, z_R^{\prime \dagger}]$ contributions vanish separately (see \eqref{eq:zz'comm}). This is not the case, however, for the remaining mixed contributions, where a non-trivial cancellation between both terms, called ``$+$'' and ``$-$'', is required. If we apply these contributions to an $n$-particle state $\Psi_n$, 
	we get a combination of two terms  abbreviated as
\ben\label{2K}
[\phi_R(u,X), \phi'_R(u',X')] \Psi_n = ( {}^+ K - {}^- K) \Psi_n \ . 
\een
Here, each ${}^\pm K$ acts as multiplication operator in $\bth = (\theta_1, \dots, \theta_n)$ and as integral kernel (depending on $X,X' \in \dS_d$ and $u,u' \in \RR$) 
on $(P_1, \dots, P_n)$. Taking into account the $R$-symmetry of the wave functions, the explicit form of those kernels is found to be:
\begin{align}
	{}^+ K_\bth(&P_1, \dots, P_n; P_1', \dots, P_n') 
	\\
	&= 
	\int_\RR \dd \lambda \,e^{2\lambda}e^{+i(u'-u)e^\la}
	\int_{\{Q_j\}} (X' \cdot Q_1)_+^{-\alpha -i\nu} (-X \cdot Q_{n+1})_-^{-\alpha+i\nu}
	\prod_{j=1}^{n} R_{\lambda-\theta_j}(Q_{j+1},P_j ;P_j',Q_j)
	\nonumber
\end{align}
for ``$+$'', whereas for ``$-$'' one has
\begin{align}
	{}^- K_\bth(&P_1, \dots, P_n; P_1', \dots, P_n') 
	\\
	&= 
	\int_\RR \dd \lambda \,e^{2\lambda}e^{-i(u'-u)e^\la}
	\int_{\{Q_j\}} (-X' \cdot Q_1)_-^{-\alpha +i\nu} (X \cdot Q_{n+1})_+^{-\alpha-i\nu}
	\prod_{j=1}^{n} R_{\theta_j-\la}(P_j,Q_j;Q_{j+1},P_j')
	\nonumber
	\,.
\end{align}
To get to the expression for ``$-$'', we have also used properties R1'') and  R3'', case 2) of the kernel $R_\theta$. A graphical expression for both kernels in the notation of Appendix~B is given in the following figure.
\begin{center}
	\includegraphics[width=\textwidth]{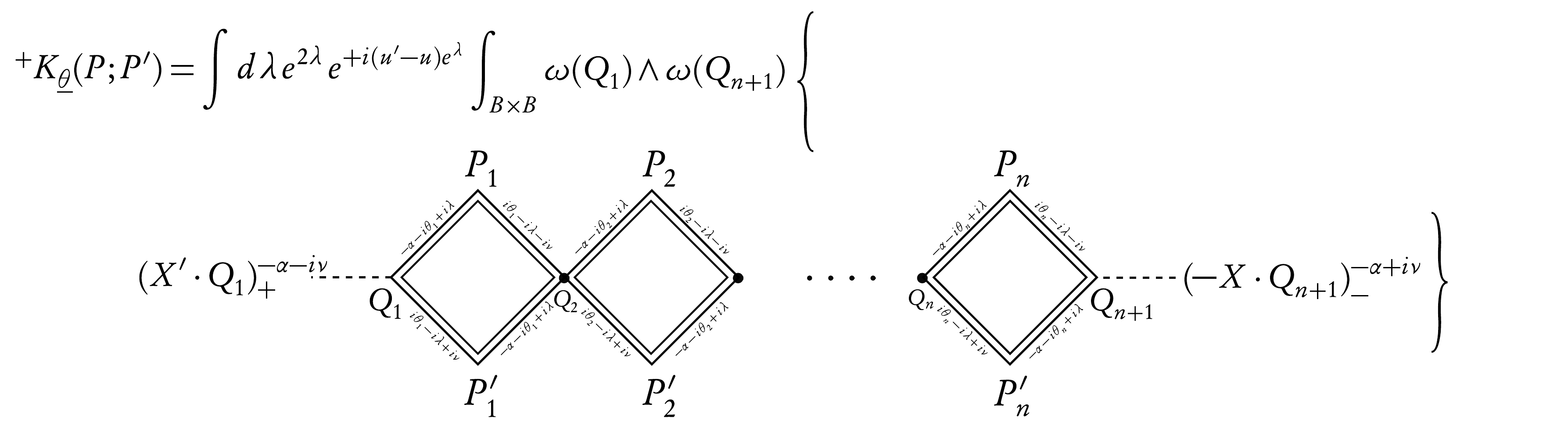}
	\\
	\includegraphics[width=\textwidth]{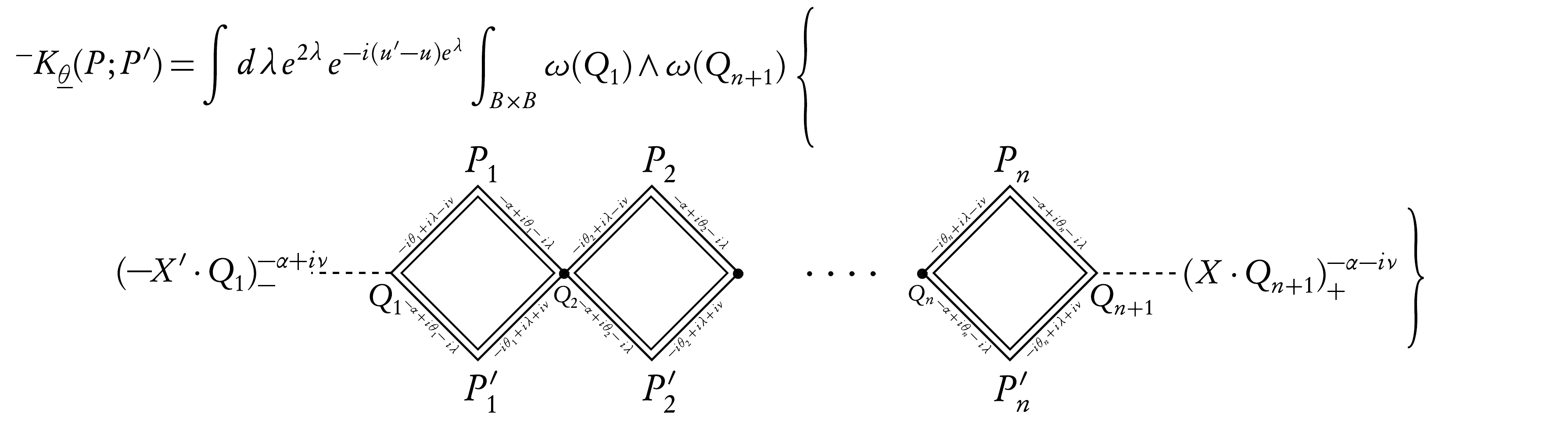}
\end{center}
Here the dots $\bullet Q_j$, $j=2,...,n$, indicate integrations over $B$ as explained in Appendix~$B$, and the two remaining integrals over $Q_1$ and $Q_{n+1}$ are written explicitly. The dashed lines connecting the $R$-kernels to the de Sitter waves simply indicate that these parts share the same de Sitter momentum ($Q_1$ and $Q_{n+1}$, respectively), and the double lined boxes mean the {\em full} integral kernel of $R_{\la-\te_j}$ respectively $R_{\te_j-\la}$, i.e. including all Gamma-factors, the crossing function $\sigma$ \eqref{eq:sigma-main}, and the rescaling $\te\to-\alpha\te/\pi$.

In order to see that the contribution from the ``$+$''-kernel cancels precisely that from the ``$-$'' kernel, we 
now shift the integration contour in ${}^+ K_\bth$ upwards to $\im(\lambda) = +\pi$. The shifted contour lies in the domain of analyticity of each 
$R_{\lambda-\theta_j}$, by (R6''). Furthermore, under $\lambda \to \lambda + i\mu, \mu \in (0,\pi)$, 
we have $i(u'-u) e^\lambda \to i(u'-u)e^\lambda(\cos \mu + i\sin \mu)$. By assumption $u'-u>0, \sin \mu > 0$, so the real part of this 
expression is negative and this provides an exponential damping of the integrand for large values of $\lambda$. The contour shift is thus permissible. 
For $\mu = \pi$, crossing symmetry (R7'') (see \eqref{eq:crossing-nu1nu2}) then implies that ${}^+K$ can also be expressed as
\begin{center}
	\includegraphics[width=\textwidth]{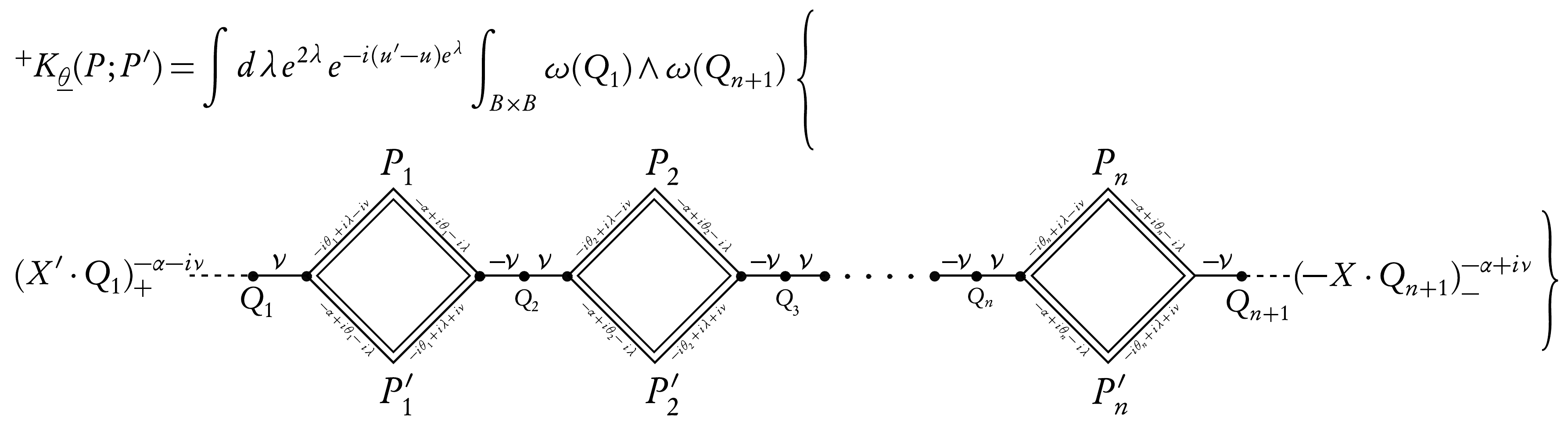}
\end{center}
Here the lines with $\pm\nu$ denote integral operators $I_{\pm\nu}$. Since $I_{-\nu}=I_\nu^{-1}$ (Lemma~\ref{lemma:TPC-principal}), the inner lines cancel. The two outer operators have the effect of switching the sign of $\nu$ in both the de Sitter waves. If we now flip the signs on $X,X'$ (taking into account that this changes the $\pm i\eps$ prescription on the de Sitter waves), it becomes apparent that the above expression coincides with ${}^-K$, so that the two kernels precisely cancel. 
This concludes the proof. 
\end{proof}

The truly local fields/observables of these model are different from the half-local fields $\phi_R$, $\phi_R'$, and can abstractly be described in an operator-algebraic setting. We therefore proceed from the pair of field operators $\phi_R$, $\phi_R'$ to the pair of von Neumann algebras $\M_R$, $\M_R'$, which by the preceding result are localized in half lines in the lightray coordinate. 

In order to build from these basic ``half line'' von Neumann algebras a net of von Neumann algebras $I \mapsto \A_R(I)$, indexed by intervals $I \subset \Rl$, one has to translate $\M_R, \M_R'$ and form intersections. One defines for the interval $I=(a,b)\subset\Rl$
\begin{align}\label{eq:chiral-net}
	\A_R(a,b)
	:=
	V(b,0,1)\M_R\,V(b,0,1)^{-1}\cap 	V(a,0,1)\M_R'\,V(a,0,1)^{-1}\,. 
\end{align}

Then by construction, we have
\begin{proposition}
	The assignment $I \mapsto \A_R(I)$ \eqref{eq:chiral-net} from open intervals to von Neumann algebras forms a $(\PG\times SO^\uparrow(d,1))$-covariant local net of von Neumann algebras on $\Hil^R$:
	\begin{enumerate}
		\item For any open interval $I\subset\Rl$, $(x,\la,\La)\in\PG_0\times SO^\uparrow(d,1)$,
		\begin{align}
			V(x,\la,\La)\A_R(I)V(x,\la,\La)^{-1}
			&=
			\A_R(e^{-\la}\,I+x)
			\,,\\
			\Theta\A_R(I)\Theta
			&=
			\A_R(-I)
			\,.
		\end{align}
		\item For two disjoint intervals $I_1,I_2\subset\Rl$, we have
		\begin{align}
			[\A_R(I_1),\,\A_R(I_2)]=\{0\}.
		\end{align}
	\end{enumerate}
\end{proposition}

The elements of $\A_R(I)$ may be understood as the local fields/observables of these models. We do not investigate them here, but just mention two important questions in this context: a) Under which conditions are the algebras $\A_R(I)$ ``large'' (for example in the sense that the Fock vacuum $\Om$ is cyclic)? And b): Under which conditions does the net $\A_R$ extend from the real line to the circle, transforming covariantly even under the 1-dimensional conformal group $PSL(2,\Rl)$, acting by fractional transformation $x\mapsto \frac{ax+b}{cx+d}$ with 
$ac-bd = \pm 1$? 

Following the same arguments as in the scalar case \cite{BostelmannLechnerMorsella:2011} (which build on \cite{GuidoLongoWiesbrock:1998}), one can show that i) the subspace $\bHil_{\rm loc}:=\overline{\A_R(a,b)\Om}\subset\bHil_R$ is independent of the considered interval $-\infty<a<b<\infty$, ii) on $\Hil_{\rm loc}$, the representation $V$ extends to PSL$(2,\Rl)\times SO^\uparrow(d,1)$, and iii) the net $I\mapsto\A_R(I)|_{\bHil_{\rm loc}}$ extends to a local conformally covariant net on the circle. A direct characterization of $\Hil_{\rm loc}$ is however difficult in general -- in particular because the nuclearity criteria \cite{BuchholzLechner:2004} that can be applied to the $O(N)$ models \cite{Alazzawi:2014} do not apply here because the representation $V_\nu$ is infinite-dimensional. We leave the analysis of these questions to a future investigation.

\subsection{Euclidean CFTs in $d-1$ dimensions}\label{EuclQFT}

Here we present a variant of our construction which leads to Euclidean conformal field theories in $(d-1)$-dimensions. 
As before, the construction is based on an $SO^\uparrow(d,1)$-invariant Yang-Baxter function $R$ such as \eqref{eq:R-scaled}.\footnote{
We do not rely on the crossing property in this section. Of course, to make a connection to the chiral models in the previous subsection (where crossing symmetry 
was essential), their $R$ functions should coincide, and therefore be crossing symmetric.}

To turn $R$ into a Yang-Baxter operator, we use here the amplification \eqref{eq:R-amp} discussed after Lemma~\ref{lemma:YBF->YBO} instead of the coupling to the representation space of the lightray. That is, we pick some $N\in\Nl$ and real numbers $\theta_1, \dots, \theta_N$. For given $\nu$ in the 
{\em complementary series} representation of $SO^\uparrow(d,1)$ (the conformal group in $(d-1)$-dimensional Euclidean space), we define the one particle space as
\ben
\H_1 = \Cl^N \otimes \K_\nu\,,
\een
and the invariant Yang-Baxter operator $(R\Psi)_{jl}=R_{\te_j-\te_l}\Psi_{lj}$ on $\Hil_1\ot\Hil_1\cong\K_\nu\tp{2}\ot\Cl^N\ot\Cl^N$, referring to an orthonormal basis $\{e_j\}_{j=1}^N$ of $\Cl^N$. The $R$-symmetric Fock space $\Hil^R$ is then defined as before. 

We now choose the orbital base $B\cong\Rl^{d-1}\cup\infty$ to be flat (see Appendix~A), which amounts to parameterizing $P\in C_d^+$ as 
\ben\label{Px}
P = (\half(|\bx|^2+1), \bx, \half(|\bx|^2-1)) \  
\een 
in terms of $\bx \in \RR^{d-1}$, and results in a one-particle space of the form $\Hil_1\cong\Cl^N\ot L^2(B)$. Vectors in this space are $N$-component functions $\bx\mapsto f_j(\bx)$, $j=1,...,N$, and the scalar product is (see \eqref{Idef} and Appendix A)
\begin{align}\label{eq:scpr}
	(f,g)
	=
	c_\nu \sum_{j=1}^N \int_B d^{d-1}\bx\int_B d^{d-1}\by\,\overline{f_j(\bx)}\,|\bx-\by|^{-2\Delta}\,g_j(\by)\,,
\end{align}
where $c_\nu = \pi^{-\alpha}2^{-i\nu} \Gamma(\alpha - i\nu)/\Gamma(i\nu)$ and $\Delta = \alpha - i\nu \in (0, \alpha)$ in the complementary series. 

We therefore have $N$ pairs of creation/annihilation operators $z^\#_{R,j}(f):=z^\#_R(e_j\ot f)$, and the $N$-component quantum fields
\begin{align}\label{phijx}
	\phi_{R,j}(\bx)
	:=
	\zd_{R,j}(\bx)+z_{R,j}(\bx)\,,\qquad  \;j=1, \dots ,N\,.
\end{align}
The $z_{R,j}^\#$ (and $\phi_{R,j}$) are operator-valued distributions, but now defined on the {\em Euclidean space $B$ instead of the ``momentum space'' $C_d^+$}. We we again describe these fields in terms of their distributional kernels $\phi_{R,j}(\bx)$ and $z_{R,j}^\#(\bx)$ at sharp points $\bx\in B$, i.e. 
$\phi_{R,j}(f) = \int \phi_{R,j}(\bx) f(\bx) d^{d-1} \bx$, etc.

As a consequence of the scalar product \eqref{eq:scpr}, the Faddeev-Zamolodchikov operators \eqref{eq:defzdR} satisfy the relations
\bena\label{zccr1}
z^\dagger_i(\bx_1) z_j(\bx_2) - R_{\theta_i-\theta_j} \ z_j(\bx_2) z^\dagger_i(\bx_1) &=& c_\nu \delta_{ij}\,|\bx_1 - \bx_2|^{-2\Delta} \\
z^\dagger_i(\bx_1) z^\dagger_j (\bx_2) - R_{\theta_i-\theta_j} \ z^\dagger_j(\bx_2) z^\dagger_i(\bx_1) &=& 0 \ . 
\eena
By construction, the field operators $\phi_{R,j}$ are real, $\phi_{R,j}(\bx)^*=\phi_{R,j}(\bx)$, and transform in the complementary series representation $V_\nu$, i.e. 
\ben
V_\nu(\Lambda) \phi_{R,j}(\bx) V_\nu(\Lambda)^{-1} = J_\Lambda(\bx)^{-\Delta} \ \phi_{R,j}(\Lambda \cdot \bx) \ , 
\een
where $\Lambda \cdot \bx$ is the usual action of conformal transformations $\Lambda \in SO^\uparrow(d,1)$ on $\RR^{d-1} \cup \infty$, and where $J_\Lambda$ is the 
corresponding conformal factor, see Appendix~A.  The scaling dimension of $\phi_{R,j}$
is therefore $\Delta = \alpha - i\nu \in (0, \alpha)$.

By the same arguments 
as before, the exponentiated fields $e^{i\phi_{R,j}(f)}$ are then well defined for any 
$f \in C_0^\infty(B)$ and any $j=1, \dots, N$. We define corresponding ``Euclidean'' von Neumann algebras
\ben
\E_R(O) = \{ e^{i\phi_{R,j}(f)} \mid f \in C^\infty_{\Rl,0}(O), \ \ \ \ j=1, \dots, N\}'' \,, 
\een
for any open, bounded region $O \subset \RR^{d-1}$. By construction, conformal transformations act geometrically on the net $O \mapsto \E_R(O)$ in the sense that $V_\nu(\Lambda) \E_R(O) V_\nu(\Lambda)^{-1} = \E_R(\Lambda \cdot O)$.
If we choose our discretized rapidities $\{\theta_j\}$ to be spaced equidistantly and formally take $N \to \infty$, then the shifts $\phi_{R,j} \mapsto \phi_{R,j+1}$ correspond to symmetries of the net $\E_R(O)$.

\bigskip

The simplest case of this construction is given when instead of the integral operators $R$, we take the flip $R=F$ on $\K_\nu\ot\K_\nu$. 
In that case, the $R_\theta$ factor drops out of the commutation relation for the $z_F,z_F^\dagger$'s, and one finds that the vacuum correlation functions of the field $\phi_{F,j}$ are of quasi-free form, i.e. 
\begin{align}
	(\Omega, \phi_{F,k_1}(\bx_1) \cdots \phi_{F,k_{2n+1}}(\bx_{2n+1}) \Omega)
	&=0\,,\\
	(\Omega, \phi_{F,k_1}(\bx_1) \cdots \phi_{F,k_{2n}}(\bx_{2n}) \Omega)
	&= 
	c \sum_P   \prod_{(i,j) \in P} \delta_{k_i k_j} |\bx_i - \bx_j|^{-2\Delta}  \ , 
\end{align}
where the sum is over all partitions of the set $\{ 1, ..., 2n\}$ into ordered pairs, $c$ is a real constant, and $\bx_i \neq \bx_j$ for all $i\neq j$ is assumed. These correlation functions correspond to an $N$-dimensional multiplet of a generalized Bosonic free Euclidean field theory.
The correlation functions are not reflection positive~\cite{NeebOlafsson:2014-2} (and so do not define a CFT in $(d-1)$-dimensional Minkowski spacetime
satisfying the usual axioms) apart from the limiting case $\Delta = \alpha$ corresponding to the standard free field. Locality of the field theory is expressed by the fact that the above correlation functions are symmetric under exchanges
$(\bx_i, k_i) \leftrightarrow (\bx_j, k_j)$. At the level of the von Neumann algebras, this amounts to saying that, if $O$ and $O'$ are disjoint, the corresponding 
von Neumann algebras commute
\ben\label{local}
[\E_F(O), \E_F(O')] = \{0\} \ . 
\een
In this sense, the Euclidean quantum field theory defined by the assignment $O \mapsto \E_F(O)$ is ``local''. 

This structure is modified if instead of the flip, we use our integral operators $R$. In that case, the fields $\phi_{R,j}$
are not ``local'' in the sense that the correlation functions are no longer symmetric, and consequently 
\eqref{local} does not hold. To obtain a local Euclidean theory, one could consider the algebras
\ben
\F_R(O) := \E_R(O) \cap \E_R(O')' \ , 
\een
where $O'$ denotes the complement of $O\subset B$, and where $\E_R(O')'$ is the commutant of the 
corresponding von Neumann algebra $\E_R(O')$. It follows directly from the definition that

\begin{proposition}
	The net $B\supset O \mapsto \F_R(O)$ is local and transforms covariantly under the conformal group $SO^\uparrow(d,1)$ in the sense that $V_\nu(\Lambda) \F_R(O) V_\nu(\Lambda)^{-1} = \F_R(\Lambda \cdot O)$.  
\end{proposition}

Local operators $\cO(\bx)$ in the conformal field theory defined by this new net should be thought of, roughly speaking, as elements in the intersection of $\F_R(O)$ for arbitrarily small $O$ containing $\bx$, i.e. in a sense
\ben
\cO(\bx) \in \bigcap_{O \supset \bx} \F_R(O) \qquad  \text{(formally).}
\een
Correlation functions $(\Omega, \cO_{a_1}(\bx_1) \cdots \cO_{a_n}(\bx_n) \Omega)$ of such fields would then again be local in the sense of being symmetric in the $(\bx_j, a_j)$. 

To make such statements precise, one should on the one hand make sure that the size of $\F_R(O)$ is sufficiently large, and one should also make precise what is meant by the above intersection, presumably
by making a construction along the lines of~\cite{Bostelmann:2005-2,FredenhagenHertel:1981}.

\section{Conclusion}\label{section:conclusion}

In this work we have constructed Yang-Baxter $R$-operators for certain unitary representations of $SO^\uparrow(d,1)$. The properties of these operators, in particular the Yang-Baxter equation, unitarity, and crossing symmetry make possible two, essentially canonical, constructions: a) A 1-dimensional ``light ray'' CFT, whose internal degrees of freedom 
transform under the given unitary representation and b) A Euclidean CFT in $(d-1)$ dimensions in which the group $SO^\uparrow(d,1)$ acts by conformal transformations. Both a) and 
b) are constructed from one and the same $R$-operator (in a complementary series representation). Theory b) depends on a discretization parameter $N$ corresponding to a 
set of $N$ discretized ``rapidities'', $\{ \theta_j\}$. The operator algebras in cases a) and b) become formally related when $N \to \infty$, but there is no evident -- even formal -- relation for finite $N$. 

The operator algebra a) is related to certain left- and right local fields on the lightray, \eqref{eq:phiR-concrete} and \eqref{eq:phi'R-concrete} whereas the operator algebra in 
case b) to certain fields of the form~\eqref{phijx}. They are built from certain generalized creation/annihilation operators $z^\dagger(\theta, P)$ in case a) and $z^\dagger_j(\bx)$ in 
case b). These operators satisfy a Zamolodchikov-Faddeev algebra into which the $R$-operators enter. $P$ is a de Sitter ``momentum'' which corresponds to $\bx$ as in Fig.~1 resp. eq.~\eqref{Px}. 
The index $j$ corresponds to the $j$-th discretized rapidity, $\theta_j$. The rapidity variable $\theta$ can be thought of--roughly speaking--as ``dual'' to the lightray variable, $u$ (in the sense of Fourier transform), whereas 
$\bx$ is dual to $X$ (de Sitter points) (in the sense of Fourier-Helgason transform~\eqref{eq:F-dS-pm}). 
Thus, when the spacing between $\{\theta_j\}$ goes to zero, the operator algebras formally coincide, and we think of this isomorphism as 
a model of a dS/CFT-type duality. 

In line with this interpretation, one is tempted to think of $N$ as a ``number of colors'' by analogy with the AdS/CFT correspondence. 
However, we note that the algebra in case b) does not have a corresponding symmetry such as $O(N)$ acting on the index $j$. What is restored in the limit as $N \to \infty$ is 
merely a symmetry corresponding to $\mathbb{Z}_N$ at finite $N$. Furthermore, unlike in the standard AdS/CFT correspondence, there is no correspondence at finite $N$. 
Another notable difference to most presentations of the AdS-CFT correspondence is that in our approach, both sides of the correspondence are automatically represented on the 
``same Hilbert space'', i.e. the unitary representation of the symmetry group $SO(d,1)$ is the same from the outset.

For a better understanding of the status and context of our proposed correspondence, it would undoubtedly be necessary to compare it with the quantization of strings in de Sitter space using more traditional approaches, e.g. via a version of the ``lightcone gauge''. Extensive investigations of this nature have been made in the case of the superstring in ${\rm AdS}_5 \times \mathbb{S}^5$, see e.g.~\cite{AF09} for a review. Although the starting point is rather different in those approaches, they also lead in the end to Yang-Baxter operators (describing the scattering of string excitations). So, a comparison ought to be possible at this level\footnote{Note that in~\cite{AF09}, the dispersion relation turns out to be that of a lattice model, rather than relativistic. In our approach, the Yang-Baxter function is a function of $\theta$, which is naturally treated as a rapidity parameter, and which is therefore 
related to a relativistic dispersion relation.}. 

Apart from the fact that we are interested here in de Sitter- rather than anti-de Sitter target spaces, the investigations~\cite{AF09} 
typically deal with a supersymmetric situation. In a supersymmetric setup, we should replace the 
de Sitter group $SO(d,1)$ -- or rather its Lie-algebra $\frak{so}(d,1)$ -- with a suitable super Lie algebra. This super Lie algebra should be a) {\em real}, b) have a bosonic part 
$\frak{so}(d,1) \oplus \frak{r}$, where the $R$-symmetry part $\frak{r}$ is {\em compact} (since there  are no unitary representations otherwise). The possibilities have been classified, see~\cite{dMH13}. For $d \ge 4$, there are only two, displayed in table~\ref{tab:tableNahm}. 

\begin{table}
\begin{center}
\begin{tabular}{|c|c|c|c|c|c|}
  \hline
  $d$ & Lie super algebra & $\frak{r}$ & de Sitter algebra & Spinor representation & Nahm label  \\ 
  \hline \hline
  $4$ & $\frak{osp}(2|1,1)$ & $\frak{u}(1,{\mathbb H})$ & $\frak{so}(4,1) \ \ ({\rm dS}_4)$ & $\HH \otimes_\RR \mathbb{C}^2$ & {\tt XII}$_1$ \\ 
  $6$ & $\frak{f}(4)^{\prime}$ & $\frak{u}(1,{\mathbb H})$ & $\frak{so}(6,1) \ \ ({\rm dS}_6)$ & $\HH \otimes_\RR 2\HH^2$ & {\tt IX}$_1$  \\ 
    \hline 
\end{tabular} \vspace*{.2cm}
\caption{Real Lie superalgebras with compact $R$-symmetry part whose ``spacetime part'' is the de Sitter algebra in $d \ge 4$ dimensions~\cite{dMH13}.
Here $\HH$ refers to the quaternions, so $\frak{u}(1, \HH) \cong \frak{su}(2)$.}
\label{tab:tableNahm}
\end{center}
\end{table}

It would be interesting to investigate a generalization of our methods to these cases. We must leave this to a future investigation.

\bigskip
\bigskip
\noindent
{\bf Acknowledgements:} The research (S.H.) leading to these results has received funding from the European Research Council under the European Union's Seventh Framework Programme (FP7/2007-2013) / ERC grant agreement no QC \& C 259562. S.H. also likes to thank R.~Kirschner from Leipzig for pointing out to him references~\cite{DerkachovKorchemskyManashov:2001,DerkachovManashov:2006,DerkachovManashov:2008,ChicherinDerkachovIsaev:2013} and for 
explanations regarding RLL-relations and related matters. 

\appendix

\section{Canonical choices for the orbital base $\boldsymbol{B}$}

It appears best to perform some of the calculations involving the form $\om$ \eqref{eq:omega} and the orbital base $B$ of the future lightcone $C_d^+$ of $\RR^{d+1}$ by using specific choices for $B$. It is known since the times of Kepler and Newton that there are 
three canonical choices, which correspond to a flat, hyperbolic, or spherical geometry for $B$. 

\begin{enumerate}
\item {\bf (Flat geometry, $\boldsymbol{B \cong \RR^{d-1} \cup \infty}$).} We realize $B$ as the intersection of $C_d^+$ with some arbitrary but fixed {\em null} plane in $\RR^{d+1}$. A parameterization of 
$B$ is in this case given by $\RR^{d-1} \owns \bx \mapsto P = (\half (|\bx|^2+1), \bx, \half(|\bx|^2 - 1)) \in C_d^+$. The induced geometry is seen to be flat. The point-pair invariant 
and $(d-1)$-form $\omega$ are given ($P,P' \in B$) in this case by
\ben
\omega = \dd^{d-1} \bx \ , \quad P \cdot P' = \frac{1}{2} \ |\bx - \bx'|^2 \ . 
\een 
The group of transformations leaving $B$ invariant is evidently $E(d-1)$, the Euclidean group. With the choice $B \cong \RR^{d-1} \cup \infty$, we may identify the representation space $\K_\nu$ with a space of square integrable functions on $B$. Under this identification, the action of $\Lambda \in SO^\uparrow(d,1)$ on a wave function $\psi(\bp)$ is given by 
\ben
(U_\nu(\Lambda) \psi)(\bx) =J_\Lambda(\bx)^{-\alpha-i\nu}  \psi(\Lambda \cdot \bx) \ , 
\een
where $\Lambda \cdot \bx$ denotes the usual action of a conformal group element on $\bx$, and where $J_\Lambda$ is the 
conformal factor of this transformation, $\Lambda^* |\dd \bx|^2 = J_\Lambda^2 |\dd \bx|^2$.

\item {\bf(Spherical geometry, $\boldsymbol{B \cong \SS^{d-1}}$).} We realize $B$ as the intersection of $C_d^+$ with some arbitrary but fixed {\em space like} plane in $\RR^{d+1}$. A parameterization of 
$B$ is in this case given by $\SS^{d-1} \owns \hat p \mapsto P = (1, \hat p) \in C_d^+$. The induced geometry is seen to be a round sphere. The point-pair invariant 
and $(d-1)$-form $\omega$ are given ($P,P' \in B$) in this case by
\ben
\omega = \dd^{d-1} \hat p \ , \quad P \cdot P' = 1-\hat p \cdot \hat p' \ , 
\een 
where we mean the standard integration element of the round sphere. 
The group of transformations leaving $B$ invariant is evidently $O(d-1)$, the rotational group. 

\item {\bf (Hyperbolic geometry, $\boldsymbol{B \cong \HH^{d-1} \times \{\pm 1\}}$).} We realize $B$ as the intersection of $C_d^+$ with some arbitrary but fixed pair of parallel {\em timelike} planes in $\RR^{d+1}$. A parameterization of the two disconnected components of $B$ is in this case given by $\RR^{d-1} \owns \bp \mapsto P = (\sqrt{|\bp|^2+1}, \bp, \pm 1) \in C_d^+$. 
The induced geometry is seen to be hyperbolic for each connected component corresponding to $\pm 1$, respectively. The point-pair invariant 
and $(d-1)$-form $\omega$ are given ($P,P' \in B$) in this case by
\ben
\omega = \frac{\dd^{d-1} \bp}{\sqrt{|\bp|^2+1}} \ , \quad P \cdot P' = \pm 1 + \sqrt{|\bp|^2+1} \sqrt{|\bp'|^2+1} - \bp \cdot \bp' \ , 
\een 
where the integration element is that of hyperbolic space. 
The group of transformations leaving $B$ invariant is evidently $SO(d-1,1)$.
\end{enumerate}

\begin{center}
	\includegraphics[width=140mm]{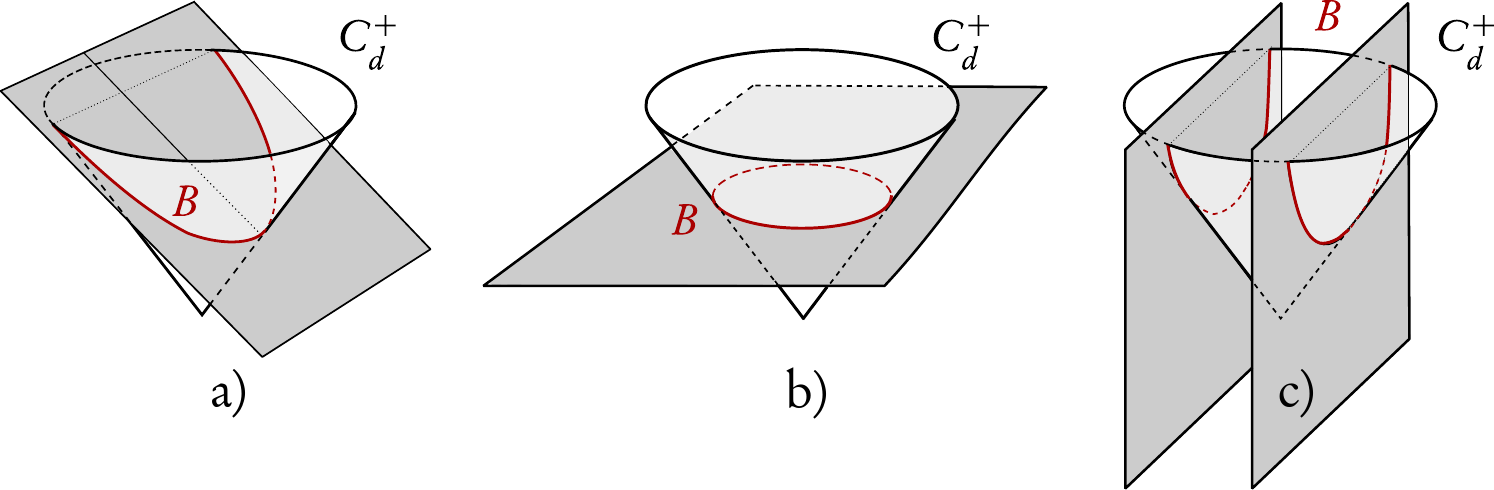}
	\\
	\small a) flat, b) spherical, and c) hyperbolic orbital base.
\end{center}
\newpage

\section{Proofs of Theorem~\ref{theorem:R-dS} and Theorem~\ref{theorem:crossing}}

The proofs of these theorems make use of the following identities: 

\begin{enumerate}

\item ({\bf Symanzik triality relation}):
\ben
\begin{split}\label{eq:triality}
& \int_{\SS^{d-1}} \dd^{d-1} \hat q \, (1-\hat p_1 \cdot \hat q)^{w_1} (1-\hat p_2 \cdot \hat q)^{w_2} (1-\hat p_3 \cdot \hat q)^{w_3} \\
&= (2\pi)^{\alpha} \frac{\Gamma(-w_1') \Gamma(-w_2') \Gamma(-w_3')}{\Gamma(-w_1) \Gamma(-w_2) \Gamma(-w_3)}  
(1-\hat p_2 \cdot \hat p_3)^{w_1'} (1-\hat p_1 \cdot \hat p_3)^{w_2'} (1-\hat p_1 \cdot \hat p_2)^{w_3'} \ , 
\end{split}
\een
for complex parameters $w_1,w_2,w_3\in\Cl$ satisfying $w_1+w_2+w_3= -(d-1),  w_i^{}, w_i'  \notin -\alpha -\Nl_{0}$, where the dual parameters are defined by 
$w_i' := -\alpha - w_i$. The integral on the left side is defined by analytic continuation 
in the parameters by the method described e.g. in~\cite{Hollands:2013}. The proof of the identity follows from formula (5.104) of~\cite{Hollands:2012}
(identical with formula (B22) of~\cite{MarolfMorrison:2011}), after a suitable analytical continuation in the parameters $z_i$ in the formula and an application of the residue theorem. 

The triality relation is best remembered in graphical form as a ``star-triangle relation'',
\begin{center}
	\begin{tikzpicture}
	\draw [thick] (0,0) -- (90:1) node[pos=.7,left]{\scriptsize$w_1$};
	\draw [thick] (0,0) -- (-150:1) node[pos=.7,below]{\scriptsize $w_2$};
	\draw [thick] (0,0) -- (-30:1) node[pos=.8,above]{\scriptsize$w_3$};
	\draw [fill=black] (0,0) circle [radius=0.06];
	\node (p1) at (90:1.5) {$P_1$};
	\node (p2) at (-150:1.5) {$P_2$};
	\node (p3) at (-30:1.5) {$P_3$};
	\node (q) at (-90:0.4) {$Q$};
	\node (=) at (0:2) {$=$};
\end{tikzpicture}
\begin{tikzpicture}
	\draw [thick] (90:1) -- (-30:1) node[pos=.7,above]{\scriptsize$\;\;w_2'$};
	\draw [thick] (-30:1) -- (-150:1) node[pos=.5,below]{\scriptsize $w_1'$};
	\draw [thick] (-150:1) -- (90:1) node[pos=.3,above]{\scriptsize$\!\!\!\!\!w_3'$};
	\node (p1) at (90:1.5) {$P_1$};
	\node (p2) at (-150:1.5) {$P_2$};
	\node (p3) at (-30:1.5) {$P_3$};
\end{tikzpicture}
\end{center}

Here a line with parameter $w$ between two ``momenta'' $P,Q$ denotes a ``propagator'' $\Gamma(-w)\cdot (P\cdot Q)^w$, a dot means integration over that variable wr.t. $(2\pi)^{-\alpha}\om(Q)$, and the product over all lines is understood. 

\item ({\bf Delta function relation}):
\ben\label{Delta}
{\rm anal. cont.}_{w \to 0} \left\{ \frac{\Gamma(\alpha+n -w)}{\Gamma(-n+w)} (1-\hat p_1 \cdot \hat p_2)^{-\alpha+w} \right\} = 
c_n (2\pi)^{\alpha} \, \Delta^n \delta(\hat p_1, \hat p_2) \ , 
\een
where we mean analytic continuation in the sense of distributions 
from the domain $\re(w) > 0$. For a proof and a mathematically precise 
explanation of this kind of analytic continuation, see e.g.~\cite{Hollands:2012,Hollands:2013}. The identity can be demonstrated by applying 
Laplacians to the composition relation below.

The following simple consequence of the triality relation and the delta function relation 
will be needed for the discrete series representations (where $n=0,1,2, \dots$ and $w \notin -\alpha-\mathbb{N}_0$):
\ben
\begin{split}\label{eq:triality1}
& \int_{\SS^{d-1}} \dd^{d-1} \hat q \, (1-\hat p_1 \cdot \hat q)^{n} (1-\hat p_2 \cdot \hat q)^{w} (1-\hat p_3 \cdot \hat q)^{-2\alpha-n-w} \\
&= c_n (2\pi)^{\alpha} \frac{\Gamma(\alpha+w) \Gamma(-\alpha-n-w)}{\Gamma(-w) \Gamma(2\alpha+n+w)}  
\Delta^n \delta(p_2,\hat p_3)(1-\hat p_1 \cdot \hat p_2)^{2\alpha+n} \ . 
\end{split}
\een
On the right side, we mean by $\delta$ the delta function on $\SS^{d-1}$ relative to the standard integration 
measure, and by $\Delta$ the Laplacian on $\SS^{d-1}$. The actual value of the constant $c_n$ is needed only for 
$n=0$, where it is $c_0=1$.

\item ({\bf Composition relation}):
This relation is obtained by taking the limit $w_3 \to 0$ in the triality relation and using the 
delta function relation. One obtains the integral identity  
\ben\label{eq:Composition-Relation}
\begin{split}
& \int_{\SS^{d-1}} \dd^{d-1} \hat q \, (1-\hat p_1 \cdot \hat q)^{ w_1} (1-\hat p_2 \cdot \hat q)^{w_2}  \\
&= (2\pi)^{2\alpha} \frac{\Gamma(\half (w_1-w_2)) \Gamma(\half(w_2-w_1))}{\Gamma(-w_1) \Gamma(-w_2)} \, \delta(\hat p_1, \hat p_2) \ , 
\end{split}
\een 
where it is assumed that $w_i \in \Cl, \sum w_i = -(d-1), w_1, w_2 \notin \Nl_{0}, \half (w_1 - w_2) \notin \ZZ$. The composition relation 
follows in the limit $w_3\to0$ from the triality relation.

\end{enumerate}

After these preparations, we turn to the proofs of Theorem~\ref{theorem:R-dS} and Theorem~\ref{theorem:crossing}.

\bigskip
\noindent{\bf Proof of Theorem~\ref{theorem:R-dS}.} 

\noindent
Some parts of the following proof are similar to the one in \cite{ChicherinDerkachovIsaev:2013} -- in particular, the verification of the Yang-Baxter equation (R4'') via a star-triangle (triality) relation.

We first restrict attention to the case when all $\nu_i$ correspond to principal or complementary series representations. 

\medskip

\noindent{\bf (R2'')} This invariance property holds by the definition of $R_\te$ in terms of Lorentz invariant inner products, and by making use of exactly the same line of argument as in eq.~\eqref{unitary}. 

\medskip

\noindent{\bf (R3'', part I)} The flip operator simply exchanges the variables, so that on the level of integral kernels, we have $(F^{\nu_2\nu_1}\,R_\theta^{\nu_1\nu_2}\,F^{\nu_2\nu_1})(P_1,P_2;P_1',P_2')=R_\theta^{\nu_1\nu_2}(P_2,P_1;P_2',P_1')$. By inspection of the kernel and the constant $c_{\nu_1\nu_2}(\te)=c_{\nu_2\nu_1}(\te)$ \eqref{eq:c}, one then sees that the first equation in (R3'') holds.

\medskip

\noindent{\bf (R1''), (R3'', part II), (R5''):} For unitarity (R1'') and the TCP symmetry (R3'', part II), we have to distinguish the principal and complementary series, because the scalar products and conjugations are different for the two series.

In case $\nu_1,\nu_2\in\Rl$ both belong to principal series representations, the scalar product is given by \eqref{inn}, and we therefore have
\begin{align*}
	(R^{\nu_1\nu_2}_\te)^*(P_1,P_2;Q_1,Q_2)
	&=
	\overline{R^{\nu_1\nu_2}_\te(Q_1,Q_2;P_1,P_2)}
	=
	R^{\nu_2\nu_1}_{-\te}(P_1,P_2;Q_1,Q_2)
	\,.
\end{align*}
Here the second step follows by direct inspection of the kernel, taking into account $\nu_1,\nu_2\in\Rl$. We thus have $(R^{\nu_1\nu_2}_\te)^*=R^{\nu_2\nu_1}_{-\te}$ in this case. 

In case both $\nu_1$ and $\nu_2$ belong to the complementary series, the scalar product is more complicated, but the conjugations $\Gamma_{\nu_1},\Gamma_{\nu_2}$ are simply complex conjugations (Lemma~\ref{lemma:TCP-complementary}), so that the TCP symmetry (R3'', part II) amounts to $\overline{R_\te^{\nu_1\nu_2}(P_1,P_2;Q_1,Q_2)}=R_{-\te}^{\nu_1\nu_2}(P_1,P_2;Q_1,Q_2)$. This equation holds true because $i\nu_1,i\nu_2$ are real for the complementary series.

The verification of $(R^{\nu_1\nu_2}_\te)^*=R^{\nu_2\nu_1}_{-\te}$ for the complementary series, and the TCP symmetry for the principal series require calculations. 

Let us consider the TCP symmetry for the principal series, $\nu_1,\nu_2\in\Rl$. Then the conjugations $\Gamma_{\nu_k}$ are given by complex conjugation and the integral operators $I_{\nu_k}$ \eqref{Idef1} (Lemma~\ref{lemma:TPC-principal}). Inserting the definitions and making use of the graphical notation introduced with the triality relation, one finds that the integral kernel of $((\Gamma_{\nu_2}\ot\Gamma_{\nu_1})R^{\nu_1\nu_2}_\te(\Gamma_{\nu_1}\ot\Gamma_{\nu_2}))(P_1,P_2;Q_1,Q_2)$ is given by
\begin{center}
\begin{tikzpicture}
	\draw [thick] (0,0) -- (2,0) node[pos=.5,below]{\scriptsize$-\alpha-i\nu_1$};
	\draw [thick] (2,0) -- (4,0) node[pos=.5,below]{\scriptsize$-\alpha-i\te+\nu_{12}^-$};
	\draw [thick] (4,0) -- (6,0) node[pos=.5,below]{\scriptsize$-\alpha+i\nu_{2}$};
	\draw [thick] (0,2) -- (2,2) node[pos=.5,above]{\scriptsize$-\alpha-i\nu_2$};
	\draw [thick] (2,2) -- (4,2) node[pos=.5,above]{\scriptsize$-\alpha-i\te-\nu_{12}^-$};
	\draw [thick] (4,2) -- (6,2) node[pos=.5,above]{\scriptsize$-\alpha+i\nu_1$};
	\draw [thick] (2,0) -- (2,2) node[pos=.85,yshift=+1.7ex, sloped,left]{\scriptsize$i\te+\nu_{12}^+$};
	\draw [thick] (4,0) -- (4,2) node[pos=.25,yshift=-1.7ex,sloped,right]{\scriptsize$i\te-\nu_{12}^+$};
	\draw [fill=black] (2,0) circle [radius=0.06];
	\draw [fill=black] (4,0) circle [radius=0.06];
	\draw [fill=black] (2,2) circle [radius=0.06];
	\draw [fill=black] (4,2) circle [radius=0.06];
	\node (p2) at (-.4,0) {$P_2$};
	\node (q2) at (6.4,0) {$Q_2$};
	\node (p1) at (-.4,2) {$P_1$};
	\node (q1) at (6.4,2) {$Q_1$};
\end{tikzpicture}
\end{center}

times $\left\{\Gamma(i\nu_1)\Gamma(i\nu_2)\Gamma(-i\nu_1)\Gamma(-i\nu_2)\Gamma(-i\te-\nu_{12}^+)\Gamma(-i\te+\nu_{12}^+)\Gamma(-i\te-\nu_{12}^-)\Gamma(-i\te+\nu_{12}^-)\right\}^{-1}$.

As a shorthand notation, we wrote $\nu_{kl}^\pm:=\frac{i}{2}(\nu_k\pm\nu_l)$ in the diagrams.

By repeated application of the triality relation, we convert the above diagram to
\begin{center}
\begin{tikzpicture}
	\draw [thick] (0,0) -- (2,0) node[pos=.5,below]{\scriptsize$-\alpha-i\nu_1$};
	\draw [thick] (2,0) -- (4,0) node[pos=.5,below]{\scriptsize$-\alpha-i\te+\nu_{12}^-$};
	\draw [thick] (4,0) -- (6,0) node[pos=.5,below]{\scriptsize$-\alpha+i\nu_{2}$};
	\draw [thick] (0,2) -- (4,2) node[pos=.5,above]{\scriptsize$-\alpha-i\te-\nu_{12}^+$};
	\draw [thick] (2,0) -- (4,2) node[pos=.5,sloped,above]{\scriptsize$i\nu_{2}$};
	\draw [thick] (4,2) -- (6,2) node[pos=.5,above]{\scriptsize$-\alpha+i\nu_1$};
	\draw [thick] (2,0) -- (0,2) node[pos=.45,sloped,above]{\scriptsize$i\te+\nu_{12}^-$};
	\draw [thick] (4,0) -- (4,2) node[pos=.25,yshift=-1.7ex,sloped,right]{\scriptsize$i\te-\nu_{12}^+$};
	\draw [fill=black] (2,0) circle [radius=0.06];
	\draw [fill=black] (4,0) circle [radius=0.06];
	\draw [fill=black] (4,2) circle [radius=0.06];
	\node (p2) at (-.4,0) {$P_2$};
	\node (q2) at (6.4,0) {$Q_2$};
	\node (p1) at (-.4,2) {$P_1$};
	\node (q1) at (6.4,2) {$Q_1$};
	\node (=) at (7,1) {$=$};
	\node (==) at (-1,1){$=$};
\end{tikzpicture}\begin{tikzpicture}
	\draw [thick] (0,0) -- (2,0) node[pos=.5,below]{\scriptsize$-\alpha-i\nu_1$};
	\draw [thick] (2,0) -- (6,0) node[pos=.5,below]{\scriptsize$-\alpha-i\te+\nu_{12}^+$};
	\draw [thick] (0,2) -- (4,2) node[pos=.5,above]{\scriptsize$-\alpha-i\te-\nu_{12}^+$};
	\draw [style=dashed,thick] (2,0) -- (4,2);
	\draw [thick] (4,2) -- (6,2) node[pos=.5,above]{\scriptsize$-\alpha+i\nu_1$};
	\draw [thick] (2,0) -- (0,2) node[pos=.45,sloped,above]{\scriptsize$i\te+\nu_{12}^-$};
	\draw [thick] (4,2) -- (6,0) node[pos=.5,yshift=.2ex,sloped,above]{\scriptsize$i\te-\nu_{12}^-$};
	\draw [fill=black] (2,0) circle [radius=0.06];
	\draw [fill=black] (4,2) circle [radius=0.06];
	\node (p2) at (-.4,0) {$P_2$};
	\node (q2) at (6.4,0) {$Q_2$};
	\node (p1) at (-.4,2) {$P_1$};
	\node (q1) at (6.4,2) {$Q_1$};
\end{tikzpicture}

\noindent\begin{tikzpicture}
	\draw [thick] (0,0) -- (6,0) node[pos=.5,below]{\scriptsize$-\alpha-i\te-\nu_{12}^-$};
	\draw [thick] (0,0) -- (0,2) node[pos=.5,sloped,above]{\scriptsize$i\te-\nu_{12}^+$};
	\draw [thick] (0,2) -- (4,2) node[pos=.5,sloped,above]{\scriptsize$-\alpha-i\te-\nu_{12}^+$};
	\draw [thick] (0,2) -- (6,0) node[pos=.5,sloped,below]{\scriptsize$i\nu_1$};
	\draw [thick] (4,2) -- (6,2) node[pos=.45,sloped,above]{\scriptsize$-\alpha+i\nu_1$};
	\draw [thick] (4,2) -- (6,0) node[pos=.25,yshift=1.7ex,sloped,right]{\scriptsize$i\te-\nu_{12}^-$};
	\draw [fill=black] (4,2) circle [radius=0.06];
	\node (p2) at (-.4,0) {$P_2$};
	\node (q2) at (6.4,0) {$Q_2$};
	\node (p1) at (-.4,2) {$P_1$};
	\node (q1) at (6.4,2) {$Q_1$};
	\node (=) at (7,1) {$=$};
	\node (==) at (-1,1){$=$};
\end{tikzpicture}\begin{tikzpicture}
	\draw [thick] (0,0) -- (6,0) node[pos=.5,below]{\scriptsize$-\alpha-i\te-\nu_{12}^-$};
	\draw [thick] (0,0) -- (0,2) node[pos=.5,sloped,above]{\scriptsize$i\te-\nu_{12}^+$};
	\draw [thick] (0,2) -- (6,2) node[pos=.5,sloped,above]{\scriptsize$-\alpha-i\te+\nu_{12}^+$};
	\draw [style=dashed,thick] (0,2) -- (6,0);
	\draw [thick] (6,2) -- (6,0) node[pos=.25,yshift=1.7ex,sloped,right]{\scriptsize$i\te+\nu_{12}^+$};
	\node (p2) at (-.4,0) {$P_2$};
	\node (q2) at (6.4,0) {$Q_2$};
	\node (p1) at (-.4,2) {$P_1$};
	\node (q1) at (6.4,2) {$Q_1$};
\end{tikzpicture}
\end{center}

Here the dotted lines result from two propagators with opposite powers that cancel each other. The first dotted line produces a factor of $\Gamma(i\nu_2)\Gamma(-i\nu_2)$, and the second dotted line produces a factor of $\Gamma(i\nu_1)\Gamma(-i\nu_1)$. Taking into account these factors (and the $\Gamma$-factor from the initial diagram), one then realizes that the last step represents the integral kernel of $R^{\nu_1\nu_2}_{-\te}$. This finishes the proof of (R3'', part II) for two principal series representations. For the mixed case (one principal series representation and one complementary series representation), the proof is similar. 

Returning to (R1'') for two complementary series representations, one finds that because of the appearance of the integral operators $I_\nu$ in the scalar product, the adjoint is given by
\begin{align}\label{eq:R-adjoint-cc}
	(R^{\nu_1\nu_2}_\te)^*(P_1,P_2;Q_1,Q_2)
	=
	\overline{((I_{\nu_2}\ot I_{\nu_1})R^{\nu_1\nu_2}_\te(I_{-\nu_1}\ot I_{-\nu_2}))(Q_1,Q_2;P_1,P_2)}
	\,.
\end{align}
Passing to the graphical notation, this kernel is given by the exact same diagram as in the TCP symmetry proof for the principal series, but with the replacements $\nu_1\to-\nu_1$, $\nu_2\to-\nu_2$, $P_1\to Q_1$, $P_2\to Q_2$. Converting the last diagram in the earlier calculation into an integral kernel then yields $R^{\nu_2\nu_1}_{-\te}(P_1,P_2;Q_1,Q_2)$. Thus, as in the case of two principal series representations, we have $(R^{\nu_1\nu_2}_\te)^*=R^{\nu_2\nu_1}_{-\te}$. Again, the proof for the mixed case is similar.

To finish the proof of (R1'') and (R5''), it now remains to show $R_\te^{\nu_1\nu_2}R_{-\te}^{\nu_2\nu_1}=1$. This follows (for all combinations of principal/complementary series representations) by application of the composition relation \eqref{eq:Composition-Relation}.

\medskip

\noindent{\bf(R4'')} The integral identity underlying the Yang-Baxter relation is the triality relation. One first calculates that the left and right hand sides of \eqref{YBE-with-nus} coincide on arbitrary $\psi\in\K_{\nu_1}\ot\K_{\nu_2}\ot\K_{\nu_3}$ if and only if the following two integral kernels (with the graphical notation introduced for the triality relation) coincide:
\begin{center}
	\includegraphics[width=150mm]{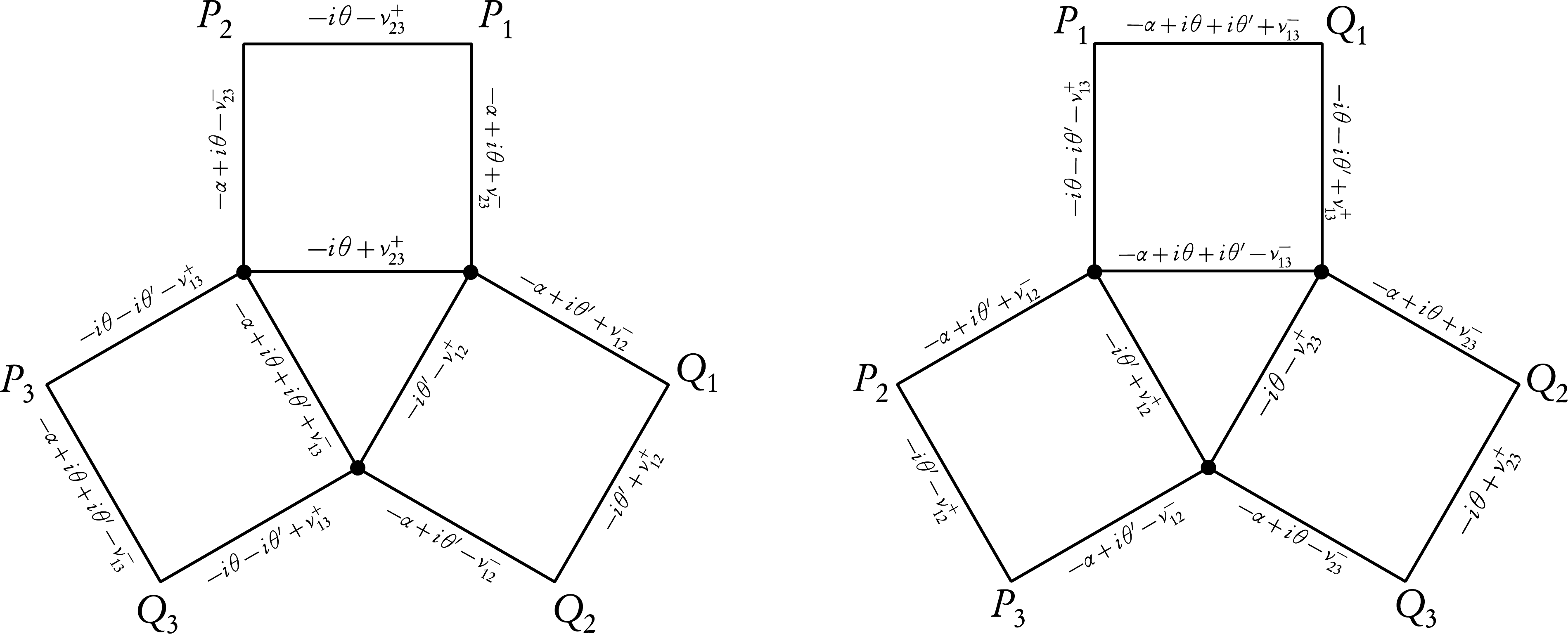}
\end{center}
As a shorthand notation, we wrote here again $\nu_{kl}^\pm:=\frac{i}{2}(\nu_k\pm\nu_l)$.

Note that these integral kernels differ from the ones arising from the Yang-Baxter equation by factors of $(2\pi)$ and Gamma functions of the parameters, but the overall factors are the same for both diagrams. 

We now use the triality relation to convert these diagrams into a more symmetrical form. Beginning with the left diagram, we first convert the interior triangle to a star, and then the three resulting stars into triangles. This shows that the left diagram above coincides with

\begin{center}
	\includegraphics[width=150mm]{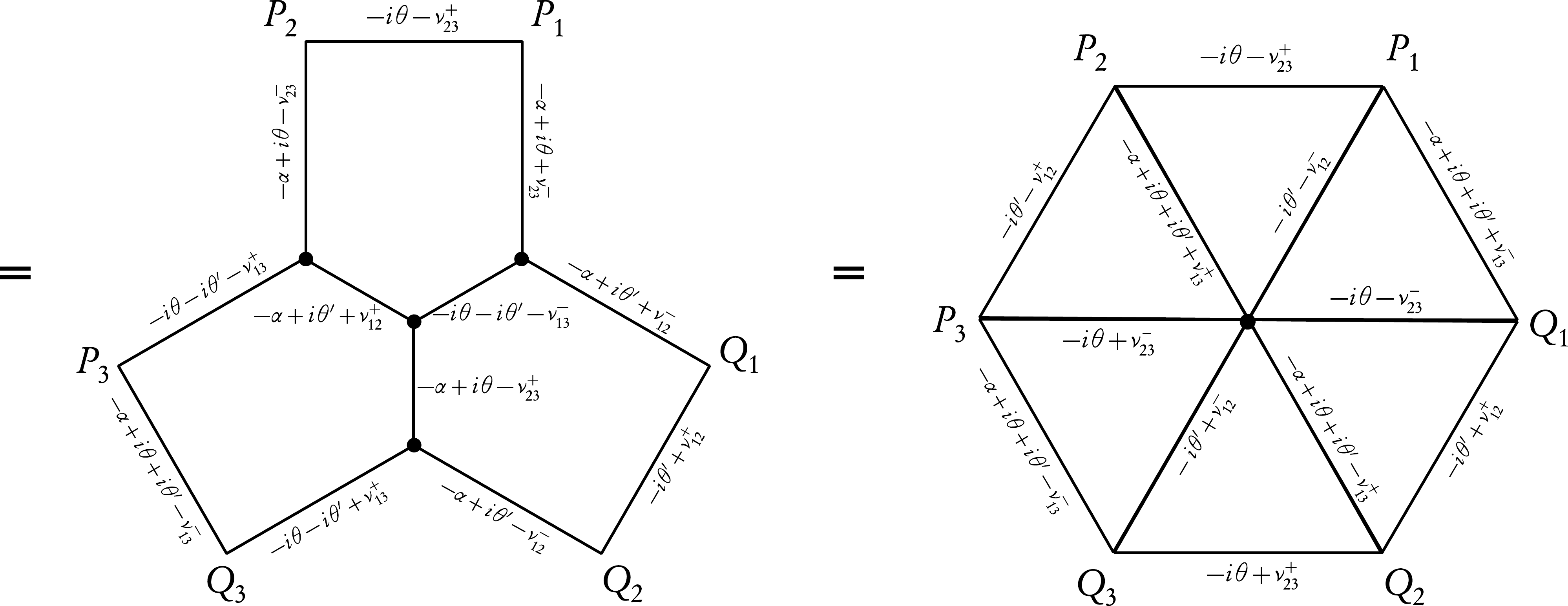}.
\end{center}

Analogous operations yield equality of the right diagram above with
\begin{center}
	\includegraphics[width=150mm]{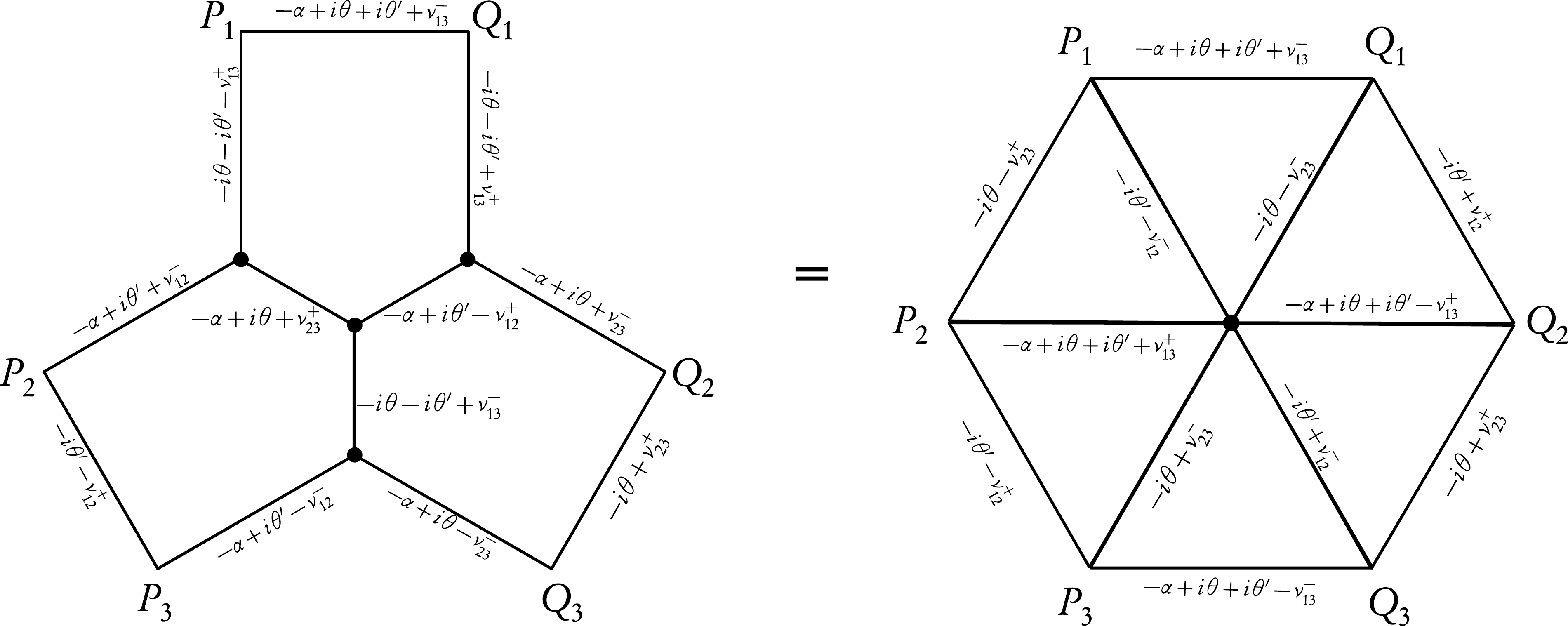}.
\end{center}
The two resulting hexagon diagrams are congruent and hence give identical integral kernels.

\medskip

Setting all $\nu$-parameters to one and the same value, the properties (R1'')--(R5'') imply the properties (R1')--(R5') of an invariant Yang-Baxter function for the representation $V_\nu$ and conjugation $\Gamma_\nu$. (The properties here are slightly stronger because of the division of (R3') into two separate equalities in (R3'').)

Finally, the normalization \eqref{eq:R-Normalization} follows by setting $\nu_1=\nu_2$ in the kernel \eqref{Rdef-kernel}, and taking the limit $\te\to0$ with the delta function relation \eqref{Delta}. 

\medskip

The case when one or more $\nu$-parameters correspond to a discrete series representation can be reduced to the previous cases by perturbing the 
corresponding $\nu$-parameters slightly from their discrete values along the real axis. Then the same arguments as given for the complementary series 
go through, and the desired identities are obtained in the limit where the relevant $\nu$-parameters go to their discrete values. 
One has to take care, however, to apply all identities to suitable wave functions, and, for the complementary series variables, use the constraint
\eqref{discr} before taking the limit. One also has to check that this constraint is actually preserved by the $R$-operator. This follows from \eqref{eq:triality1}.

\hfill$\square$

\bigskip
\noindent{\bf Proof of Theorem~\ref{theorem:crossing}.}

\noindent We will use the shorthand notation $\nu_{12}^\pm:=\frac{i}{2}(\nu_1\pm\nu_2)$ throughout the proof, and we first assume that 
both representations belong to the principal series, $\nu_1,\nu_2\in\Rl$. 

The integral kernel \eqref{Rdef-kernel} is entire analytic in $\te$ for non-coinciding momenta, and for $-\alpha<\im\te<0$ and principal series representations, all singularities are integrable. 
Thus the matrix elements of $R_\te$ are analytic functions on the strip $S_\alpha=\{\te\,:\,-\alpha<\im\te<0\}$, and moreover bounded in $\te$ on this domain.

We now explain the reason for the particular form of the factor $\sigma_{\nu_1\nu_2}$ \eqref{eq:sigma-main}. Inserting \eqref{eq:R-scaled}, we see that (R7'') is equivalent to 
\begin{align}\label{eq:crossing-B}
	\sigma_{\nu_1\nu_2}(-i\alpha-\te)\cdot \langle\xi_2\ot\psi_1,R^{\nu_1\nu_2}_{-i\alpha-\te}\,(\psi_1' \ot\xi'_2)\rangle
	=
	\sigma_{\nu_2\nu_1}(\te)\cdot		\langle\psi_1\ot\Gamma_{\nu_2}\xi_2',\,R_{\te}^{\nu_2\nu_1}(\Gamma_{\nu_2}\xi_2\ot\psi_1')\rangle
	\,.
\end{align}
Since $\nu_1,\nu_2\in\Rl$ belong to principal series representations, the conjugations are $\Gamma_{\nu_k}=CI_{\nu_k}=I_{-\nu_k}C$, where $C$ denotes pointwise complex conjugation. Using this and Lemma~\ref{lemma:TPC-principal}~$a)$, one checks that \eqref{eq:crossing-B} amounts on the level of integral kernels to, $\te\in\Rl$, $P_1,P_2,Q_1,Q_2\in C_d^+$,
\begin{align}\label{eq:crossing-nu1nu2}
	\frac{\sigma_{\nu_1\nu_2}(-i\alpha-\te)}{\sigma_{\nu_2\nu_1}(\te)}
	\cdot R^{\nu_1\nu_2}_{-i\alpha-\te}(P_1,P_2;Q_1,Q_2)
	=
	\left((1\ot I_{\nu_2})R_{\te}^{\nu_2\nu_1}(I_{-\nu_2}\ot1)\right)(P_2,Q_2;P_1,Q_1)
	\,.
\end{align}
In the graphical notation, the right hand side can be transformed with the triality relation into
\begin{center}
	\includegraphics[width=\textwidth]{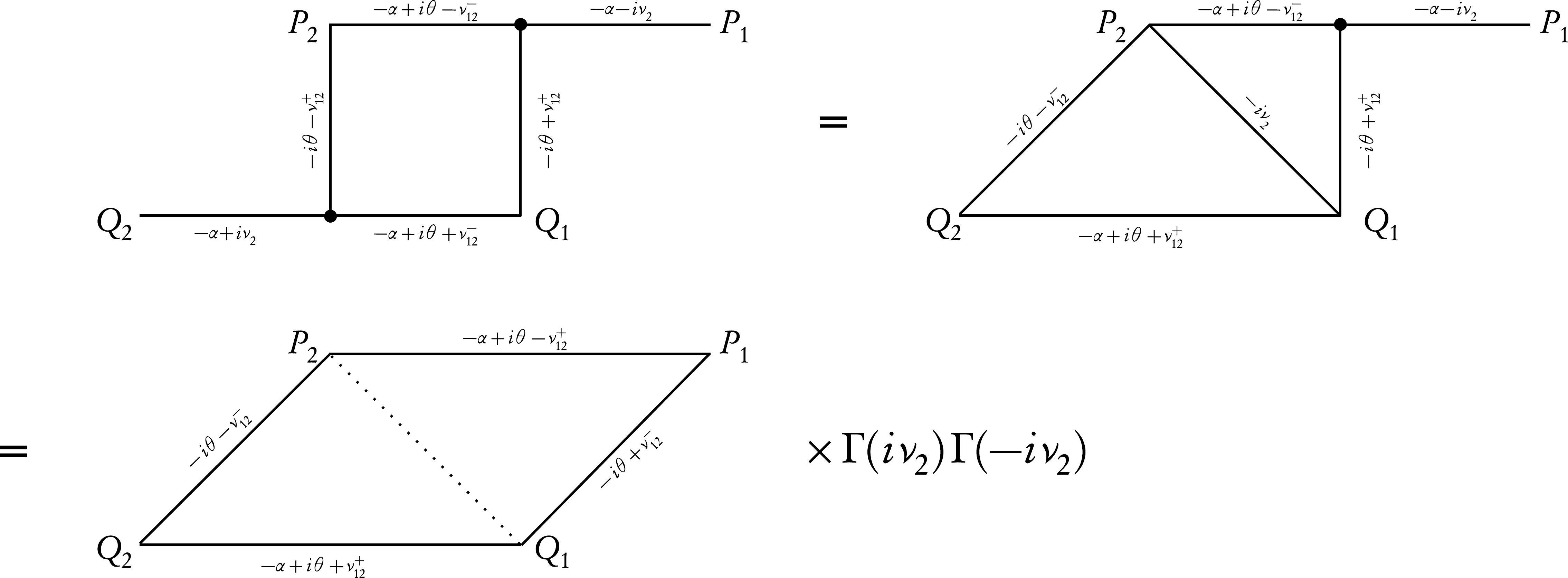}
\end{center}
A factor of $\left\{\Gamma(i\nu_2)\Gamma(-i\nu_2)\Gamma(i\te-\nu_{12}^-)\Gamma(i\te+\nu_{12}^-)\Gamma(i\te+\nu_{12}^+)\Gamma(i\te-\nu_{12}^+)\right\}^{-1}$ has been suppressed in all these diagrams. Comparing with the analytically continued matrix elements on the left hand side of \eqref{eq:crossing-nu1nu2}, one then finds that \eqref{eq:crossing-nu1nu2} holds if
\begin{align}\label{eq:func-eq}
	\frac{\sigma_{\nu_1\nu_2}(-i\alpha-\te)}{\sigma_{\nu_2\nu_1}(\te)}
	&=
	\frac{\Gamma(\alpha-i\te-\nu_{12}^-)\Gamma(\alpha-i\te+\nu_{12}^-)\Gamma(\alpha-i\te-\nu_{12}^+)\Gamma(\alpha-i\te+\nu_{12}^+)}{\Gamma(i\te-\nu_{12}^-)\Gamma(i\te+\nu_{12}^-)\Gamma(i\te-\nu_{12}^+)\Gamma(i\te+\nu_{12}^+)}
	\,.
\end{align}
In order not to spoil the analyticity of the matrix elements of $R^{\nu_1\nu_2}$, we have to choose $\sigma_{\nu_1\nu_2}$ analytic on the strip $S_\alpha$. Furthermore, $\sigma_{\nu_1\nu_2}$ must satisfy the requirements of Proposition~\ref{proposition:tweak-R}, i.e. it must be symmetric in $\nu_1,\nu_2$, and $\overline{\sigma_{\nu_1\nu_2}(\te)}=\sigma_{\nu_1\nu_2}(\te)^{-1}=\sigma_{\nu_1\nu_2}(-\te)$ for $\te\in\Rl$.

We claim that all requirements are satisfied by 
\begin{align}\label{eq:sigma}
	\sigma_{\nu_1\nu_2}(\te)
	=
	\frac{\Gamma(c-i\te)^2}{\Gamma(c+i\te)^2}\cdot e^{if_{\nu_1 \nu_2}(\te)}\,,
	\qquad
	f_{\nu_1\nu_2}(\te):=
	\frac{1}{2\pi}\int_0^\infty dp\,\frac{\sin(\te p)}{p}\,\frac{\ghat_{\nu_1\nu_2}(p)}{\cosh\frac{\alpha p}{2}}\,,
\end{align}
where $c>0$ is a real sufficiently large parameter that will be chosen later, and  
\begin{align}\label{eq:ghat}
	\ghat_{\nu_1\nu_2}(p) 
	&:=
	i\int_\Rl dt\,e^{ipt}\,
	\frac{d}{dt}\log
	g_{\nu_1\nu_2}(t-\tfrac{i\alpha}{2})\,,\\
	g_{\nu_1\nu_2}(t-\tfrac{i\alpha}{2})
	&:=
	\frac{\Gamma(\alpha/2-it-\nu_{12}^-)\Gamma(\alpha/2-it+\nu_{12}^-)\Gamma(\alpha/2-it-\nu_{12}^+)\Gamma(\alpha/2-it+\nu_{12}^+)}{\Gamma(\alpha/2+ it-\nu_{12}^-)\Gamma(\alpha/2 + it+\nu_{12}^-)\Gamma(\alpha/2+ it-\nu_{12}^+)\Gamma(\alpha/2 + it+\nu_{12}^+)} \\
	& \hspace{.8cm} 
	\cdot \frac{\Gamma(\alpha/2 + c+it)^2\Gamma(-\alpha/2 + c +it)^2}{\Gamma(\alpha/2 + c-it)^2\Gamma(-\alpha/2+c-it)^2}
	\,.
	\nonumber
\end{align}
To verify this claim, we need to examine the functions $g_{\nu_1\nu_2}$ and $\ghat_{\nu_1\nu_2}$.
One first checks the poles of the Gamma functions and sees that $g_{\nu_1\nu_2}$ is analytic in the strip $S_\alpha$ if $c>0$ is sufficiently large, for instance if we 
take $c=\alpha$. Its logarithmic derivative can trivially be expressed in terms of the Digamma function $\psi=\Gamma'/\Gamma$ as
\begin{align}\label{Gdef}
	& G_{\nu_1\nu_2}(t-\tfrac{i\alpha}{2})
	:=
	i\frac{d}{dt}\log
	g_{\nu_1\nu_2}(t-\tfrac{i\alpha}{2}) \\
	=&
	+\psi(\alpha/2-it-\nu_{12}^-)+\psi(\alpha/2-it+\nu_{12}^-)+\psi(\alpha/2-it-\nu_{12}^+)+\psi(\alpha/2-it+\nu_{12}^+)+
	\nonumber\\
	&
	+\psi(\alpha/2+it-\nu_{12}^-)+\psi(\alpha/2+it+\nu_{12}^-)+\psi(\alpha/2+it-\nu_{12}^+)+\psi(\alpha/2+ it+\nu_{12}^+)-
	\nonumber
	\\
	&
	-2\psi(c+\alpha/2+it)-2\psi(c-\alpha/2+it)-2\psi(c+\alpha/2-it)-2\psi(c-\alpha/2-it) \nonumber
	\,.
\end{align}
Clearly, this function $G_{\nu_1 \nu_2}(\theta)$ is analytic in $S_\alpha$ as well. By taking into account the asymptotic expansion $\psi(z)\sim\log z-\frac{1}{2z}+{\cal O}(\frac{1}{z^2})$ as $z\to\infty$ in $|\arg(z)|<\pi$, and going through all terms, one also finds that $|G_{\nu_1\nu_2}(t)|$ vanishes quadratically in $\re(t)$ in the strip~$S_\alpha$ for $|t| \to \infty$. Thus $\ghat_{\nu_1\nu_2}$ \eqref{eq:ghat} is well-defined.

Furthermore, we have the symmetry properties $G_{\nu_1\nu_2}(t)=G_{\nu_1\nu_2}(-i\alpha-t)$ and $\overline{G_{\nu_1\nu_2}(-\bar{t})}=G_{\nu_1\nu_2}(t)$, $t\in S_\alpha$, which imply that $\ghat_{\nu_1\nu_2}$ \eqref{eq:ghat} is even and real (for real arguments $p$). Since $\ghat_{\nu_1\nu_2}$ also decays fast because it is the Fourier transform of a smooth function, we see that $f_{\nu_1\nu_2}$ \eqref{eq:sigma} is well-defined, odd, and real (for real $\te$). It then follows that $\sigma_{\nu_1\nu_2}$ \eqref{eq:sigma} satisfies the requirements of Prop.~\ref{proposition:tweak-R}.

It remains to check that $\sigma_{\nu_1\nu_2}$ is bounded and analytic in $S_\alpha$, and that \eqref{eq:func-eq} holds. Regarding analyticity, the integrand of $f_{\nu_1\nu_2}$ is entire in $\te$, and we may estimate the growth of the sine function by $e^{|\im(\te)|p}$. This growing factor is compensated by the falloff of $|g_{\nu_1\nu_2}(p)/\cosh\tfrac{\alpha p}{2}|$. Clearly $|1/\cosh\frac{\alpha p}{2}|\leq 2\,e^{-\alpha p/2}$, and furthermore $|\ghat_{\nu_1\nu_2}(p)|$ decays like $e^{-\alpha p/2}$ as well -- this latter fact follows from a contour shift in the Fourier integral \eqref{eq:ghat}. Together with the remaining decay of the Digamma functions, this establishes the analyticity of $f_{\nu_1\nu_2}$ in the strip~$S_\alpha$. By analogous arguments, one also shows that $e^{if_{\nu_1\nu_2}(\te)}$ is bounded in the strip.

\medskip

To verify \eqref{eq:func-eq}, we compute the Fourier transform $\fti_{\nu_1\nu_2}$ of $f_{\nu_1\nu_2}$. We will use that since $\ghat_{\nu_1\nu_2}$ is even, we have $f_{\nu_1\nu_2}(\te)=\frac{1}{4\pi i}\int_\Rl dp\,\frac{e^{i\te p}}{p}\frac{\ghat_{\nu_1\nu_2}(p)}{\cosh\frac{\alpha p}{2}}$, and we will also make use of $\ghat_{\nu_1\nu_2}(p)=\sqrt{2\pi}\,e^{-\alpha p/2}\tilde{G}_{\nu_1\nu_2}(p)$. This gives
\begin{align*}
	(e^{\alpha q}+1)\,\fti_{\nu_1\nu_2}(q)
	&=
	\frac{\tilde{G}_{\nu_1\nu_2}(q)}{iq}
	\,,
\end{align*}
and after an inverse Fourier transformation, we arrive at
\begin{align*}
	f_{\nu_1\nu_2}(-i\alpha-\te)-f_{\nu_1\nu_2}(\te)
	&=
	-i\,\log g_{\nu_1\nu_2}(\te)
	\Longrightarrow
	\frac{e^{if_{\nu_1\nu_2}(-i\alpha-\te)}}{e^{if_{\nu_1\nu_2}(\te)}}
	=
	g_{\nu_1\nu_2}(\te)
	\,.
\end{align*}
Using the definitions of $\sigma_{\nu_1\nu_2}$ and $g_{\nu_1\nu_2}$, the desired equality \eqref{eq:func-eq} then follows. 

\medskip

We must also check the analyticity and boundedness of the $\Gamma$ factors in the definition of $R^{\nu_1\nu_2}$ and $\sigma_{\nu_1\nu_2}$. These follow from the well-known facts that $\Gamma$ is non-vanishing, 
has poles at the non-positive integers, and the standard asymptotic formula ($|y| \to \infty$)
\ben
|\Gamma(x+iy)| \sim (2\pi)^{\half} \ |y|^{x-\half} e^{-\pi|y|/2} \ .  
\een
It follows that the $\Gamma$ factors are bounded by $O(|\theta|^{(d-1)})$ for large $|\theta|$. 
 The $\theta$-dependence of coming from the exponentials is analytic $\theta$ once we 
form the matrix elements of $R_\theta$, and these exponentials are also clearly bounded in $\theta$. If we take matrix 
elements with smooth wave functions, we get from these factors decay as $|\theta|^{-n}$ where $n$ is as large as we wish.
Thus, all pieces in $R_\theta$ are analytic  in the strip $S_\alpha$ and decay faster than any inverse power 
$|\theta|^{-n}$ for $|\theta | \to \infty$ if we take matrix elements with smooth wave functions. 

\medskip

To derive the infinite product formula for $\sigma_{\nu_1 \nu_2}(\theta)$ quoted in \eqref{eq:sigma-main} is rather lengthy, and we only sketch the main steps. First, we expand the Digamma functions in 
the definition of $G_{\nu_1 \nu_2}$ using the well-known series
\ben
\psi(z) = -\gamma_E + \sum_{n=0}^\infty \left( \frac{1}{n+1} - \frac{1}{n+z} \right) \ . 
\een
Substituting this series for each of the terms in the expression~\eqref{Gdef} for $G_{\nu_1 \nu_2}$, we find that all contributions from Euler's constant $\gamma_E$ and from the sums over $1/(n+1)$ cancel each other. We next calculate $\hat g_{\nu_1 \nu_2}(p)$ 
by performing the integral \eqref{eq:ghat} over $t$ separately for each term in the series (this is admissible, because both the series and the integral are absolutely convergent). The resulting integrals all have the form
\ben
\int_{-\infty}^\infty dt \ \frac{(\beta+n) e^{itp}}{(\beta+n -it)(\beta+n+it)} = -\pi \ e^{-(n+\beta)|p|} \ , 
\een
where the residue theorem was used, and where $\beta$ stands for the various constants that appear. The sum over $n$ can then be easily done with the aid of a geometric series, resulting in the expression
\ben\label{fint}
f_{\nu_1 \nu_2}(\theta) = 4 \int_0^\infty dp \ \frac{\sin(\theta p)}{p(1+e^{\alpha p})(1-e^{-p})} [\cosh(\nu_{12}^+ p) + \cosh(\nu_{12}^-p) - e^{-cp}(1+e^{\alpha p})] \ . 
\een 
In order to perform this integral, we expand out the factors $(1+e^{\alpha p})^{-1}, (1-e^{-p})^{-1}$ using a geometric series, resulting altogether in a double series indexed by natural numbers $n,m$. 
The integral can be pulled inside this double series and can then be performed fairly easily for each term. Each such term turns out to be a logarithm, so the double series of these logarithms becomes a 
logarithm of a doubly infinite product. The end result can be written as
\begin{align}
f_{\nu_1 \nu_2}(\theta) &= \frac{1}{i} \log \prod_{n,m=0}^\infty  \bigg\{ \prod_{p,q=\pm} \frac{(-i\theta + \alpha(2n+1) + m - p \nu_{12}^q)(i\theta + \alpha(2n+2) + m - p \nu_{12}^q)}{(i\theta + \alpha(2n+1) + m - p \nu_{12}^q)(-i\theta + \alpha(2n+2) + m - p \nu_{12}^q)}   \nonumber \\
& \qquad\qquad\qquad \cdot \frac{(-i \theta + \alpha(2n+2)+m-c)(i\theta + \alpha(2n) + m -c)}{(i\theta + \alpha(2n+2) +m-c)(-i\theta + \alpha(2n) + m -c)} \bigg\} \ . 
\end{align}
The product over $m$ can be performed with the aid of the well-known infinite product 
\ben
\Gamma(z+1) = \prod_{m=1}^\infty \left( 1+ \frac{z}{m} \right)^{-1} e^{z/m} \ , 
\een
and this results in the formula \eqref{eq:sigma-main} for $\sigma_{\nu_1 \nu_2}$ quoted in the main text after choosing for $c$ the value $c=\alpha$.

In case we consider two coinciding complementary series representations $\nu_i=\nu$, we have to take into account the different conjugation and different scalar product \eqref{Idef}. One then finds that the same analytic properties, and in particular the same functional equation \eqref{eq:func-eq} 
are required for the factor $\sigma_{\nu\nu}$. 
The solution is given by the same infinite product formula quoted in the main text \eqref{eq:sigma-main}, but we now need to choose $c=\alpha+i\nu$. This guarantees in particular 
absolute convergence of the infinite product, as one may see using standard asymptotic expansions of the Gamma function. 
\hfill$\square$


\end{document}